\documentclass[11pt]{article}

\usepackage{charter}
\usepackage[charter]{mathdesign}
\usepackage{eucal}
\usepackage{microtype}

\usepackage[margin=1in,letterpaper]{geometry}

\usepackage{enumitem}

\usepackage{amsmath}
\usepackage{amsthm}
\newtheorem{thm}{Theorem}
\newtheorem{lem}[thm]{Lemma}
\newtheorem{cor}[thm]{Corollary}
\newtheorem{prop}[thm]{Proposition}
\newtheorem{obs}[thm]{Observation}
\newtheorem{clm}{Claim}
\theoremstyle{definition}
\newtheorem{defn}{Definition}
\theoremstyle{remark}
\newtheorem*{rem}{Remark}

\usepackage[ruled,vlined,linesnumbered]{algorithm2e}
\DontPrintSemicolon
\SetFuncSty{textsc}
\SetProcNameSty{textsc}
\SetArgSty{textnormal}
\newcommand{\procref}[1]{\ProcNameSty{\ref{#1}}}
 
\usepackage{ifthen,tikz}
\usetikzlibrary{positioning}
\usepackage{subcaption}

\usepackage{todonotes}

\usepackage[sort]{cite}

\newcommand{\T}{\mathcal{T}}
\newcommand{\stc}{s_{\mathrm{tc}}}
\newcommand{\htc}{h_{\mathrm{tc}}}

\newcommand{\None}{\textsc{None}}
\newcommand{\tree}[2][]{T\ifthenelse{\equal{#1}{}}{}{_{#1}}^{(#2)}}
\newcommand{\TreeChildSequence}{\FuncSty{TCS}}
\newcommand{\TreeChildSequencePruned}{\FuncSty{TCS2}}

\newcommand{\UpdateR}{\FuncSty{UpdateR}}
\newcommand{\cc}[3]{\textit{count}(#1,#2,#3)}

\newcommand{\markj}[1]{#1}
\newcommand{\yuki}[1]{#1}
\newcommand{\leo}[1]{#1}

\newcommand{\reduce}[2]{#1/#2}

\begin{document}

\title{A Practical Fixed-Parameter Algorithm for Constructing Tree-Child
  Networks from Multiple Binary Trees\thanks{Leo van Iersel, Remie Janssen, Mark Jones and Yukihiro Murakami were supported by the Netherlands Organization for Scientific Research (NWO), including Vidi grant 639.072.602, and van Iersel also by the 4TU Applied Mathematics Institute. Norbert Zeh was supported by the Natural Sciences and Engineering Research Council of Canada.}}
\author{Leo van Iersel\footnote{Delft Institute of Applied Mathematics, Delft University of Technology, Van Mourik Broekmanweg 6, 2628 XE, Delft, The Netherlands, \texttt{\{L.J.J.vanIersel,R.Janssen-2,M.E.L.Jones,Y.Murakami\}@tudelft.nl}.}\and Remie Janssen\footnotemark[2]\and Mark Jones\footnotemark[2]\and Yukihiro
  Murakami\footnotemark[2]\and Norbert Zeh\footnote{Faculty of Computer Science, Dalhousie University, 6050 University Ave, Halifax, NS B3H 1W5, Canada, \texttt{nzeh@cs.dal.ca}.}}
  
\maketitle

\begin{abstract}\noindent
  We present the first fixed-parameter algorithm for constructing a tree-child
  phylogenetic network that displays an arbitrary number of binary input trees
  and has the minimum number of reticulations among all such networks.
  The algorithm uses the recently introduced framework of cherry picking
  sequences and runs in $O((8k)^k \mathrm{poly}(n, t))$ time, where $n$ is the
  number of leaves of every tree, $t$ is the number of trees, and $k$ is the
  reticulation number of the constructed network.
  Moreover, we provide an efficient parallel implementation of the algorithm
  and show that it can deal with up to~$100$ input trees on a standard desktop
  computer, thereby providing a major improvement over previous phylogenetic
  network construction methods.
\end{abstract}

\section{Introduction}

Evolutionary histories are usually described by phylogenetic trees or networks.
A phylogenetic tree describes how a collection of studied taxa (e.g., species,
strains or languages) have evolved over time by divergence events, often
also called \emph{speciation} events.
A phylogenetic network can additionally describe events where lineages merge,
such as hybridization or lateral gene transfer, which are called
\emph{reticulation} events.
A central goal of computational phylogenetics is to develop methods for
reconstructing phylogenetic networks from various types of inputs.

One of the most fundamental problems in this area, \textsc{Hybridization
  Number}, is to find a phylogenetic network with the minimum number of
reticulation events among all networks that contain a given collection
of phylogenetic trees.
The network is said to \emph{display} each of the input trees.
Each of these trees represents the evolution, through speciation events
and mutation, of a particular gene.
Reticulation events such as hybridization or lateral gene transfer can lead
to discordance between gene trees.
The requirement that each gene tree should be contained in the constructed
network ensures that the network provides the required paths along which
each gene could be passed from ancestors to descendants in a manner consistent
with its gene tree.
Following the parsimony principle, a network with the minimum number of
reticulations that displays all inputs trees offers a simplest possible
model of the evolution of a set of taxa consistent with the given gene
trees.
Hence the goal to compute a \leo{phylogenetic} network with as few reticulations
as possible.
Since not all discordance between gene trees is due to reticulation events,
such a network provides only an estimate of the actual number of reticulation
events.
Nevertheless, hybridization networks have proven to be a valuable tool in the
study of the evolution of different sets of taxa.
Computing hybridization networks with the minimum number of reticulations,
however, has proven to be a major challenge.

Initial research focused on the special case that the input consists of only
two trees, in which case there exists a nice mathematical characterization of
the problem in terms of maximum agreement forests (MAFs)
\cite{baroni2005bounding}.
This characterization has shown to be extremely useful for the development of
fixed-parameter algorithms for phylogenetic network construction problems
on two trees \cite{bordewich2007computing,whidden2013fixed,chen2012algorithms},
with the currently fastest algorithm for \textsc{Hybridization Number} running
in $O(3.18^kn)$ time~\cite{whidden2013fixed}.

When the input consists of more than two trees, the problem becomes
significantly harder.
Kernelization is still possible~\cite{van2013quadratic,van2016kernelizations}.
However, existing algorithms for solving kernelized instances,
\textsc{Treetistic}~\cite{treetistic}, PIRN~\cite{wu2010close},
PIRNs~\cite{mirzaei2016fast} and
\textsc{Hybroscale}~\cite{albrecht2015computing,albrecht2014computing}, are
limited to (very) small numbers of input trees and/or (very) small numbers of
reticulation events.
None of these algorithms is fixed-parameter tractable (FPT) unless combined
with kernelization.
A bounded-search FPT algorithm with running time $O(c^k \textrm{poly}(n))$
for the special case of three input trees was proposed
in~\cite{van2016hybridization} ($n$ is the number of taxa, $k$ the number
of reticulations), but the constant~$c$ is much too big for the algorithm to be
useful in practice.

The main bottleneck hindering the development of practical algorithms seemed to
be the missing mathematical characterization for the problem on more than two
trees, analogous to the MAF characterization for two trees.
Such a characterization, in terms of cherry picking sequences, was developed
recently and is very different from the MAF characterization for two trees.
The first characterization in terms of cherry picking sequences was developed
for the restricted class of temporal networks~\cite{humphries2013cherry}.
Subsequently, it was generalized to the larger class of tree-child
networks~\cite{LinzSemple2017}, in which each non-leaf vertex is required to
have at least one non-reticulate child.
However, Humphries, Linz, and Semple~\cite{humphries2013cherry} provide only a
theoretical FPT result based on kernelization for temporal networks, and Linz
and Semple~\cite{LinzSemple2017} do not present any algorithmic results.
Hence, the fixed-parameter tractability of the tree-child version of
\textsc{Hybridization Number} remained open, as well as the development of
practical FPT algorithms based on the new characterization.

Our contribution is to fill this algorithmic gap.
We show that there exists an FPT algorithm for \textsc{Hybridization Number}
restricted to tree-child networks on an arbitrary collection of binary input
trees.
Its running time is $O((8k)^k \cdot \mathrm{poly}(n, t))$, where $n$ is the
number of taxa, $t$ is the number of trees, and $k$ is the number of
reticulations in the computed network.
We verify experimentally that, combined with two heuristic improvements that
both preserve the correctness of the algorithm, it can solve fairly complex
instances of tree-child \textsc{Hybridization Number}.
These two heuristics are cluster reduction~\cite{bordewich2007clustering}
and a redundant branch elimination technique introduced in this paper.
The implementation used in our experiments is available from
\verb|https://github.com/nzeh/tree_child_code|.

The main practical benefit of our algorithm is that it can handle many more
input trees than existing methods.
Indeed, in experiments on synthetic inputs, the running time grows roughly
linearly in the number of trees and taxa.
On the other hand, the running time still has a large exponential dependency on
the number of reticulation events~$k$.
Nevertheless, as long as~$k$ is small (at most 7--12), our algorithm can solve
inputs with up to~$100$ input trees and $200$ taxa.
In our experiments on real-world data, we observed that these data sets have
substantially more structure than random synthetic data sets, which makes
cluster reduction and redundant branch elimination more effective and allowed
our algorithm to solve inputs with \leo{up to 8 trees and
50 reticulations}.
As the number of trees increases, however, the inputs become less
``clusterable'', which reduces the number of reticulations our algorithm can
handle.

We also compared our algorithm directly to \textsc{Hybroscale}.
For instances consisting of two input trees, \textsc{Hybroscale} is much faster
because it exploits the MAF characterization for this case.
When the number of input trees is at least three, our algorithm turns
out to be much faster than \textsc{Hybroscale}, which could handle only very
few instances with more than five trees.

We restrict our attention to tree-child networks for two reasons.
First, although Linz and Semple~\cite{LinzSemple2017} also provided a
characterization of unrestricted hybridization networks in terms of cherry
picking sequences, this characterization is based on adding leaves;
since it is not known where to add these leaves, this characterization does not
seem to be directly useful for developing FPT algorithms.
Furthermore, we observed in our experiments that the optimal tree-child
network for a set of trees often has the same number of reticulations as
an optimal unrestricted hybridization network.
Hence, the restriction to tree-child networks allows us to deal with larger
numbers of input trees without changing the problem substantially.

The remainder of this paper is organized as follows:
Section~\ref{sec:preliminaries} formally defines the key concepts including
the \textsc{Hybridization Number} and \textsc{Tree-Child Hybridization}
problems.
Section~\ref{sec:algorithm} presents our FPT algorithm for
\textsc{Tree-Child Hybridization}.
Section~\ref{sec:redundant-branch-elimination} presents our redundant
branch elimination heuristic for speeding up the algorithm in practice.
This section also shows that redundant branch elimination preserves the
correctness of the computed cherry picking sequence.
Section~\ref{sec:experiments} presents some details of our implementation
of the algorithm and discusses our experimental results.
We present some concluding remarks in Section~\ref{sec:conclusions}.

\section{Preliminaries and Definitions}

\label{sec:preliminaries}

\subsection{Phylogenetic Trees and Networks}

Throughout this paper, we denote by $X$ a finite non-empty set of taxa.
A \emph{phylogenetic network on a subset $X' \subseteq X$} is a directed
acyclic graph $N$ whose nodes satisfy the following properties:
There is a single node of in-degree $0$ and out-degree $2$, called the
\emph{root}; the nodes of in-degree $1$ and out-degree~$0$ are bijectively
labelled with elements from $X'$ (the \emph{leaves}); all other nodes either
have in-degree $1$ and out-degree~$2$ (the \emph{tree nodes}) or have
out-degree $1$ and in-degree at least $2$ (the \emph{reticulations}).
This is illustrated in Figure~\ref{fig:network-and-trees}a.
A \emph{phylogenetic tree on $X'$} is a phylogenetic network on $X'$ without
reticulations; see Figure~\ref{fig:network-and-trees}b.
Given a directed edge $uv$ in a phylogenetic network or tree, we say that $u$
is a \emph{parent} of $v$ and $v$ is a \emph{child} of $u$.
If $|X'| = 1$ then a phylogenetic network or tree on $X'$ consists of a single
node labelled with the unique element of $X'$.

For brevity, we usually refer to phylogenetic networks and phylogenetic trees
as \emph{networks} and \emph{trees}, respectively.
When we feel the need to state the label set $X'$ of a phylogenetic tree
explicitly, especially when we want to emphasize that a set of trees all share
the same leaf set, we do refer to this tree as an \emph{$X'$-tree}.

Given a directed edge $uv$ in a network $N$, we call $uv$ a \emph{reticulation
  edge} if $v$ is a reticulation;
otherwise, $uv$ is a \emph{tree edge}.
A \emph{tree path} in $N$ is a directed path composed of only tree edges.
A tree path is shown in red in Figure~\ref{fig:network-and-trees}a.
The \emph{reticulation number} of $N$ is the number of reticulation edges in
$N$ minus the number of reticulations.
Alternatively, the reticulation number is the number of edges that need to be
deleted from the network to obtain a tree.

The \emph{restriction} of an $X$-tree $T$ to a subset $X' \subseteq X$ is
the smallest subtree of $T$ that contains all edges on paths between leaves in
$X'$.
If $T$ is an $X$-tree and $T'$ is the restriction of $T$ to some subset
$X' \subseteq X$, we write $T' \subseteq T$.
We also write $T \setminus T'$ to denote the difference $X \setminus X'$ of
the label sets of the two trees.

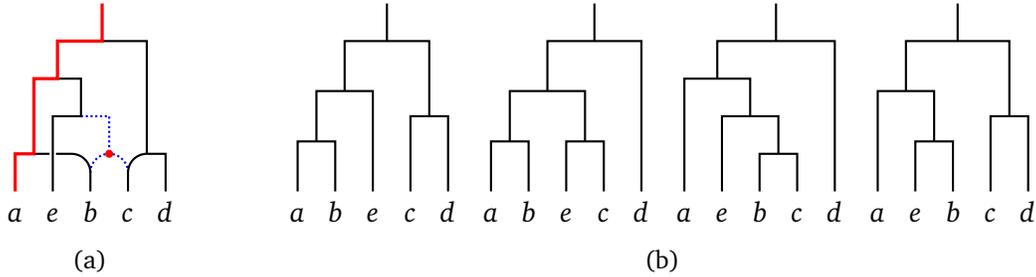
\begin{figure}[t]
  \hspace{\stretch{1}}%
  \subcaptionbox{}{\begin{tikzpicture}[x=0.5cm,y=0.5cm,
      node/.style={fill,circle,inner sep=0pt,minimum size=3pt}]
      \coordinate (a) at (0,0); \path (a) +(90:1) coordinate (pa);
      \coordinate (e) at (1,0); \path (e) +(90:2) coordinate (pe);
      \coordinate (b) at (2,0); \path (b) +(90:0.5) coordinate (pb);
      \coordinate (c) at (3,0); \path (c) +(90:0.5) coordinate (pc);
      \coordinate (d) at (4,0); \path (d) +(90:1) coordinate (pd);
      \coordinate (ab) at (0.5,1); \path (ab) +(90:2) coordinate (pab);
      \coordinate (bc) at (2.5,1); \path (bc) +(90:1) coordinate (pbc);
      \coordinate (cd) at (3.5,1); \path (cd) +(90:3) coordinate (pcd);
      \coordinate (bce) at (barycentric cs:pbc=0.5,pe=0.5); \path (bce) +(90:1) coordinate (pbce);
      \coordinate (abce) at (barycentric cs:pab=0.5,pbce=0.5); \path (abce) +(90:1) coordinate (pabce);
      \coordinate (abcde) at (barycentric cs:pabce=0.5,pcd=0.5); \path (abcde) +(90:1) coordinate (pabcde);
      \foreach \v in {a,b,c,d,e} {
        \node [anchor=north,text height=height("b")] at (\v) {$\v$};
      }
      \draw [thick,densely dotted,blue] (pc) arc (0:180:0.5);
      \draw [thick] (ab) -- ++(0:1) arc (90:0:0.5) -- (b);
      \draw [thick] (d) -- (pd) -- ++(180:0.5) arc (90:180:0.5) -- (c);
      \draw [line width=2.4pt,white] (e) -- (pe);
      \draw [blue,densely dotted,thick] (bce) -- (pbc) -- (bc);
      \draw [thick] (abce) -- (pbce) -- (bce) -- (pe) -- (e);
      \draw [thick] (abcde) -- (pcd) -- (cd);
      \draw [very thick,red] (pabcde) -- (abcde) -- (pabce) -- (abce) -- (pab) -- (ab) -- (pa) -- (a);
      \node [node,red] at (bc) {};
    \end{tikzpicture}}%
  \hspace{\stretch{1}}%
  \subcaptionbox{}{\begin{tikzpicture}[x=0.5cm,y=0.5cm]
      \coordinate (a) at (0,0); \path (a) +(90:1.33) coordinate (pa);
      \coordinate (b) at (1,0); \path (b) +(90:1.33) coordinate (pb);
      \coordinate (e) at (2,0); \path (e) +(90:2.67) coordinate (pe);
      \coordinate (c) at (3,0); \path (c) +(90:2) coordinate (pc);
      \coordinate (d) at (4,0); \path (d) +(90:2) coordinate (pd);
      \coordinate (ab) at (barycentric cs:pa=0.5,pb=0.5); \path (ab) +(90:1.34) coordinate (pab);
      \coordinate (abe) at (barycentric cs:pe=0.5,pab=0.5); \path (abe) +(90:1.33) coordinate (pabe);
      \coordinate (cd) at (barycentric cs:pc=0.5,pd=0.5); \path (cd) +(90:2) coordinate (pcd);
      \coordinate (abcde) at (barycentric cs:pabe=0.5,pcd=0.5); \path (abcde) +(90:1) coordinate (pabcde);
      \foreach \v in {a,b,c,d,e} {
        \node [anchor=north,text height=height("b")] at (\v) {$\v$};
      }
      \draw [thick] (a) -- (pa) -- (pb) -- (b);
      \draw [thick] (c) -- (pc) -- (pd) -- (d);
      \draw [thick] (ab) -- (pab) -- (pe) -- (e);
      \draw [thick] (abe) -- (pabe) -- (pcd) -- (cd);
      \draw [thick] (abcde) -- (pabcde);
    \end{tikzpicture}
    \begin{tikzpicture}[x=0.5cm,y=0.5cm]
      \coordinate (a) at (0,0); \path (a) +(90:1.33) coordinate (pa);
      \coordinate (b) at (1,0); \path (b) +(90:1.33) coordinate (pb);
      \coordinate (e) at (2,0); \path (e) +(90:1.33) coordinate (pe);
      \coordinate (c) at (3,0); \path (c) +(90:1.33) coordinate (pc);
      \coordinate (d) at (4,0); \path (d) +(90:4) coordinate (pd);
      \coordinate (ab) at (barycentric cs:pa=0.5,pb=0.5); \path (ab) +(90:1.34) coordinate (pab);
      \coordinate (ce) at (barycentric cs:pc=0.5,pe=0.5); \path (ce) +(90:1.34) coordinate (pce);
      \coordinate (abce) at (barycentric cs:pab=0.5,pce=0.5); \path (abce) +(90:1.33) coordinate (pabce);
      \coordinate (abcde) at (barycentric cs:pabce=0.5,pd=0.5); \path (abcde) +(90:1) coordinate (pabcde);
      \foreach \v in {a,b,c,d,e} {
        \node [anchor=north,text height=height("b")] at (\v) {$\v$};
      }
      \draw [thick] (a) -- (pa) -- (pb) -- (b);
      \draw [thick] (e) -- (pe) -- (pc) -- (c);
      \draw [thick] (ab) -- (pab) -- (pce) -- (ce);
      \draw [thick] (abce) -- (pabce) -- (pd) -- (d);
      \draw [thick] (abcde) -- (pabcde);
    \end{tikzpicture}
    \begin{tikzpicture}[x=0.5cm,y=0.5cm]
      \coordinate (a) at (0,0); \path (a) +(90:3) coordinate (pa);
      \coordinate (e) at (1,0); \path (e) +(90:2) coordinate (pe);
      \coordinate (b) at (2,0); \path (b) +(90:1) coordinate (pb);
      \coordinate (c) at (3,0); \path (c) +(90:1) coordinate (pc);
      \coordinate (d) at (4,0); \path (d) +(90:4) coordinate (pd);
      \coordinate (bc) at (barycentric cs:pb=0.5,pc=0.5); \path (bc) +(90:1) coordinate (pbc);
      \coordinate (bce) at (barycentric cs:pbc=0.5,pe=0.5); \path (bce) +(90:1) coordinate (pbce);
      \coordinate (abce) at (barycentric cs:pa=0.5,pbce=0.5); \path (abce) +(90:1) coordinate (pabce);
      \coordinate (abcde) at (barycentric cs:pabce=0.5,pd=0.5); \path (abcde) +(90:1) coordinate (pabcde);
      \foreach \v in {a,b,c,d,e} {
        \node [anchor=north,text height=height("b")] at (\v) {$\v$};
      }
      \draw [thick] (b) -- (pb) -- (pc) -- (c);
      \draw [thick] (e) -- (pe) -- (pbc) -- (bc);
      \draw [thick] (a) -- (pa) -- (pbce) -- (bce);
      \draw [thick] (abce) -- (pabce) -- (pd) -- (d);
      \draw [thick] (abcde) -- (pabcde);
    \end{tikzpicture}
    \begin{tikzpicture}[x=0.5cm,y=0.5cm]
      \coordinate (a) at (0,0); \path (a) +(90:2.67) coordinate (pa);
      \coordinate (e) at (1,0); \path (e) +(90:1.33) coordinate (pe);
      \coordinate (b) at (2,0); \path (b) +(90:1.33) coordinate (pb);
      \coordinate (c) at (3,0); \path (c) +(90:2) coordinate (pc);
      \coordinate (d) at (4,0); \path (d) +(90:2) coordinate (pd);
      \coordinate (be) at (barycentric cs:pb=0.5,pe=0.5); \path (be) +(90:1.34) coordinate (pbe);
      \coordinate (abe) at (barycentric cs:pa=0.5,pbe=0.5); \path (abe) +(90:1.33) coordinate (pabe);
      \coordinate (cd) at (barycentric cs:pc=0.5,pd=0.5); \path (cd) +(90:2) coordinate (pcd);
      \coordinate (abcde) at (barycentric cs:pabe=0.5,pcd=0.5); \path (abcde) +(90:1) coordinate (pabcde);
      \foreach \v in {a,b,c,d,e} {
        \node [anchor=north,text height=height("b")] at (\v) {$\v$};
      }
      \draw [thick] (e) -- (pe) -- (pb) -- (b);
      \draw [thick] (c) -- (pc) -- (pd) -- (d);
      \draw [thick] (a) -- (pa) -- (pbe) -- (be);
      \draw [thick] (abe) -- (pabe) -- (pcd) -- (cd);
      \draw [thick] (abcde) -- (pabcde);
    \end{tikzpicture}}%
  \hspace{\stretch{1}}
  \caption{\label{fig:network-and-trees}(a) A phylogenetic network that
    is not tree-child because both children of the red node are reticulations.
    Its reticulation number is 2.
    A tree path from the root to the leaf labelled $a$ is shown in red.
    (b) The four phylogenetic trees displayed by the network in (a).
    For example, the first tree can be obtained by deleting the dotted edges
    in (a).
    The red and black edges constitute an embedding of this tree into the
    network.}
\end{figure}

Let $N'$ be a subgraph (e.g., a path) of the network $N$.
Any edge $uv \in N$ such that $u \in N'$ and $v \notin N'$ is called a
\emph{pendant edge} of $N'$; $v$ is a \emph{pendant node} of $N'$.
When $N$ is a tree, we say the subtree rooted at $v$ is a \emph{pendant
  subtree} of $N'$.

\begin{rem}
  We note that phylogenetic networks as defined in this paper have out-degree
  at most $2$ on all nodes.
  This is consistent with the definitions used by Linz and
  Semple~\cite{LinzSemple2017}.
  As noted by Linz and Semple, restricting network nodes to have out-degree at
  most $2$ does not result in any loss of generality.
  In particular, for the problems discussed in this paper, any instance that has
  a network with out-degree greater than $2$ as a solution also has a network
  with out-degree at most $2$ as a solution.

  While phylogenetic trees may in general have unbounded out-degree,
  we require phylogenetic trees to have maximum out-degree $2$ in this paper,
  that is, we restrict our attention to what are normally called ``binary''
  trees.
  It is an open question whether our algorithm can be extended to
  input trees of unbounded out-degree.
  We note that Linz and Semple's result relating tree-child networks to
  tree-child sequences imposes no restriction on the out-degree of
  phylogenetic trees but does not offer any algorithm to find an optimal
  tree-child sequence or network even for binary trees.
\end{rem}

\subsection{Minimum Tree-Child Hybridization}

Given a network $N$ on a set of taxa $X$ and a tree $T$ on a subset $X'
\subseteq X$, we say that $N$ \emph{displays} $T$ if $T$ can be obtained from a
subgraph of $N$ by suppressing nodes of out-degree and in-degree $1$
\markj{(a node $v$ with out-degree and in-degree $1$ is \emph{suppressed} by deleting $v$ and replacing the edges $uv$ and $vw$ with the single edge $uw$)}.
Equivalently, $N$~displays $T$ if there exists a function $f$, called an
\emph{embedding} of $T$ into $N$, that maps nodes of $T$ to nodes of $N$, and
edges of $T$ to directed paths in $N$, such that
\begin{itemize}[noitemsep]
 \item Every leaf of $T$ is mapped to the leaf of $N$ with the same label;
 \item For each edge $uv$ in $T$, the path $f(uv)$ is a directed path in $N$
   from $f(u)$ to $f(v)$; and
 \item For any two distinct edges $e$ and $e'$ of $T$, the paths $f(e)$ and
   $f(e')$ are edge-disjoint.
\end{itemize}
For any embedding $f$ and any node or edge $x$, we call $f(x)$ the \emph{image}
of $x$ (under $f$).
This definition extends naturally to arbitrary subgraphs $T' \subseteq T$ by
defining the image $f(T')$ of $T'$ to be the union of the images of all nodes and
edges in $T'$.
For a set of trees $\T = \{T_1, \dots, T_t\}$, we say that $N$ \emph{displays}
$\T$ if $N$ displays every tree $T_i \in \T$.
For example, the network in Figure~\ref{fig:network-and-trees}a displays all
trees in Figure~\ref{fig:network-and-trees}b.
An embedding of the first tree into the network is shown.

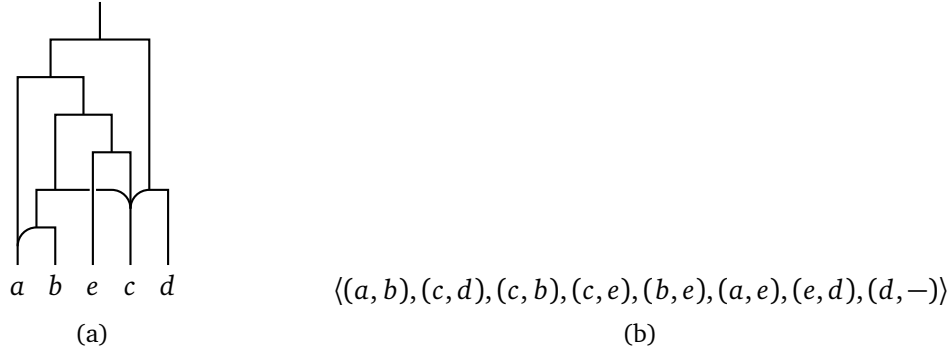
\begin{figure}[t]
  \hspace{\stretch{1}}%
  \subcaptionbox{}{\begin{tikzpicture}[x=0.5cm,y=0.5cm,
      node/.style={fill,circle,inner sep=0pt,minimum size=3pt}]
      \coordinate (a) at (0,0); \path (a) +(90:0.5) coordinate (pa) +(90:5) coordinate (ppa);
      \coordinate (b) at (1,0); \path (b) +(90:1) coordinate (pb);
      \coordinate (e) at (2,0); \path (e) +(90:3) coordinate (pe);
      \coordinate (c) at (3,0); \path (c) +(90:3) coordinate (pc);
      \coordinate (d) at (4,0); \path (d) +(90:2) coordinate (pd);
      \coordinate (ab) at (0.5,1); \path (ab) +(90:1) coordinate (pab);
      \coordinate (abc) at (1,2); \path (abc) +(90:2) coordinate (pabc);
      \coordinate (ce) at (barycentric cs:pc=0.5,pe=0.5); \path (ce) +(90:1) coordinate (pce);
      \coordinate (abce) at (barycentric cs:pabc=0.5,pce=0.5); \path (abce) +(90:1) coordinate (pabce);
      \coordinate (cd) at (3.5,2); \path (cd) +(90:4) coordinate (pcd);
      \coordinate (aabce) at (barycentric cs:ppa=0.5,pabce=0.5); \path (aabce) +(90:1) coordinate (paabce);
      \coordinate (abcde) at (barycentric cs:paabce=0.5,pcd=0.5); \path (abcde) +(90:1) coordinate (pabcde);
      \foreach \v in {a,b,c,d,e} {
        \node [anchor=north,text height=height("b")] at (\v) {$\v$};
      }
      \draw [thick] (pa) arc (180:90:0.5) -- (pb) -- (b);
      \draw [thick] (ab) -- (pab) -- ++(0:2) arc (90:0:0.5);
      \draw [thick] (d) -- (pd) -- ++(180:0.5) arc (90:180:0.5);
      \draw [line width=2.4pt,white] (e) -- (pe);
      \draw [thick] (e) -- (pe) -- (pc) -- (c);
      \draw [thick] (abc) -- (pabc) -- (pce) -- (ce);
      \draw [thick] (a) -- (ppa) -- (pabce) -- (abce);
      \draw [thick] (aabce) -- (paabce) -- (pcd) -- (cd);
      \draw [thick] (abcde) -- (pabcde);
    \end{tikzpicture}}%
  \hspace{\stretch{1}}%
  \subcaptionbox{}{$\langle (a,b), (c,d), (c,b), (c,e), (b,e), (a,e), (e,d), (d,-) \rangle$}%
  \hspace{\stretch{1}}
  \caption{\label{fig:tree-child-networks}(a) An optimal tree-child network for
    the four trees in Figure~\ref{fig:network-and-trees}b.
    Note that this network has reticulation number~3, one more than the
    non-tree-child hybridization network for these trees in
    Figure~\ref{fig:network-and-trees}a.
    The tree-child cherry picking sequences corresponding to this network
    is shown in (b).}
\end{figure}

The {\sc Minimum Hybridization} problem takes as input a set $\T$ of
phylogenetic trees and an integer $k$, and asks for a network displaying $\T$
and with reticulation number at most $k$, if such a network exists.
In this paper, we focus on a restricted version of {\sc Minimum Hybridization},
described below.

A network $N$ is \emph{tree-child} if every non-leaf node of $N$ has at least
one child that is a tree node.
Note that this is equivalent to requiring that every node in $N$ has a tree
path to a leaf.
The network in Figure~\ref{fig:network-and-trees}a is not tree-child because
the children of the red node are both reticulations.
A tree-child network displaying the trees in Figure~\ref{fig:network-and-trees}b
is shown in Figure~\ref{fig:tree-child-networks}a.

\begin{trivlist}
\item \textsc{Minimum Tree-Child Hybridization}
\item \textbf{Input:} A set $\T = \{T_1, \dots, T_t\}$
  of phylogenetic trees on $X$ and an integer $k$.
\item \textbf{Output:} A tree-child phylogenetic network $N$ on $X$ that
  displays $\T$ and has at most $k$ reticulations, if such a network exists;
  {\sc None} otherwise.
\end{trivlist}

For a set $\T = \{T_1, T_2, \ldots, T_t\}$ of $X$-trees,
let $h(\T)$ denote the \emph{hybridization number} of $\T$, 
that is, the minimum reticulation number of all networks that display $\T$.
Similarly,
let $\htc(\T)$ denote the \emph{tree-child hybridization number} of $\T$,
that is, the minimum reticulation number of all tree-child networks that
display $\T$.

\subsection{Cherry Picking Sequences}

For any tree $T$ on $X' \subseteq X$ and any two taxa $x,y \in X'$, 
we say that $\{x,y\}$ is a \emph{cherry} of $T$ if the leaves labelled with $x$
and $y$ are siblings in $T$.
Observe that any tree with two or more leaves contains at least one cherry.
A pair $\{x,y\}$ is a cherry of a set of trees $\T$ if it is a cherry of at
least one tree in $\T$.
It is a \emph{trivial cherry} of $\T$ if $\{x,y\}$ is a cherry of every tree in
$\T$ that contains both $x$ and $y$.

Linz and Semple~\cite{LinzSemple2017} gave a characterization of tree-child
hybridization number in terms of \emph{cherry picking sequences}, which we
define next. \markj{Informally, a cherry picking sequence is a sequence of pairs of leaves, describing a sequence of operations on a set of trees $\T$. In particular a pair of the form $(x,y)$ denotes the operation of removing leaf $x$ from any tree in $\T$ that has $\{x,y\}$ as a cherry, while a pair of the form $(x, -)$ is used when at least one tree in $\T$ has been reduced to the single leaf $x$.}

\markj{Formally, a} \emph{cherry picking sequence} is a sequence
\begin{equation*}
  S = \langle(x_1, y_1), (x_2, y_2), \ldots, (x_r, y_r), (x_{r+1}, -),
  (x_{r+2}, -), \ldots, (x_s, -)\rangle
\end{equation*}
with $\{x_1, x_2, \ldots, x_s, y_1, y_2, \ldots, y_r\} \subseteq X$.
We write $|S|$ to denote the length $s$ of $S$.
It may be that $s=r$, in which case the last element is $(x_r,y_r)$, that is,
there are no pairs of the form $(x_j, -)$.
We call such a sequence a \emph{partial} cherry picking sequence.
A sequence is \emph{full} if $s > r$ and $\{x_1, \dots x_s\} = X$.
For any $1\leq i \leq j \leq s$, we denote by $S_{i,j}$ the subsequence
$\langle(x_i,y_i), \dots, (x_j,y_j)\rangle$ (where $y_h$ is replaced with $-$
for $h > r$). 
Given two sequences $S = \langle(x_1, y_1),  \ldots, (x_r, y_r) \rangle$ and
$S' = \langle(x'_1, y'_1),  \ldots, (x'_{r'}, y'_{r'}), (x'_{r'+1}, -), \ldots,
(x'_{s'}, -)\rangle$, we denote by $S \circ S'$ the sequence $\langle(x_1,
y_1),  \ldots, (x_r, y_r), (x'_1, y'_1),  \ldots, (x'_{r'}, y'_{r'}),
(x'_{r'+1}, -), \ldots, (x'_{s'}, -)\rangle$.
We say that $S \circ S'$ is an \emph{extension} of $S$, and that $S$ is a
\emph{prefix} of $S \circ S'$.
If $S' \neq \langle\rangle$, then we call $S$ a \emph{proper} prefix of $S
\circ S'$.

For a tree $T$ on $X' \subseteq X$, the sequence $S$ defines a sequence of
trees $\langle T^{(0)}, T^{(1)}, \dots, T^{(r)}\rangle$ as follows:
\begin{itemize}[noitemsep]
 \item $T^{(0)} = T$;
 \item If $\{x_j, y_j\}$ is a cherry of $\tree{j-1}$, then $\tree{j}$ is
    obtained from $\tree{j-1}$ by removing $x_j$ and suppressing $y_j$'s
    parent.
    Otherwise, $\tree{j} = \tree{j-1}$.
\end{itemize}
For notational convenience, we refer to $\tree{r}$ as $\reduce{T}{S}$, the
tree obtained by \emph{applying} the sequence $S$ to $T$.
In addition, for a set of trees $\T = \{T_1, \dots, T_t\}$,
we write $\T^{(j)}$ to denote the set $\{\tree[1]{j}, \dots, \tree[t]{j}\}$,
and $\reduce{\T}{S}$ to denote the set $\{\reduce{T_1}{S}, \dots,
  \reduce{T_r}{S}\}$.

A full cherry picking sequence $S = \langle(x_1, y_1), (x_2, y_2), \ldots,
(x_r, y_r), (x_{r+1}, -), (x_{r+2}, -), \ldots, (x_s, -)\rangle$ is a
cherry picking sequence \emph{for a set of trees $\T$} if every tree in
$\reduce{T}{S}$ has a single leaf and that leaf is in $\{x_{r+1}, \dots,
  x_s\}$.
(Note in particular that every cherry picking sequence for a set
of $X$-trees is full.)
The \emph{weight} $w(S)$ of $S$ is defined to be $|S| - |X|$.

A cherry picking sequence $S$ is \emph{tree-child} if $s \leq r+1$ and $y_j
\neq x_i$ for all $1 \le i < j \le s$.
(Thus, if $S$ is a tree-child cherry picking sequence for $\T$, then $T/S$
consists of the single leaf $x_s$ for \emph{every} tree $T \in \T$.)
The tree-child cherry picking sequences for the set of trees in
Figure~\ref{fig:network-and-trees}b corresponding to the two tree-child
networks in Figures~\ref{fig:tree-child-networks}a,b are shown in
Figure~\ref{fig:tree-child-networks}c,d.
If $S$ is a tree-child cherry picking sequence, we refer to the leaves
$\{x_1, \dots x_r\}$ as \emph{forbidden leaves} with respect to $S$, since
they are forbidden to appear as the second element of any cherry $(x_j, y_j)$
with $j > r$ in any tree-child extension of $S$.
We say that $S \circ S'$ is an \emph{optimal tree-child extension} of $S$ if $S
\circ S'$ is a tree-child sequence for $\T$ and every extension $S \circ S''$
of $S$ that is a tree-child sequence for $\T$ satisfies $w(S\circ S'') \ge w(S
\circ S')$.
For the purposes of algorithmic construction of sequences, we adopt the
convention that $S\circ \None = \None$ for any sequence $S$ and that
$w(\None) = \infty$.

Let $\stc(\T)$ be the minimum weight of all tree-child sequences for $\T$.
Linz and Semple showed that the problem of finding the tree-child hybridization
number of a set $\T$ of $X$-trees is equivalent to finding the minimum weight
of a tree-child cherry picking sequence for $\T$:

\begin{thm}[Linz and Semple \cite{LinzSemple2017}]
  \label{thm:cherry-picking-is-hybridization}
  Let $X$ be a set of taxa, and $\T = \{T_1, T_2, \ldots, T_t\}$ a collection
  of phylogenetic $X$-trees.
  Then
  \begin{equation*}
    \stc(\T) = \htc(\T).
  \end{equation*}
\end{thm}




\section{Finding an Optimal Tree-Child Sequence}

\label{sec:algorithm}

In this section, we show that {\sc Minimum Tree-Child Hybridization} is
fixed-parameter tractable with respect to $k$.
Our proof is based on Linz and Semple's characterization of tree-child
hybridization number in terms of tree-child cherry picking sequences (see
Theorem~\ref{thm:cherry-picking-is-hybridization}).
As such, our main technical contribution is to give a fixed-parameter
algorithm, \procref{alg:tree-child-sequence}, for the problem of finding a
tree-child cherry picking sequence of weight at most $k$, if such a sequence
exists.
By the following proposition, a corresponding tree-child network can then be
found in polynomial time.
 
\begin{prop}[Linz and Semple \cite{LinzSemple2017}]
  \label{prop:network-from-sequence}
  There exists a linear-time algorithm that, given a set $\T$ of
  $X$-trees and a tree-child cherry picking sequence $S$ for $\T$, computes a
  tree-child network $N$ displaying $\T$ with $h(N) \leq w(S)$.
\end{prop}

For completeness, the pseudocode of this algorithm,
\procref{alg:tree-child-network-from-sequence}, is given in the appendix.
(Linz and Semple do not state a running time for this algorithm, but it is easy
to observe that their algorithm takes linear time in $n = |X|$, given that there
are at most $n$ reticulations.)

Our algorithm for computing a tree-child cherry picking sequence of length
at most $k$ has the following structure:
Starting with the set of trees $\T$ and the empty sequence $S =
\langle\rangle$, the algorithm repeats the following as long as $\T / S$ still
has a cherry.
If $\T / S$ has a trivial cherry $\{x, y\}$ such that $y$ is not forbidden
with respect to $S$, it adds $(x, y)$ to the end of $S$.
If $\T / S$ has no trivial cherry, we show that $\T / S$ has at most $4k$
unique cherries or $\htc(\T) > k$.
The algorithm makes one recursive call for each pair $(x, y)$ such that
$\{x, y\}$ is a cherry of $\T / S$, starting each recursive call by adding $(x,
y)$ to the end of $S$.
(Note that every cherry $\{x, y\}$ of $\T / S$ gives rise to two recursive
calls, one for the pair $(x, y)$ and one for the pair $(y, x)$.)
As this kind of branching step cannot occur more than $k$ times in a sequence
of weight at most $k$, this gives a search tree for our algorithm of depth $k$
and branching number at most $8k$.

\begin{procedure}[t]
  \caption{TCS()$(\T,S,k)$\label{alg:tree-child-sequence}}
  \KwIn{A collection of phylogenetic trees $\T$, a partial tree-child cherry
    picking sequence $S$, and an integer $k$}
  \KwOut{An optimal solution of $(\T, S)$ if $(\T, S)$ has a solution of weight
    at most $k$; $\None$ otherwise} 
  \While{there exists a trivial cherry $\{x,y\}$ of $\T/S$ with $y$ not
    forbidden with respect to $S$\label{lin:LoopTrivialCherryReductionStart}}{ 
    $S \gets S \circ \langle (x,y)
    \rangle$\;\label{lin:LoopTrivialCherryReductionEnd}
  }
  $\T' \gets \T / S$\;\label{lin:PruneCherries}
  \uIf{$\T'$ contains a cherry $\{x,y\}$ with $x,y$ both forbidden with respect
    to $S$}{
    \Return{$\None$\;}\label{lin:FailForbiddenCherry}
  }\Else{
   $n' \gets |\{x \in X: x$ is a leaf of a tree in
     $\T'\}|$\;\label{lin:DefineNumberRemainingLeaves}
   $k' \gets |S| - |X| + n'$\;\label{lin:DefineCurrentWeight}
   $C \gets \{(x,y) \mid \text{$\{x,y\}$ is a cherry of some tree in
       $\T'$}\}$\;\label{lin:DefineCherrySet}
   \uIf{$|C| = 0$}{
     \Return $S \circ \langle(x,-)\rangle$, where $x$ is the last remaining
     leaf in all trees\;\label{lin:SolutionFound}
   }\uElseIf{$|C| > 8k$ or $k' \geq k$\label{lin:FailureConditions}}{
     \Return{$\None$\;}\label{lin:FailTooManyCherries}\label{lin:FailDugTooDeep}
   }\Else{
     $S_{\textrm{opt}} \gets \None$\;\label{lin:RecursiveSolutionStart}
     \ForEach{$(x,y) \in C$ with $y$ not forbidden with respect to $S$}{    
       $S_{\textrm{temp}} \gets \TreeChildSequence(\T, S \circ
       \langle(x,y)\rangle, k)$\;\label{lin:RecursiveCall}
       \If{$w(S_{\textrm{temp}}) < w(S_{\textrm{opt}})$}{
         $S_{\textrm{opt}} \gets S_{\textrm{temp}}$\;
       }
     }
     \Return{$S_{\textrm{opt}}$\;}\label{lin:RecursiveSolutionEnd}
   }
 } 
\end{procedure}
  
%
   
In the remainder of this section, we prove the correctness of
procedure \procref{alg:tree-child-sequence} and analyze its running time.
This is summarized in the following theorem
\markj{(we denote by $\lg$ the logarithmic function with base~$2$)}.

\begin{thm}
  \label{thm:main}
  Given a collection $\T$ of $t$ $X$-trees with $|X| = n$, it takes $O((8k)^k
  nt \lg t + nt \lg nt)$ time to decide whether $\T$ has tree-child
  hybridization number at most $k$ and, if so, compute a corresponding
  tree-child cherry picking sequence.
\end{thm}

Combined with Proposition~\ref{prop:network-from-sequence}, this proves the
following corollary.

\begin{cor}
  \label{cor:tc-network}
  Given a collection $\T$ of $t$ $X$-trees with $|X| = n$, it takes
  $O((8k)^k nt \lg t + nt \lg nt)$ time to decide whether $\T$ has tree-child
  hybridization number at most $k$ and, if so, compute a corresponding
  tree-child hybridization network that displays $\T$.
\end{cor}

It is easy to see that procedure \procref{alg:tree-child-sequence} returns a
sequence $S$ only if it is a valid tree-child cherry picking sequence for $\T$.
Thus, it suffices to show that if a partial tree-child cherry picking sequence
$S$ has an extension $S \circ S'$ of weight at most $k$ that is a cherry
picking sequence for $\T$, then the invocation $\TreeChildSequence(\T,S,k)$
finds a shortest such extension.
In the remainder of this section, we call an extension $S \circ S'$ of a
partial tree-child cherry picking sequence $S$ a \emph{solution of $(\T, S)$}
if $S \circ S'$ is a cherry picking sequence for $\T$; $S \circ S'$ is an
\emph{optimal} solution of $(\T, S)$ if there is no 
\markj{solution of $(\T, S)$}
that is shorter than $S \circ S'$.

We split the proof of Theorem~\ref{thm:main} into two parts:
First, we show that we deal with trivial cherries correctly:
if $(\T, S)$ has a solution of weight at most $k$ and $\T' = \reduce{\T}{S}$
has a trivial cherry $\{x, y\}$ such that $y$ is not forbidden with respect to
$S$, then $(\T, S \circ \langle (x, y) \rangle)$ has a solution of weight
at most $k$ and any optimal solution of $(\T, S \circ \langle (x, y) \rangle)$
is also an optimal solution of $(\T, S)$.
Thus, adding trivial cherries to $S$ as \procref{alg:tree-child-sequence} does
in
lines~\ref{lin:LoopTrivialCherryReductionStart}--\ref{lin:LoopTrivialCherryReductionEnd}
is safe.
Section~\ref{sec:reducing-trivial-cherries} presents this first part of our
proof.
Second, we show that if $\T'$ has no trivial cherries, then either the trees in
$\T'$ have at most $4k$ unique cherries or $(\T, S)$ has no solution of weight
at most $k$.
Thus, aborting the search if $|C| > 8k$ (since $C$ contains two pairs
for each cherry of $\T'$), as we do in line~\ref{lin:FailTooManyCherries}, is
correct.
The proof of this bound on the number of unique cherries is divided into two
parts.
In Section~\ref{sec:counting-cherries-base-case}, we show that this bound holds
if $S = \langle\rangle$, that is, if all trees in $\T'$ are $X$-trees.
In Section~\ref{sec:counting-cherries-general-case}, we extend this result to
arbitrary partial tree-child cherry picking sequences~$S$.
Section~\ref{sec:wrapping-up} then completes the proof of
Theorem~\ref{thm:main}.

\subsection{Pruning Trivial Cherries}

\label{sec:reducing-trivial-cherries}

Our algorithm begins by repeatedly pruning trivial cherries in
lines~\ref{lin:LoopTrivialCherryReductionStart}--\ref{lin:LoopTrivialCherryReductionEnd};
that is, as long as there exists a trivial cherry $\{x,y\}$ in $\T/S$ with
$y$ not forbidden with respect to $S$, the
algorithm extends $S$ by adding the pair $(x,y)$ to $S$.
In this section, we show that this is safe:
if $(\T, S)$ has solution of weight at most $k$, then so does $(\T, S \circ
\langle (x, y) \rangle)$, and any optimal solution of
$(\T, S \circ \langle (x, y) \rangle)$ is an optimal solution of $(\T, S)$.
We begin with some simple observations.

\begin{prop}
  \label{prop:basic-sequence-properties}
  Let $S = \langle(x_1, y_1), (x_2, y_2), \ldots, (x_r, y_r), (x_{r+1},
  -)\rangle$ be a tree-child cherry picking sequence for a set of $X$-trees
  $\T$.
  Then the following properties hold for all $j \in [r]$:%
  \footnote{We use $[m]$ to denote the set of integers $\{1, \ldots, m\}$ and
    $[m]_0$ to denote the set of integers $\{0, \ldots, m\}$.}

  \begin{enumerate}[label=(\roman{*}),noitemsep,leftmargin=*,widest=iii]
    \item\label{item:yNotForbidden} If $y \in X$ is not forbidden with respect
      to $S_{1,j}$, then $y$ is a leaf in every tree in $\T^{(j)}$.
    \item\label{item:xyCherry} If $\{x,y\}$ is a cherry of\/ $\T^{(j)}$, then
      either $(x,y)$ or $(y,x)$ is a pair in $S_{j+1,r}$.
    \item\label{item:trivialCherry} If $\{x_j,y_j\}$ is a trivial cherry of\/
      $\T^{(j-1)}$, then $x_j$ is not in any tree in  $\T^{(j)}$.
  \end{enumerate}
\end{prop}

\begin{proof}
  Property~\ref{item:yNotForbidden} holds because $y$ is not forbidden with
  respect to $S_{1,j}$ and, thus, $y \ne x_i$ for all $1 \le i \le j$.
  Property~\ref{item:xyCherry} follows because $S_{j+1,r}$ must
  delete at least one of $x,y$ from the tree containing $\{x,y\}$ as a cherry
  and only the pair $(x,y)$ or $(y,x)$ achieves this.
  To see why Property~\ref{item:trivialCherry} holds, observe that $y_j$ is not
  forbidden with respect to $S_{1,j-1}$.
  Thus, by Property~\ref{item:yNotForbidden}, every tree in $\T^{(j-1)}$
  contains $y_j$ as a leaf.
  In particular, every tree in $\T^{(j-1)}$ containing $x_j$ also contains
  $y_j$.
  Thus, by the definition of a trivial cherry, every tree $\T^{(j-1)}$
  containing $x_j$ contains the cherry $\{x_j,y_j\}$.
  Thus, applying the pair $(x_j, y_j)$ to $\T^{(j-1)}$ deletes $x_j$
  from any tree containing $x_j$ and no tree in $\T^{(j)}$ contains $x_j$.
\end{proof}

\begin{lem}
  \label{lem:mayAssumexy}
  Let $S = \langle(x_1, y_1), (x_2, y_2), \ldots, (x_r, y_r), (x_{r+1},
  -)\rangle$ be a tree-child cherry picking sequence for a set of $X$-trees
  $\T$ and suppose that $\{x,y\}$ is a trival cherry of $\T^{(j)}$ and $y$ is
  not forbidden with respect to $S_{1,j}$.
  Then there exists a tree-child cherry picking sequence $S'$ for $\T$
  such that $|S'| = |S|$, $S'_{1,j} = S_{1,j}$, and $(x,y)$ is a pair in
  $S'_{j+1,r}$.
\end{lem}

\begin{proof}
  We start with the following trivial observation:
  Let $\T$ be a set of trees and let $S$ be a tree-child cherry picking
  sequence for $\T$.
  For an arbitrary permutation $\pi$ of $X$ and any $X$-tree $T$, let
  $T_{|\pi}$ be the tree obtained from $T$ by changing the label of each leaf
  from its label $z$ in $T$ to the label $\pi(z)$ in $T_{|\pi}$.
  Let $\T_{|\pi} = \{T_{|\pi} \mid T \in \T\}$.
  Similarly, let $S_{|\pi}$ be the sequence obtained from $S$ by replacing
  every occurrence of an element $z \in X$ in $S$ with $\pi(z)$.
  Then $S_{|\pi}$ is a tree-child cherry picking sequence for $\T_{|\pi}$.
  Here, we consider the permutation $\pi$ such that $\pi(x) = y$, $\pi(y) = x$,
  and $\pi(z) = z$ for all $z \in X \setminus \{x, y\}$, where $\{x, y\}$ is
  a trivial cherry of $\T^{(j)}$.

  By Proposition~\ref{prop:basic-sequence-properties}\ref{item:xyCherry},
  either $(x,y)$ or $(y,x)$ is a pair in $S_{j+1,r}$.
  In the former case, the sequence $S' = S$ satisfies the lemma.
  In the latter case, neither $x$ nor $y$ is forbidden with respect to
  $S_{1,j}$.
  It follows from
  Proposition~\ref{prop:basic-sequence-properties}\ref{item:yNotForbidden}
  and the fact that $\{x, y\}$ is a trivial cherry of $\T^{(j)}$ that
  every tree in $\T^{(j)}$ has $\{x,y\}$ as a cherry.
  In particular, neither $x$ nor $y$ is part of a pair in $S_{1,j}$.
  Thus, since $S$ is a tree-child cherry picking sequence, the sequence
  $S' = S_{1,j} \circ (S_{j+1,r+1})_{|\pi}$ is a tree-child cherry picking
  sequence such that $S'_{1,j} = S_{1,j}$ and $(x, y) \in S'_{j+1,r}$.
  To see that $S'$ is a tree-child cherry picking sequence \emph{for} $\T$,
  observe that $S_{j+1,r+1}$ is a tree-child cherry picking sequence for
  $\T^{(j)}$.
  Thus, as just observed, $(S_{j+1,r+1})_{|\pi}$ is a tree-child cherry picking
  sequence for $\T^{(j)}_{|\pi}$.
  However, since $\{x, y\}$ is a cherry of \emph{every} tree in $\T^{(j)}$, we
  have $\T^{(j)}_{|\pi} = \T^{(j)}$, that is,
  $(S_{j+1,r+1})_{|\pi}$ is a tree-child cherry picking sequence for
  $\T^{(j)}$ and $S' = S_{1,j} \circ (S_{j+1,r+1})_{|\pi}$ is a tree-child
  cherry picking sequence for $\T$.
\end{proof}

\begin{lem}
  \label{lem:subtrees-are-preserved}
  Let $T$ be an $X$-tree, let $T' \subseteq T$, and let $S = \langle (x_1,
  y_1), (x_2, y_2), \ldots, (x_r, y_r)\rangle$ be a partial tree-child cherry
  picking sequence such that $(T \setminus T') \cap \{y_1, y_2, \ldots, y_r\} =
  \emptyset$.
  Then $\reduce{T'}{S} \subseteq \reduce{T}{S}$.
\end{lem}

\begin{proof}
  We prove the claim by induction on $|S|$.
  If $|S| = 0$, then $\reduce{T'}{S} = T' \subseteq T = \reduce{T}{S}$, so the
  claim holds in this case.
  If $|S| > 0$, then let $R' = \reduce{T'}{S_{1,1}}$ and $R =
  \reduce{T}{S_{1,1}}$.
  Note that $R \supseteq T - x_1$.
  If $x_1 \notin T'$, then $R' = T' \subseteq T - x_1 \subseteq R$.
  If $y_1 \notin T'$, then $y_1 \notin T$ because $y_1 \notin T \setminus T'$.
  Thus, $R' = T' \subseteq T = R$.

  So assume that $x_1, y_1 \in T'$.
  If $\{x_1, y_1\}$ is a cherry of $T'$, then $R' = T' - x_1 \subseteq T - x_1
  \subseteq R$.
  If $\{x_1, y_1\}$ is not a cherry of $T'$, then $x_1, y_1 \in T'$ implies
  that the path from $x_1$ to $y_1$ in $T'$ has at least one pendant subtree.
  Since $T' \subseteq T$, this implies that the path from $x_1$ to $y_1$ in
  $T$ also has at least one pendant subtree, that is, $\{x_1, y_1\}$ is not a
  cherry of $T$.
  Therefore, $R' = T' \subseteq T = R$.

  We have shown that in all possible cases, $R' \subseteq R$.
  Now observe that $R \setminus R' \subseteq (T \setminus T') \cup \{x_1\}$.
  Since $S$ is a partial tree-child cherry picking sequence, $S_{2,r}$ is a
  partial tree-child cherry picking sequence and $x_1 \notin \{y_2, y_3,
    \ldots, y_r\}$.
  Since $(T \setminus T') \cap \{y_2, y_3, \ldots, y_r\} = \emptyset$,
  this implies that $(R \setminus R') \cap \{y_2, y_3, \ldots, y_r\} =
  \emptyset$.
  Thus, by the induction hypothesis, $\reduce{T'}{S} = \reduce{R'}{S_{2,r}}
  \subseteq \reduce{R}{S_{2,r}} = \reduce{T}{S}$.
\end{proof}

We are now ready to prove a stronger version of
Lemma~\ref{lem:mayAssumexy}, which establishes that pruning trivial
cherries is safe.

\begin{prop}
  \label{prop:move-up-common-cherries}
  Let $S = \langle(x_1, y_1), (x_2, y_2), \ldots, (x_r, y_r), (x_{r+1},
  -)\rangle$ be a tree-child cherry picking sequence for a set of $X$-trees
  $\T$ and suppose that $\{x,y\}$ is a trival cherry of\/ $\T^{(j)}$ and $y$ is
  not forbidden with respect to~$S_{1,j}$.
  Then there exists a tree-child cherry picking sequence $S' =
  \langle(x_1', y_1'), (x_2', y_2'), \ldots, (x'_{r'}, y'_{r'}), (x'_{r'+1}, -)
  \rangle$ for $\T$
  such that $|S'| \le |S|$, $S'_{1,j} = S_{1,j}$, and $(x'_{j+1}, y'_{j+1}) =
  (x,y)$.
\end{prop}

\begin{proof}
  By Lemma~\ref{lem:mayAssumexy}, there exists a tree-child cherry
  picking sequence $S' = \langle(x_1', y_1'), (x_2', y_2'), \ldots,\break
  (x'_{r'}, y'_{r'}), (x'_{r'+1}, -) \rangle$ for $\T$ such that $r' \le r$,
  $S'_{1,j} = S_{1,j}$ and $(x,y) \in S'_{j+1,r'}$.
  We choose $S'$ from the set of all such cherry picking sequences so that
  the index $j' > j$ with $(x_{j'}, y_{j'}) = (x, y)$ is minimized.
  If $j' = j + 1$, the lemma holds.
  If $j' > j + 1$, we obtain a contradiction to the choice of $S'$ by
  transforming $S'$ into another tree-child cherry picking sequence $S'' =
  \langle (x_1'', y_1''), \ldots, (x''_{r''}, y'_{r''}),
  (x''_{r''+1}, -) \rangle$ for $\T$ such that $|S''| \le |S'| \le |S|$,
  $S''_{1,j} = S'_{1,j} = S_{1,j}$, and $(x''_{j'-1}, y''_{j'-1}) = (x,y)$.

  So assume that $j' > j + 1$ and let $(x'_{j'-1}, y'_{j'-1}) = (v,w)$.
  We distinguish two cases:

  \begin{description}
    \item[\boldmath$w = x$:]
      In this case, we set $r'' = r' - 1$, $(x_h'', y_h'') = (x_h', y_h')$ for
      all $1 \le h \le j' - 2$, and $(x''_h, y''_h) = (x'_{h+1}, y'_{h+1})$ for
      all $j'-1 \le h \le r'' + 1$; that is, we obtain $S''$ by deleting the
      pair $(x'_{j'-1}, y'_{j'-1})$ from $S'$.
      Thus, $S'' \subset S'$, $|S''| < |S'|$, and $(x''_{j'-1}, y''_{j'-1}) =
      (x'_{j'}, y'_{j'}) = (x, y)$.
      Since $S'$ is a tree-child cherry picking sequence, this implies that
      $S''$ also is a tree-child cherry picking sequence.
      To see that $S''$ is a tree-child cherry picking sequence \emph{for\/
        $\T$}, it suffices to prove that $T/S'_{1,j'-2} = T/S'_{1,j-1}$ and,
      thus, $T/S'' = (T/S'_{1,j'-2})/S'_{j',r'} = (T/S'_{1,j'-1})/S'_{j',r'} =
      T/S'$ for every tree $T \in \T$.

      To prove this, observe that $v \ne y$ and $y \in T/S'_{1,h}$ for all $T
      \in \T$ and all $1 \le h < j'$ because $y_{j'} = y$, that is, $y$ is not
      forbidden with respect to $S'_{1,j'-1}$.
      Thus, since $\{x, y\}$ is a trivial cherry of $\T/S_{1,j} = \T/S'_{1,j}$
      and $j < j'$, $\{x, y\}$ is a cherry of every tree $T/S'_{1,j}$ in
      $\T/S'_{1,j}$ that contains $x$.
      Since $y$ is also a leaf of every tree $T/S'_{1,j'-2}$ in
      $\T/S'_{1,j'-2}$ (again, because $y$ is not forbidden with respect to
      $S'_{1,j'-1}$), this implies that $\{x, y\}$ is also a cherry of every
      tree in $\T/S'_{1,j'-2}$ that contains $x$.
      In particular, since $v \ne y$, $\{v,w\} = \{v,x\}$ is not a cherry of
      any tree $T/S'_{1,j'-2}$ in $\T/S'_{1,j'-2}$ and $T/S'_{1,j'-2} =
      T/S'_{1,j'-1}$ for all $T \in \T$.

    \item[\boldmath$w \ne x$:]
      In this case, we set $(x''_{j'-1}, y''_{j'-1}) = (x'_{j'}, y'_{j'})$,
      $(x''_{j'}, y''_{j'}) = (x'_{j'-1}, y'_{j'-1})$, and $(x''_h, y''_h) =
      (x'_h, y'_h)$ for all $h \notin \{j'-1, j'\}$, that is, we obtain $S''$
      by swapping $(x'_{j'-1}, y'_{j'-1}) = (v, w)$ and $(x'_{j'}, y'_{j'}) =
      (x, y)$ in $S$.
      This clearly implies that $|S''| = |S'|$ and 
      $(x''_{j'-1}, y''_{j'-1}) = (x'_{j'}, y'_{j'}) = (x, y)$.
      To see that $S''$ is a tree-child cherry picking sequence, observe that
      every pair $(x''_h, y''_h)$ in $S''$ with $h \ne j'$ is preceded by a
      subset of the pairs that precede it in $S'$.
      Thus, since $S'$ is a tree-child cherry picking sequence, $y''_h$ is not
      forbidden with respect to $S''_{1,h-1}$.
      For the pair $(x''_{j'}, y''_{j'})$, $y''_{j'}$ is not forbidden with
      respect to $S''_{1,j'-2}$ because $S''_{1,j'-2} = S'_{1,j'-2}$ and
      $(x''_{j'}, y''_{j'}) = (x'_{j'-1}, y'_{j'-1})$.
      This implies that $y''_{j'}$ is not forbidden with respect to
      $S''_{1,j'-1}$ because $y''_{j'} = y'_{j'-1} = w \ne x = x'_{j'} =
      x''_{j'-1}$.

      It remains to show that $T/S'' = T/S'$ for all $T \in \T$.
      To this end, it suffices to show that $T/S'' \subseteq T/S'$ because
      $T/S'$ has only one leaf, $x_{r+1}$, and $T/S'' \ne \emptyset$, that is,
      $T/S'' \subseteq T/S'$ implies that $T/S'' = T/S'$.

      To see that $T/S'' \subseteq T/S$, let $T' = T/S'_{1,j'-2}$.
      Then $T'/\langle (x, y) \rangle \subseteq T'$, $T' \setminus
     ( \reduce{T'}{\langle (x, y) \rangle}) \subseteq \{x\}$, and $x \notin
      \{w, y\}$.
      By Lemma~\ref{lem:subtrees-are-preserved}, this implies that
      $T'/\langle (x, y), (v, w), (x, y) \rangle \subseteq T'/\langle (v, w),
      (x, y) \rangle$.
      However, as argued above, $\{x,y\}$ is a cherry of~$T'$, so $x \notin
      T'/\langle (x,y) \rangle$ and, thus, $x \notin T'/\langle (x,y), (v,w)
      \rangle$.
      This implies that $T'/\langle (x,y), (v,w), (x,y) \rangle =
      T'/\langle (x,y), (v,w) \rangle$ and, therefore, $T'/\langle (x,y), (v,w)
      \rangle \subseteq T' \langle (v,w), (x,y) \rangle$.
      Since $T' = T/S'_{1,j'-2} = T/S''_{1,j'-2}$, $S'_{1,j'} = S'_{1,j'-2}
      \circ \langle (v,w), (x,y) \rangle$, and $S''_{1,j'} = S''_{1,j'-2} \circ
      \langle (x,y), (v,w) \rangle$, this shows that
      $T/S''_{1,j'} \subseteq T/S'_{1,j'}$.
      Using Lemma~\ref{lem:subtrees-are-preserved} again, this shows that
      $T/S'' = (T/S''_{1,j'})/S''_{j'+1,r'} = (T/S''_{1,j'})/S'_{j'+1,r'}
      \subseteq (T/S'_{1,j'})/S'_{j'+1,r'} = T/S'$.\qedhere
  \end{description}
\end{proof}

\subsection{\boldmath Bounding the Number of Cherries in Irreducible $X$-trees}

\label{sec:counting-cherries-base-case}

Once the algorithm has eliminated all trivial cherries from a set of input
trees, each of the remaining (non-trivial) cherries of $\T/S$ is a candidate
for being the next pair to be added to $S$.
Our algorithm makes one recursive call for each possible choice of this next pair
(lines~\ref{lin:RecursiveSolutionStart}--\ref{lin:RecursiveSolutionEnd}).
In order to limit the number of recursive calls it makes, the algorithm aborts
and reports failure if there are more than $8k$ choices to branch on.
To prove that this does not prevent us from finding a tree-child cherry
picking sequence of weight at most $k$, if such a sequence exists, we need to
prove the following claim:

\begin{prop}
  \label{prop:few-non-trivial-cherries}
  If $(\T, S)$ has a solution of weight at most $k$ and $\reduce{\T}{S}$ has
  no trivial cherries, then the number of unique cherries in $\T/S$ is at most
  $4k$.
\end{prop}

Note that this claim refers to the weight $k$ of the \emph{whole} sequence
$S\circ S'$, not the weight of $S'$.
\markj{This is because the proof uses the structure of $S$ as well as $S'$ to bound the number of unique cherries in $\reduce{\T}{S}$.}

Our proof has two parts:
In this subsection, we consider the case when $S = \langle\rangle$, that is,
when we have a set of $X$-trees $\T$ with hybridization number at most $k$ and
no trivial cherries.
In the next subsection, we prove the claim for $S \ne \langle\rangle$, via a
reduction to the case when $S = \langle\rangle$.

\begin{lem}
  \label{lem:4k-cherries}
  If $\T$ is a set of $X$-trees without trivial cherries and with tree-child
  hybridization number $k$, then the total number of cherries of the trees
  in $\T$ is at most $4k$.
\end{lem}

\begin{proof}
  Let $N$ be a tree-child network with $k$ reticulations that displays $\T$
  and, for each tree $T_i \in \T$, let $f_i$ be an embedding of $T_i$ into $N$.
  Our strategy is to ``charge'' each cherry $\{x, y\}$ of $\T$ to some
  reticulation edge in a manner that charges every reticulation edge for at
  most two cherries.
  Since $N$ has hybridization number at most $k$ and, therefore, at most $2k$
  reticulation edges, this proves the lemma.

  We start by proving a number of auxiliary claims about how the images of
  cherries interact with reticulation edges and with each other.
  The first three claims consider a fixed cherry $\{x, y\}$ of some tree $T_i
  \in \T$ and a fixed tree $T_j$ that does not have $\{x, y\}$ as a cherry.
  Since $\{x, y\}$ is non-trivial, such a tree $T_j$ exists.
  Let $p$ be the common parent of $x$ and $y$ in $T_i$ and let $e_x = px$
  and $e_y = py$ be the parent edges of $x$ and $y$ in $T_i$, respectively.
  Since $T_j$ is an $X$-tree, we have $x, y \in T_j$.
  Let $u$ be the LCA of $x$ and $y$ in~$T_j$, and let $P_x$ and $P_y$ be the
  paths from $u$ to $x$ and from $u$ to $y$ in~$T_j$, respectively.
  Since $\{x, y\}$ is not a cherry of~$T_j$, the path $P_x \cup P_y$
  has at least one pendant edge.

  \begin{clm}
    \label{clm:pendants-are-reticulations}
    All pendant nodes of $f_i(e_x) \cup f_i(e_y)$ are reticulations.
  \end{clm}

  \begin{proof}
    Consider any pendant node $w$ of $f_i(e_x) \cup f_i(e_y)$ and let $e$
    be the edge connecting $w$ to a node $v$ in $f_i(e_x) \cup f_i(e_y)$.
    Neither endpoint of $e$ is the root of $N$.
    Since $N$ is a tree-child network, there exists a tree path $Q$
    from $w$ to a leaf $f_i(\ell_w)$.
    Consider the path $P$ from the root to $\ell_w$ in $T_i$.
    Since $e_x$ and $e_y$ are not in~$P$, $f_i(e_x) \cup f_i(e_y)$ and
    $f_i(P)$ are edge-disjoint.
    On the other hand, since $Q$ is a tree path, $Q \subseteq f_i(P)$.
    Since $w$ is not the root of $N$ and $f_i(P)$'s top endpoint is the
    root of $N$, $Q$ is a \emph{proper} subpath of $f_i(P)$,
    that is, $f_i(P)$ contains a parent edge of $w$.
    If $f_i(P)$ contained $e$, then $f_i(P)$ would be a proper superpath
    of $Q \cup e$ because $e$'s top endpoint also is not the root of $N$.
    Thus, $f_i(P)$ would contain the parent edge of $v$, that is, $f_i(P)$ and
    $f_i(e_x) \cup f_i(e_y)$ would not be edge-disjoint, a contradiction.
    Therefore, $e \notin f_i(P)$ and $w$ has another parent edge,
    that is, $w$ is a reticulation.
  \end{proof}

  \begin{clm}
    \label{clm:one-or-no-path-reticulation}
    The path $f_i(e_x) \cup f_i(e_y)$ contains at most one reticulation.
    This reticulation is incident to $f_i(p)$.
  \end{clm}

  \begin{proof}
    We prove that only the top edge of $f_i(e_x)$ can be a reticulation
    edge.
    An analogous argument shows that only the top edge of $f_i(e_y)$ can be
    a reticulation edge.
    Thus, all reticulation edges in $f_i(e_x) \cup f_i(e_y)$ are incident
    to $f_i(p)$.
    If the top edges of $f_i(e_x)$ and $f_i(e_y)$ are both reticulation edges,
    then both children of $f_i(p)$ are reticulations, a contradiction because
    $N$ is a tree-child network.
    Thus, $f_i(e_x) \cup f_i(e_y)$ contains at most one reticulation.

    So assume that $f_i(e_x)$ contains a reticulation edge and choose such an
    edge $e$ that is closest to $f_i(p)$.
    If $e$ is incident to $p$, our claim holds.
    So assume $e$ is not incident to $p$ and let $z$ be its top endpoint.
    By the choice of $e$, $z$ is a tree node.
    However, by Claim~\ref{clm:pendants-are-reticulations}, this implies
    that both of $z$'s children are reticulations, a contradiction again
    because $N$ is a tree-child network.
  \end{proof}

  \begin{clm}
    \label{clm:pendant-or-path-reticulation}
    If the path $f_i(e_x) \cup f_i(e_y)$ contains no reticulation, then it
    has at least one pendant node.
  \end{clm}

  \begin{proof}
    If $f_i(e_x) \cup f_i(e_y)$ contains no reticulation and has no pendant
    nodes, then $f_i(x)$ and $f_i(y)$ are children of $f_i(p)$ in $N$.
    Thus, both $f_j(P_x)$ and $f_j(P_y)$ include $f_i(p)$.
    Since $f_j(P_x)$ and $f_j(P_y)$ share only their top endpoint $f_j(u)$, we
    have $f_j(u) = f_i(p)$ and thus $f_j(P_x) = f_i(e_x)$ and $f_j(P_y) =
    f_i(e_y)$.
    This, however, is a contradiction because $f_i(e_x) \cup f_i(e_y)$ has no
    pendant nodes but $P_x \cup P_y$ has a pendant node in $T_j$, that is,
    $f_j(P_x) \cup f_j(P_y)$ must also have a pendant node in $N$.
  \end{proof}

  For the next two claims, fix two distinct cherries $\{x, y\}$ and $\{w, z\}$
  of two trees $T_i \in \T$ and $T_j \in \T$, respectively.
  Let $p$ be the common parent of $x$ and $y$ in $T_i$, and let $q$ be the
  common parent of $w$ and $z$ in $T_j$.

  \begin{clm}
    \label{clm:path-reticulations-are-unique}
    $f_i(e_x) \cup f_i(e_y)$ and $f_j(e_w) \cup f_j(e_z)$ do not share any
    reticulation edge.
  \end{clm}

  \begin{proof}
    Assume the contrary.
    Then let $e$ be a reticulation edge in $(f_i(e_x) \cup f_i(e_y)) \cap
    (f_j(e_w) \cup f_j(e_z))$ and assume w.l.o.g.\ that $e \in f_i(e_x) \cap
    f_j(e_w)$.
    By Claim~\ref{clm:one-or-no-path-reticulation}, $f_i(p) = f_j(q)$; $e$ is
    the first edge in both $f_i(e_x)$ and in $f_j(e_w)$; $f_i(e_y)$ and
    $f_j(e_z)$ are both tree paths from $f_i(p)$ to $f_i(y)$ and $f_j(z)$,
    respectively; and the subpaths of $f_i(e_x)$ to $f_i(e_w)$ from $e$'s
    bottom endpoint to $f_i(x)$ and $f_j(w)$, respectively, are also tree
    paths.

    Since every pendant node of $f_i(e_y)$ is a reticulation, by
    Claim~\ref{clm:pendants-are-reticulations}, none of these pendant
    nodes can belong to $f_j(e_z)$.
    Thus, $f_j(z) = f_i(y)$, that is, $z = y$.
    Similarly, none of the pendant nodes of the subpath of $f_i(x)$ from
    $e$'s bottom endpoint to $f_i(x)$ can belong to $f_j(w)$.
    Thus, $f_j(w) = f_i(x)$, that is, $w = x$.
    This shows that $\{x, y\} = \{w, z\}$, a contradiction.
  \end{proof}

  \begin{clm}
    \label{clm:pendant-reticulations-are-unique}
    If neither $f_i(e_x) \cup f_i(e_y)$ nor $f_j(e_w) \cup f_j(e_z)$ contains
    a reticulation edge, then these two paths are vertex-disjoint.
  \end{clm}

  \begin{proof}
    Assume that neither $f_i(e_x) \cup f_i(e_y)$ nor $f_j(e_w) \cup f_j(e_z)$
    contains a reticulation edge and assume first that $f_i(e_x) \cup f_i(e_y)$
    and $f_j(e_w) \cup f_j(e_z)$ are not \emph{edge-disjoint}.
    Then, w.l.o.g., $f_i(e_x)$ and $f_j(e_w)$ share an edge~$e$.
    Since $f_i(e_x)$ and $f_j(e_w)$ are tree paths, the same argument as in the
    proof of Claim~\ref{clm:path-reticulations-are-unique} shows that $x = w$.
    If $f_i(e_y)$ and $f_j(e_z)$ also share an edge, then the same argument
    shows that $y = z$.
    Otherwise, w.l.o.g.\ $f_j(q)$ is an internal node of $f_i(e_x)$ and
    the first node after $f_j(q)$ in $f_j(e_z)$ is a pendant node of
    $f_i(e_x)$.
    By Claim~\ref{clm:pendants-are-reticulations}, this node is a reticulation,
    a contradiction.
    This shows that $f_i(e_x) \cup f_i(e_y)$ and $f_j(e_w) \cup f_j(e_z)$
    are edge-disjoint.

    If $f_i(e_x) \cup f_i(e_y)$ and $f_j(e_w) \cup f_j(e_z)$ are edge-disjoint
    but not vertex-disjoint, then their shared vertex $v$ satisfies either $v
    \ne f_i(p)$ and $v \ne f_j(q)$ or w.l.o.g.\ $v = f_i(p)$.
    In the former case, the parent edge of $v$ belongs to both
    $f_i(e_x) \cup f_i(e_y)$ and $f_j(e_w) \cup f_j(e_z)$, a contradiction.
    In the latter case, both child edges of $v$ belong to
    $f_i(e_x) \cup f_i(e_y)$ and $f_j(e_w) \cup f_j(e_z)$ has to contain at least
    one of them, again a contradiction.
  \end{proof}

  Now we call a cherry $\{x, y\}$ of some tree $T_i$ a \emph{type-I cherry} if
  the path $f_i(e_x) \cup f_i(e_y)$ contains a reticulation edge; otherwise, it
  is a \emph{type-II cherry}.
  We charge each type-I cherry $\{x, y\}$ to the reticulation edge in $f_i(e_x)
  \cup f_i(e_y)$.
  By Claim~\ref{clm:path-reticulations-are-unique}, every reticulation edge
  is charged for at most one type-I cherry.
  For every type-II cherry $\{x, y\}$,
  Claim~\ref{clm:pendant-or-path-reticulation} shows that w.l.o.g., $f(x)$'s
  sibling $v$ in $N$ is a pendant node of $f(e_x)$.
  By Claim~\ref{clm:pendants-are-reticulations}, $v$ is a reticulation.
  Thus, the edge $e$ between $v$ and $f(x)$'s parent is a reticulation edge.
  We charge the cherry $\{x, y\}$ to $e$.
  Since $e$ has an endpoint in $f(e_x)$,
  Claim~\ref{clm:pendant-reticulations-are-unique} implies that $e$ is
  charged for only one type-II cherry.
  This proves that every reticulation edge is charged for at most two
  reticulations, one of type I and one of type II\@.
  This finishes the proof.
\end{proof}

\subsection{Bounding the Number of Cherries in General Irreducible Trees}

\label{sec:counting-cherries-general-case}

Having shown, in Lemma~\ref{lem:4k-cherries}, that
Proposition~\ref{prop:few-non-trivial-cherries} holds when $S =
\langle\rangle$, we extend the proof to arbitrary partial tree-child cherry
picking sequences in this section, thereby completing the proof of
Proposition~\ref{prop:few-non-trivial-cherries}.
The main idea is to construct a set of $X$-trees $\hat\T$ that has the same set
of cherries as $\T/S$ (and in particular has no trivial cherries) and then show
that $\hat\T$ has reticulation number at most $k$.
By Lemma~\ref{lem:4k-cherries}, this implies that $\hat\T$, and thus $\T/S$,
has at most $4k$ cherries.

\begin{lem}
  \label{lem:maintain-4k-cherries}
  Let $\T$ be a set of $X$-trees and let $S = \langle(x_1, y_1), (x_2, y_2),
  \ldots, (x_r, y_r), (x_{r+1}, -)\rangle$ be a tree-child cherry picking
  sequence for $\T$ of weight at most $k$.
  For any $j \in [r]_0$, either there exists a trivial cherry of
  $\T^{(j)}$, or $\T^{(j)}$ has at most $4k$ unique cherries.
\end{lem}

\begin{proof}
  For $j = 0$, the claim holds by Lemma~\ref{lem:4k-cherries}.
  For $j > 0$, we cannot apply Lemma~\ref{lem:4k-cherries} directly because the
  trees in $\T^{(j)}$ may have different leaf sets.
  Assume that $\T^{(j)}$ has no trivial cherry, because otherwise the lemma
  holds.
  In order to use Lemma~\ref{lem:4k-cherries} to bound the number of
  unique cherries in $\T^{(j)}$, we transform $\T^{(j)}$ into a set of
  $X$-trees $\hat\T^{(j)}$ with the following properties:
  \begin{enumerate}[noitemsep]
    \item\label{con:sameCherries} $\hat{\T}^{(j)}$ has the same unique cherries
      as  $\T^{(j)}$;
    \item\label{con:noTrivialCherries}  $\hat{\T}^{(j)}$ has no trivial
      cherries; and
    \item\label{con:hybridizationNumber}  $\hat{\T}^{(j)}$ has tree-child
      hybridization number at most $k$.
  \end{enumerate}
  By Properties~\ref{con:noTrivialCherries} and~\ref{con:hybridizationNumber}
  and Lemma~\ref{lem:4k-cherries}, $\hat\T^{(j)}$ has at most $4k$ unique
  cherries.
  Thus, by Property~\ref{con:sameCherries}, $\T^{(j)}$ has at most $4k$
  unique cherries.
 
  To obtain $\hat\T^{(j)}$ from $\T^{(j)}$, let $\T' \subseteq \T$ be the
  subset of trees $T \in \T$ such that $T^{(j)}$ has at least two leaves.
  We can assume that $\T' \ne \emptyset$ because otherwise, $\T^{(j)}$ has
  no cherries and the claim holds.
  Also note that every cherry of $\T^{(j)}$ is a cherry of some tree $T^{(j)}$
  with $T \in \T'$.
  Now consider any tree $T \in \T'$ and let $i_1 < \ldots < i_\ell$ be the
  indices in $[j]$ such that $(x_{i_h}, y_{i_h})$ is a cherry of $T^{(i_h -
    1)}$ for all $1 \le h \le \ell$.
  In other words, $T^{(i)} \ne T^{(i-1)}$ if and only if $i \in \{i_1, \ldots,
    i_\ell\}$.
  Observe that $T^{(j)}$ has label set $X \setminus \{x_{i_1}, \ldots,
    x_{i_\ell}\}$.
  Let $C$ be a caterpillar with leaf set $\{z, x_{i_1}, \ldots, x_{i_\ell}\}$,
  from bottom to top.
  We construct a tree $\hat{T}^{(j)}$ from $T^{(j)}$ and $C$ by identifying $z$
  with the root of $T^{(j)}$.
  This is illustrated in Figure~\ref{fig:hat-T}.
  $\hat\T^{(j)}$ is the set of all such trees $\hat T^{(j)}$:
  $\hat\T^{(j)} = \{\hat T^{(j)} \mid T \in \T'\}$.

  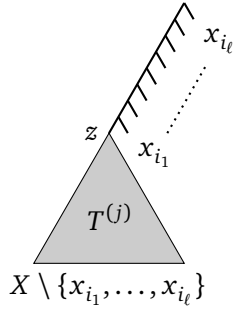
\begin{figure}[t]
    \centering
    \begin{tikzpicture}[x=0.5cm,y=0.5cm]
      \path [draw=black,fill=black!20] (0,0) -- ++(0:4) -- ++(120:4) -- cycle;
      \node [anchor=north] at (2,0) {$X \setminus \{x_{i_1}, \ldots, x_{i_\ell}\}$};
      \draw [thick] (60:4) -- +(60:4);
      \draw [thick] (60:4) foreach \i in {0,...,7} {
        ++(60:0.5) -- +(300:0.5) coordinate (c\i)
      };
      \node [anchor=north west] (l0) at (c0) {$x_{i_1}$};
      \node [anchor=north west] (l1) at (c7) {$x_{i_\ell}$};
      \draw [thick,dotted] (l0) -- (l1);
      \node [anchor=east] at (60:4) {$z$};
      \node at (2,1.25) {$T^{(j)}$};
    \end{tikzpicture}
    \caption{The construction of the tree $\hat T^{(j)}$ from $T^{(j)}$ and
      a caterpillar $C$ with leaf set $\{z, x_{i_1}, \ldots, x_{i_\ell}\}$.%
      \label{fig:hat-T}}
  \end{figure}

  Property~\ref{con:sameCherries} holds because the trees in $\T^{(j)}
  \setminus (\T')^{(j)}$ have no cherries and, for every tree $T \in \T'$,
  $\hat T^{(j)}$~has the same cherries as $T^{(j)}$:
  $T^{(j)}$ is a pendant subtree of $\hat T^{(j)}$, so every cherry of
  $T^{(j)}$ is a cherry of $\hat T^{(j)}$.
  Every cherry of $\hat T^{(j)}$ that is not a cherry of $T^{(j)}$ would have
  to involve some leaf $x_{i_h}$, but none of these leaves is part of a cherry
  because $T^{(j)}$ has at least two leaves.

  To see that Property~\ref{con:noTrivialCherries} holds, observe that every
  trivial cherry $\{x, y\}$ would have to be a cherry of \emph{every} tree
  in $\hat\T^{(j)}$ because all trees in $\hat\T^{(j)}$ have the same label
  set.
  Thus, by Property~\ref{con:sameCherries}, $\{x, y\}$ would be a cherry of
  every tree $T^{(j)}$ such that $T \in \T'$.
  By the definition of $\T'$, $\{x, y\}$ would therefore be a trivial cherry of
  $\T^{(j)}$, but $\T^{(j)}$ has no trivial cherries.
  Thus, $\hat\T^{(j)}$ has no trivial cherries.

  To prove that $\hat\T^{(j)}$ has tree-child hybridization number at most $k$
  (Property~\ref{con:hybridizationNumber}), we construct a tree-child cherry
  picking sequence $\hat S$ of weight at most $k$ for $\hat\T^{(j)}$.
  This sequence is defined as
  \begin{equation*}
    \hat S = \langle (x_{j+1}, y_{j+1}), \ldots, (x_r, y_r), (x_1, x_{r+1}),
    \ldots, (x_j, x_{r+1}), (x_{r+1}, -) \rangle,
  \end{equation*}
  that is, we swap the subsequences $\langle (x_1, y_1), \ldots, (x_j, y_j)
  \rangle$ and $\langle (x_{j+1}, y_{j+1}), \ldots, (x_r, y_r) \rangle$ of
  $S$ and then replace $y_i$ with $x_{r+1}$ in each pair $(x_i, y_i)$ with
  $1 \le i \le j$.
  By construction, $\hat S$ has the same weight as $S$, that is, its weight is
  at most $k$.

  To see that $\hat S$ is a tree-child cherry picking sequence, observe that
  $\langle (x_{j+1}, y_{j+1}), \ldots, (x_r, y_r) \rangle$ is a subsequence
  of a tree-child cherry picking sequence, namely $S$, and is thus a partial
  tree-child cherry picking sequence.
  Since $S$ reduces each tree in $\T$ to the single leaf $x_{r+1}$, we have
  $x_{r+1} \notin \{x_1, \ldots, x_r\}$, so $x_{r+1}$ is not forbidden with
  respect to $\langle (x_{j+1}, y_{j+1}), \ldots, (x_r, y_r), (x_1, x_{r+1}),
  \ldots, (x_i, x_{r+1}) \rangle$, for any $i \in [j]_0$.
  Thus, $\hat S$ is a tree-child cherry picking sequence.

  It remains to prove that $\hat S$ is a cherry picking sequence for every
  tree $\hat T^{(j)} \in \hat\T^{(j)}$.
  Observe that the sequence $S' = \langle (x_{j+1}, y_{j+1}), \ldots,
  (x_r, y_r) \rangle$ reduces $T^{(j)}$ to the single leaf $x_{r+1}$.
  Thus, after applying $S'$ to~$\hat T^{(j)}$, we obtain a subtree $C'$ of
  the caterpillar $C$ with $z$ replaced with $x_{r+1}$.
  ($S'$ may also delete some leaves of $C$.)
  Since the leaves $x_{i_1}, \ldots, x_{i_\ell}$ of $C$ appear in this order
  from bottom to top in $C$, the sequence $\langle (x_1, x_{r+1}), \ldots,
  (x_j, x_{r+1}) \rangle$ reduces $C'$ to the single leaf $x_{r+1}$.
  Thus, $\hat S$ is a cherry picking sequence for~$\hat T^{(j)}$.
\end{proof}

\subsection{Proof of Theorem~\ref{thm:main}}

\label{sec:wrapping-up}

Using the results from the previous three subsections, we are now ready to
prove Theorem~\ref{thm:main}.
While our algorithm computes $\T'$ only in line~\ref{lin:PruneCherries}, and
$n'$, $k'$, and $C$ only in
lines~\ref{lin:DefineNumberRemainingLeaves}--\ref{lin:DefineCherrySet}, it is
convenient for the sake of this proof to view them as quantities that evolve
over time, as functions of $S$.
We define $n'(\T, S) = |\{x \in X \mid x \text{ is a leaf of a tree in }
  \T/S\}|$ and $k'(\T, S) = |S| - |X| + n'(\T, S)$ for any partial tree-child
cherry picking sequence~$S$.

We divide the proof of Theorem~\ref{thm:main} into three parts.
First, we prove that $k'(\T, S)$ is invariant over the course of any invocation
$\TreeChildSequence(\T, S, k)$ and that $0 \le k'(\T, S) \le k$ in every
invocation the algorithm makes.
This will be used in the analysis of the running time of the algorithm and in
proving the correctness of the algorithm in the case when it returns a sequence
in line~\ref{lin:SolutionFound}.
Then, we bound the running time of the algorithm by $O((8k)^k nt \lg t + nt \lg
nt)$, where $n = |X|$ and $t = |\T|$.
This implies in particular that the number of recursive calls the algorithm
makes is finite, a fact that will be used in the correctness proof.
Finally, we consider the tree of recursive calls the algorithm makes and
use induction on the number of descendant invocations of any invocation
$\TreeChildSequence(\T, S, k)$ to prove the correctness of this invocation.

\begin{lem}
  \label{lem:non-trivial-cherries-add-weight}
  For a collection of $X$-trees $\T$, any partial cherry picking sequence
  $S$, and any non-trivial cherry $\{x, y\}$ of $\T/S$, $k'(\T, S \circ
  \langle (x, y) \rangle) = k'(\T, S) + 1$.
\end{lem}

\begin{proof}
  Since $\{x, y\}$ is a non-trivial cherry of $\T/S$, there exists a tree
  $T/S \in \T/S$ that contains both $x$ and $y$ but not the cherry $\{x, y\}$.
  Thus, applying the pair $(x, y)$ to $\T/S$ does not remove $x$ from all trees
  in~$\T/S$.
  In particular, $n'(\T, S \circ \langle (x, y) \rangle) = n'(\T, S)$ and,
  therefore, $k'(\T, S \circ \langle (x, y) \rangle) =
  |S \circ \langle (x, y) \rangle| - |X| + n'(\T, S \circ \langle (x, y)
  \rangle) = |S| + 1 - |X| + n'(\T, S) = k'(\T, S) + 1$.
\end{proof}

\begin{lem}
  \label{lem:invariant-parameter}
  The value of $k'(\T, S)$ is invariant over the course of any invocation
  $\TreeChildSequence(\T, S, k)$ and satisfies $0 \le k'(\T, S) \le k$.
  Moreover, an invocation $\TreeChildSequence(\T, S, k)$ satisfies
  $k'(\T, S) = 0$ if and only if $S = \langle\rangle$.
\end{lem}

\begin{proof}
  First we prove that $k'(\T, S)$ does not change over the course of any
  invocation $\TreeChildSequence(\T, S, k)$.
  Note that in a given invocation $\TreeChildSequence(\T, S, k)$, $S$ changes
  only in line~\ref{lin:LoopTrivialCherryReductionEnd}.
  Each execution of line~\ref{lin:LoopTrivialCherryReductionEnd} adds a pair
  $(x, y)$ to $S$, thereby increasing $|S|$ by one.
  Since $\{x, y\}$ is a trivial cherry of $\T'$ and $y$ is not forbidden with
  respect to $S$ in this case, this also removes $x$ from all trees in
  $\T/S$, so $n'(\T, S)$ decreases by one and $k'(\T, S) = |S| - |X| + n'(\T,
  S)$ remains unchanged.

  We prove the bounds on $k'(\T, S)$ for each invocation
  $\TreeChildSequence(\T, S, k)$ by induction on $|S|$.

  If $|S| = 0$, then $S = \langle\rangle$.
  In this case, $\T/S = \T$, so $n'(\T, S) = |X|$, that is,
  $k'(\T, S) = |S| - |X| + n'(\T, S) = |S| - |X| + |X| = 0$.

  If $|S| > 0$, then $\TreeChildSequence(\T, S, k)$ is called by another
  invocation $\TreeChildSequence(\T, S', k)$ with $|S'| < |S|$.
  By the induction hypothesis, we have $k'(\T, S') \ge 0$.
  Let $S''$ be a snapshot of $S'$ in line~\ref{lin:FailureConditions} of
  the invocation $\TreeChildSequence(\T, S', k)$.
  Then $S = S'' \circ \langle (x, y) \rangle$, where $\{x, y\}$ is a
  non-trivial cherry of $\T/S''$.
  Thus, by Lemma~\ref{lem:non-trivial-cherries-add-weight}, $k'(\T, S) =
  k'(\T, S'') + 1$.
  Since $k'(\T, S'') = k'(\T, S')$, this implies that
  $k'(\T, S) > k'(\T, S') \ge 0$.
  By the second condition in line~\ref{lin:FailureConditions}, we have
  $k'(\T, S'') < k$ (because $\TreeChildSequence(\T, S', k)$ makes the
  recursive call $\TreeChildSequence(\T, S, k)$), so $k'(\T, S) =
  k'(\T, S'') + 1 \le k$.
\end{proof}

The following proposition now establishes the running time bound stated in
Theorem~\ref{thm:main}.

\begin{prop}
  \label{prop:running-time}
  The total running time of the invocation $\TreeChildSequence(\T,
  \langle\rangle, k)$ and all its descendant invocations is $O((8k)^k nt \lg t
  + nt \lg nt)$, where $n = |X|$ and $t = |\T|$.
\end{prop}

\begin{proof}
  We only provide a sketch of the argument that the algorithm's state can be
  initialized in $O(nt \lg nt)$ time and that each invocation of
  procedure~\procref{alg:tree-child-sequence}, excluding the recursive calls it
  makes, has cost $O(nt \lg t)$.
  A careful proof is straightforward but tedious.
  To prove the proposition, it then suffices to prove that the algorithm makes
  $O((8k)^k)$ recursive calls.

  Instead of computing $\T'$ from scratch as in the pseudo-code of
  procedure~$\TreeChildSequence$, we first construct the state of the top-level
  invocation $\TreeChildSequence(\T, \langle\rangle, k)$ consisting of $\T'$
  and the lists of trivial and non-trivial cherries.
  Whenever an invocation makes a recursive call, it makes a copy of its state
  to be modified by the recursive call.

  Identifying the cherries in $\T' = \T$ for the top-level invocation
  $\TreeChildSequence(\T, \langle\rangle, k)$ takes $O(nt \lg nt)$ time using
  appropriate dictionaries (e.g., balanced binary search trees) to identify
  leaves with the same labels in different trees and to collect all occurrences
  of the same cherry in different trees.

  Copying the state of the current invocation for each recursive call the
  algorithm makes takes $O(nt)$ time because the state is easily seen to have
  size $O(nt)$.
  We charge this cost to the recursive call.
  Each pair added to $S$ eliminates the corresponding cherry from up to
  $t$ trees and thereby creates up to $t$ new cherries.
  Updating $\T'$ and the lists of trivial and non-trivial cherries for each
  such cherry takes $O(\lg t)$ time, $O(t \lg t)$ time in total for each
  pair added to $S$.
  Each invocation adds at most $n$ pairs corresponding to trivial cherries
  to $S$, in line~\ref{lin:LoopTrivialCherryReductionEnd}.
  Each pair $(x, y)$ added to $S$ in line~\ref{lin:RecursiveCall} can be
  charged to the recursive call $\TreeChildSequence(\T, S \circ \langle (x, y)
  \rangle, k)$ made in line~\ref{lin:RecursiveCall}.
  Thus, each invocation adds at most one pair corresponding to a non-trivial
  cherry to $S$.
  The cost of updating $\T'$ and the list of trivial and non-trivial cherries
  in each invocation is thus $O(nt \lg t)$.
  Adding the cost of making a copy of the parent invocation's state at the
  beginning of each invocation, the cost per invocation is thus $O(nt \lg t)$.
  To obtain the time bound stated in the proposition, it remains to bound the
  number of recursive calls the algorithm makes by $O((8k)^k)$.

  Let $m_{k'}$ be the number of invocations $\TreeChildSequence(\T, S, k)$
  with $k'(\T, S) = k'$.
  By Lemma~\ref{lem:invariant-parameter}, every invocation
  $\TreeChildSequence(\T, S, k)$ the algorithm makes satisfies $0 \le k'(\T, S)
  \le k$ and the total number of invocations is therefore $\sum_{k'=0}^k
  m_{k'}$.
  Also by Lemma~\ref{lem:invariant-parameter}, there is exactly one invocation
  $\TreeChildSequence(\T, S, k)$ with $k'(\T, S) = 0$, namely the top-level
  invocation $\TreeChildSequence(\T, \langle\rangle, k)$.
  Finally, by Lemma~\ref{lem:non-trivial-cherries-add-weight}, every child
  invocation $\TreeChildSequence(\T, S_2, k)$ of an invocation
  $\TreeChildSequence(\T, S_1, k)$ satisfies $k'(\T, S_2, k) = k'(\T, S_1, k) +
  1$.
  Thus, since each invocation makes at most $8k$ recursive calls in
  line~\ref{lin:RecursiveCall}, we obtain $m_{k'+1} \le 8k \cdot m_{k'}$.
  A simple inductive argument now shows that $m_{k'} \le (8k)^{k'}$ for all
  $0 \le k' \le k$.
  Thus, the total number of recursive calls the algorithm makes
  is at most $\sum_{k'=0}^k (8k)^{k'} = \frac{(8k)^{k+1} - 1}{8k - 1} =
  O((8k)^k)$.
\end{proof}

To establish the correctness of procedure~\procref{alg:tree-child-sequence}, we
need a few simple auxiliary lemmas.

\begin{lem}
  \label{lem:k-prime-is-lower-bound}
  Let $S$ be a partial cherry picking sequence $S$ without any pairs of the
  form $(x, -)$.
  Any solution of $(\T, S)$ has weight at least $k'(\T, S)$.
\end{lem}

\begin{proof}
  Consider any cherry picking sequence $S \circ S'$ for $\T$.
  Let $X_1$ be the set of leaf labels of the trees in $\T / (S \circ S')$, and
  let $X_2$ be the subset of leaf labels of the trees in $\T / S$ that are not
  in $X_1$.
  Then $n'(\T, S) = |X_1| + |X_2|$.

  Every leaf $x \in X_2$ must be removed from the trees in $\T/S$ by at least
  one pair $(x, y) \in S'$.
  For every leaf $x \in X_1$, $S'$ must contain a pair $(x, -)$.
  Thus, $|S'| \ge |X_1| + |X_2| = n'(\T, S)$.
  Therefore, $|S \circ S'| - |X| = |S| + |S'| - |X| \ge
  |S| - |X| + n'(\T, S) = k'(\T, S)$.
\end{proof}

\begin{lem}
  \label{lem:choose-from-C}
  Let $\T$ be a collection of $X$-trees, and $S$ a partial tree-child cherry
  picking sequence such that at least one tree in $\T/S$ has more than one
  leaf.
  Then any optimal solution of $(\T, S)$ is an extension of some sequence $S
  \circ \langle (x, y) \rangle$, where $\{x, y\}$ is a cherry of\/ $\T/S$.
\end{lem}

\begin{proof}
  Consider any optimal solution $S \circ S'$ of $(\T, S)$.
  Since there exists a tree $T \in \T$ such that $T/S$ has at least two leaves,
  the first pair in $S'$ is a pair $(x, y)$ with $x, y \in X$.
  Let $S' = \langle (x, y) \rangle \circ S''$ and assume for the sake of
  contradiction that $\{x, y\}$ is not a cherry of any tree in $\T/S$.
  Then $S \circ S'' \subset S \circ S'$, so $S \circ S''$ is a tree-child
  cherry picking sequence and $|S \circ S''| < |S \circ S'|$.
  Since $\{x, y\}$ is not a cherry of any tree in $\T/S$, we have
  $T/(S \circ \langle (x, y) \rangle) = \T/S$ for all $T \in \T$.
  Thus, $T / (S \circ S'') = T / (S \circ \langle (x, y) \rangle \circ S'') =
  T / (S \circ S')$ for all $T \in \T$.
  Since $S \circ S'$ is a cherry picking sequence for~$\T$, this shows that
  $S \circ S''$ is a cherry picking sequence for $\T$, a contradiction.
\end{proof}

The following proposition now finishes the proof of Theorem~\ref{thm:main} by
proving that the invocation $\TreeChildSequence(\T, \langle\rangle, k)$ returns
a shortest tree-child cherry picking sequence for $\T$ if and only if $\T$
has a tree-child cherry picking sequence of weight at most $k$.

\begin{prop}
  \label{prop:algCorrectness}
  Given a set $\T$ of $X$-trees, a partial tree-child cherry picking sequence
  $S$, and an integer $k$, $\TreeChildSequence(\T, S, k)$ returns an optimal
  solution of $(\T, S)$ if and only if $(\T, S)$ has a solution of weight at
  most $k$.
  Otherwise, it returns {\None}.
\end{prop}

\begin{proof}
  Consider the tree of recursive calls made by the algorithm and let
  $|\TreeChildSequence(\T, S, k)|$ be the number of descendant invocations of
  the invocation $\TreeChildSequence(\T, S, k)$, including the invocation
  $\TreeChildSequence(\T, S, k)$ itself.
  By Proposition~\ref{prop:running-time}, $|\TreeChildSequence(\T, S, k)|$ is
  finite.  Thus, we can use induction on $|\TreeChildSequence(\T, S, k)|$ to
  prove the proposition.

  If $|\TreeChildSequence(\T, S, k)| = 1$, then $\TreeChildSequence(\T, S, k)$
  makes no recursive calls.
  Thus, it returns a sequence in line~\ref{lin:SolutionFound} or $\None$ in
  line~\ref{lin:FailForbiddenCherry} or~\ref{lin:FailTooManyCherries}.
  \markj{(Note that $\TreeChildSequence(\T, S, k)$ cannot reach line~\ref{lin:RecursiveSolutionEnd} without making a recursive call, as this is only possible if $|C| = 0$ or every cherry $\{x,y\}$ of some tree in $\T'$ has $x,y$ both forbidden, and these cases are covered by lines~\ref{lin:SolutionFound} and ~\ref{lin:FailForbiddenCherry} respectively.)}
  By Proposition~\ref{prop:move-up-common-cherries}, if $S_1$ is a snapshot
  of $S$ at the start of the invocation $\TreeChildSequence(\T, S, k)$ and
  $S_2$ is a snapshot of $S$ in line~\ref{lin:PruneCherries}, then $(\T, S_1)$
  has a solution of weight at most $k$ if and only if $(\T, S_2)$ has
  a solution of weight at most $k$, and any optimal solution of $(\T, S_2)$
  also is an optimal solution of $(\T, S_1)$.

  If $\TreeChildSequence(\T, S, k)$ returns {\None} in
  line~\ref{lin:FailForbiddenCherry}, then $\T / S_2$ has a cherry $\{x, y\}$
  with both $x$ and $y$ forbidden with respect to $S_2$.
  Any solution $S_2 \circ S'$ of $(\T, S_2)$ must include the pair $(x, y)$ or
  $(y, x)$ in $S'$ because otherwise the tree in $\reduce{\T}{S_2}$ that has
  $\{x, y\}$ as a cherry is not reduced to a single leaf by $S'$.
  Since both $x$ and $y$ are forbidden with respect to $S_2$, there is no such
  extension $S_2 \circ S'$ of $S_2$ that is tree-child.
  Thus, $(\T, S_2)$ has no solution, and neither does $(\T, S_1)$.
  It is therefore correct to return {\None}.

  If $\TreeChildSequence(\T, S, k)$ returns $S_2 \circ \langle (x, -) \rangle$
  in line~\ref{lin:SolutionFound}, then observe that $S_2$ is a partial
  tree-child cherry picking sequence.
  Indeed, by the assumption of the proposition, $S_1$ is a partial tree-child
  cherry picking sequence.
  For every pair $(x, y)$ added to $S$ in
  line~\ref{lin:LoopTrivialCherryReductionEnd}, $y$ is not forbidden with
  respect to $S$, so $S \circ \langle (x, y) \rangle$ is also tree-child.
  By applying this argument inductively, we conclude that $S_2$ is tree-child.

  Since $\TreeChildSequence(\T, S, k)$ returns $S_2 \circ \langle (x, -)
  \rangle$ in line~\ref{lin:SolutionFound} only if $|C| = 0$, $S_2$ reduces
  each tree in $\T$ to a single leaf.
  Since $S_2$ is tree-child, this is the same leaf $x$ for every tree $T \in
  \T$.
  Thus, $S_2 \circ \langle (x, -) \rangle$ is a solution of $(\T, S_2)$.
  Since every solution $S_2 \circ S'$ of $(\T, S_2)$ must include at
  least one pair $(z, -)$ in $S'$, $S_2 \circ \langle (x, -) \rangle$ is an
  optimal solution of $(\T, S_2)$ and, therefore, also of $(\T, S_1)$.
  Finally, by Lemma~\ref{lem:invariant-parameter}, $|S_2| - |X| + n'(\T, S_2) =
  k'(\T, S_2) \le k$; $n'(\T, S_2) = 1$ because, as just observed, each tree in
  $\T / S_2$ has $x$ as its only leaf.
  Thus, $|S_2| - |X| < k$ and $|S_2 \circ \langle (x, -) \rangle| - |X| \le k$,
  that is, $(\T, S_2)$ and $(\T, S_1)$ both have solutions of weight at most
  $k$ and returning $S_2 \circ \langle (x, -) \rangle$ is correct.

  Finally, if $\TreeChildSequence(\T, S, k)$ returns {\None} in
  line~\ref{lin:FailDugTooDeep}, then $|C| > 8k$ or $C \ne \emptyset$ and
  $k'(\T, S_2, k) \ge k$.

  If $|C| > 8k$, then $\T / S_2$ has more than $4k$ unique cherries.
  Since $\T / S_2$ has no trivial cherries,
  Proposition~\ref{prop:few-non-trivial-cherries} shows that $(\T, S_2)$ has
  no solution \markj{of weight at most $k$}, and neither does $(\T, S_1)$.
  Thus, returning {\None} is correct.

  If $C \ne \emptyset$ and $k'(\T, S_2) \ge k$, then observe that $\{x, y\}$
  is a non-trivial cherry of $\T / S_2$ for every pair $(x, y) \in C$.
  Lemma~\ref{lem:non-trivial-cherries-add-weight} shows that $k'(\T, S_2 \circ
  \langle (x, y) \rangle) = k'(\T, S_2) + 1 > k$ for all $(x, y) \in C$.
  By Lemma~\ref{lem:k-prime-is-lower-bound}, this shows that
  $(\T, S_2 \circ \langle (x, y) \rangle)$ has no solution of weight at most
  $k$ for any $(x, y) \in C$.
  By Lemma~\ref{lem:choose-from-C}, this implies that $(\T, S_2)$ has no
  solution of weight at most $k$, and neither does $(\T, S_1)$.
  Thus, returning {\None} is correct.
  This finishes the proof that every invocation $\TreeChildSequence(\T, S, k)$
  that makes no recursive calls gives a correct answer.
  
  Next consider an invocation $\TreeChildSequence(\T, S, k)$ that does make
  recursive calls.
  Then $C \ne \emptyset$.
  By Lemma~\ref{lem:choose-from-C}, $(\T, S_2)$ (and thus $(\T, S_1)$) has
  a solution of weight at most $k$ if and only if there exists a pair $(x, y)
  \in C$ such that $(\T, S_2 \circ \langle (x, y) \rangle)$ has a solution
  of weight at most $k$.
  Moreover, if such a pair exists, then one such pair has the property
  that any optimal solution of $(\T, S_2 \circ \langle (x, y) \rangle)$ also is
  an optimal solution of $(\T, S_2)$ and, thus, of $(\T, S_1)$.

  If there exists a pair $(x, y)$ such that $(\T, S_2 \circ \langle (x, y)
  \rangle)$ has a solution of weight at most $k$, then choose $(x, y)$ so
  that any optimal solution of $(\T, S_2 \circ \langle (x, y) \rangle)$ also is
  an optimal solution of $(\T, S_1)$.
  By the induction hypothesis, the invocation $\TreeChildSequence(\T, S_2 \circ
  \langle (x, y) \rangle, k)$ in line~\ref{lin:RecursiveCall} returns an
  optimal solution $S'$ of $(\T, S_2 \circ \langle (x, y) \rangle)$.
  The solution $S_{\textrm{opt}}$ of $(\T, S_1)$ returned in
  line~\ref{lin:RecursiveSolutionEnd} is no longer than $S'$.
  Since $S_{\textrm{opt}}$ is a solution of some instance $(\T, S_2 \circ
  \langle (x', y') \rangle)$ with $(x', y') \in C$, it is a solution of
  $(\T, S_2)$ and is thus an optimal solution of $(\T, S_2)$ and $(\T, S_1)$.
  Thus, the algorithm produces the correct answer.

  If there is no pair $(x, y) \in C$ such that $(\T, S_2 \circ \langle (x, y)
  \rangle)$ has a solution of weight at most $k$, then all recursive calls made
  in line~\ref{lin:RecursiveCall} of the invocation $\TreeChildSequence(\T, S,
  k)$ return \textsc{None}.
  Thus, $\TreeChildSequence(\T, S, k)$ also returns \textsc{None}.
  Since Lemma~\ref{lem:choose-from-C} shows that $(\T, S_1)$ has no solution
  of weight at most $k$ in this case, this is correct.
\end{proof}

\section{Redundant Branch Elimination: A Heuristic Improvement}

\label{sec:redundant-branch-elimination}

In this section, we discuss a method used in our implementation of procedure
\procref{alg:tree-child-sequence} to improve its running time.
We prove that it preserves the correctness of the algorithm, but we do not
know whether it provably improves the algorithm's running time.
In this sense, it is a heuristic.

The intuition behind rendundant branch elimination is the following: Suppose
that $\reduce{\T}{\langle(x,y), (z,w)\rangle}$ and $\reduce{\T}{\langle(z,w),
  (x,y)\rangle}$ result in the same set of trees.
(This can easily happen, for example, if $x,y,z,w$ are all distinct.)
Then the branch of the algorithm that starts by applying the sequence $\langle
(x, y), (z, w) \rangle$ finds a solution if and only if the branch that
starts by applying the sequence $\langle (z, w), (x, y) \rangle$ does.
So the algorithm does not need to explore this second branch; it is redundant,
and redundant branch elimination ensures that the algorithm does not make
this recursive call.

Procedure~\procref{alg:tree-child-sequence-refined} below is a modified version
of procedure~\procref{alg:tree-child-sequence} that uses redundant branch
elimination.
The only difference between procedures~\ref{alg:tree-child-sequence}
and~\ref{alg:tree-child-sequence-refined} is
that~\ref{alg:tree-child-sequence-refined} maintains a set $R$ of redundant
pairs (with $R$ set to $\emptyset$ in the top-level invocation
$\TreeChildSequencePruned(\T, \langle\rangle, k, \emptyset)$) and ignores
extensions $S \circ \langle (x, y) \rangle$ of the current sequence $S$ such
that $(x, y) \in R$.
If $\{x, y\}$ is a trivial cherry, this means that the invocation
$\TreeChildSequencePruned(\T, S, k, R)$ returns {\None}.
If $\{x, y\}$ is a non-trivial cherry, then
$\TreeChildSequencePruned(\T, S, k, R)$ does not make the recursive call
$\TreeChildSequencePruned(\T, S \circ \langle (x, y) \rangle, k, R)$.
Note that $R$ does not contain \emph{all} redundant pairs for $S$, only
a subset for which we prove below that they can safely be ignored based on the
recursive calls the algorithm has made so far.

Procedure~\procref{alg:tree-child-sequence-refined} calls a
procedure~\procref{alg:update-R} in lines~\ref{lin:ModTrivialCherryUpdateR}
and~\ref{lin:ModRecursiveCallUpdateR}.
Given a partial tree-chlid cherry picking sequence $S$, a set of pairs $R$ that
are redundant for $S$, and a pair $(x, y)$, $\UpdateR(\T, S, (x, y), R)$
returns the subset $R' \subseteq R$ containing all pairs that are redundant
also for $S \circ \langle (x, y) \rangle$.

The following definition formalizes the concept of a redundant pair.

\begin{procedure}[p]
  \caption{TCS2()$(\T, S, k, R)$\label{alg:tree-child-sequence-refined}}
  \KwIn{A collection of phylogenetic trees $\T$, a partial tree-child cherry
    picking sequence $S$, an integer $k$, and a set $R$ of redundant pairs for
    $S$}
  \KwOut{An optimal solution of $(\T, S)$ if $(\T, S)$ has a solution of
    weight at most $k$ and there do not exist a proper prefix $S_p \subset S$
    and a pair $(x, y) \in R$ such that $S_p \circ \langle (x, y) \rangle$
    dominates some optimal solution of $(\T, S)$.
    {\None} if $(\T, S)$ has no solution of weight at most $k$.
    In any other case, the output may be {\None} or a (possibly suboptimal)
    solution of $(\T, S)$.}
  \While{there exists a trivial cherry $\{x, y\}$ of $\T/S$ with $y$ not
    forbidden with respect to
    $S$\label{lin:ModLoopTrivialCherryReductionStart}}{ 
    \uIf{$(x,y) \notin R$}{
      $R \gets \UpdateR(\T, S, (x,y), R)$\;\label{lin:ModTrivialCherryUpdateR}
      $S \gets S \circ \langle (x,y) \rangle$\;\label{lin:ModTrivialCherryUpdateS}
    }\Else{
      Return $\None$\; \label{lin:ModFailRedundantTrivialCherry}
    }
  } \label{lin:ModLoopTrivialCherryReductionEnd}
  $\T' \gets \reduce{\T}{S}$\;\label{lin:ModTrivialCherriesDone}
  \uIf{$\T'$ contains a cherry $\{x, y\}$ with $x, y$ both forbidden with
    respect to $S$}{
    \Return{$\None$\;}\label{lin:ModFailForbiddenCherry}
  }\Else{
    $n' \gets |\{x \in X: x \text{ is a leaf of a tree in } \T'\}|$\;
    $k' \gets |S| - |X| + n'$\;  \label{lin:ModDefineCurrentWeight}
    $C \gets \{(x, y) \mid \text{$\{x, y\}$ is a cherry of some tree in
        $\T'$}\}$\;\label{lin:ModDefineCherrySet}
    \uIf{$|C| = 0$}{
      \Return $S \circ \langle(x, -)\rangle$, where $x$ is the last remaining
      leaf in all trees\;\label{lin:ModSolutionFound}
    }\uElseIf{$|C| > 8k$ or $k' \ge k$}{
      \Return{$\None$\;}\label{lin:ModFailTooManyCherries}\label{lin:ModFailDugTooDeep}
    }\Else{
      $S_{\textrm{opt}} \gets \None$\;
      $R' \gets R$\;\label{lin:ModInitRprime}
      \ForEach{$(x,y) \in C \setminus R$ with $y$ not forbidden with respect to
        $S$}{\label{lin:ModRecursiveSolutionStart}
        $R'' \gets \UpdateR(\T, S, (x,y), R')$\;\label{lin:ModRecursiveCallUpdateR}
        $S_{\textrm{temp}} \gets
        \TreeChildSequencePruned(\T, S \circ \langle(x, y)\rangle, k,
        R'')$\;\label{lin:ModRecursiveCall}
        \If{$w(S_{\textrm{temp}}) < w(S_{\textrm{opt}})$}{
          $S_{\textrm{opt}} \gets S_{\textrm{temp}}$\;
        }
        $R' \gets R' \cup \{(x, y)\}$\;
      }
      \Return{$S_{\textrm{opt}}$\;}\label{lin:ModRecursiveSolutionEnd}
    }
  }
\end{procedure}
\begin{procedure}[p]
  \caption{UpdateR()$(\T, S, (x,y), R)$\label{alg:update-R}}
  \KwIn{A collection of phylogenetic $X$-trees $\T$, a partial tree-child
    cherry picking sequence $S$, a pair $(x, y) \in X \times X$, a set $R$ of
    redundant pairs for $S$\;}
  \KwOut{A subset $R' \subseteq R$  of redundant pairs for $S \circ \langle{(x,
      y)}\rangle$\;}
  \Return $\{(x', y') \in R \mid x' \ne y \text{ and }
    \cc{x'}{y'}{\reduce{\T}{(S \circ \langle(x,y)\rangle)}} =
    \cc{x'}{y'}{\reduce{\T}{S}}\}$\;\label{lin:define-R-prime}
\end{procedure}

\begin{defn}\label{def:dominatingSequence}
  Let $\T$ be a set of $X$-trees, $S$ a tree-child sequence, and $(x, y) \in
  X \times X$.
  Let $\cc{x}{y}{\T/S}$ be the number of trees in $\T/S$ that have $\{x, y\}$
  as a cherry.
  An extension $S \circ S'$ of $S$ is \emph{dominated by} $S \circ
  \langle(x,y)\rangle$ if there exists an index $j > 1$ that satisfies the
  following conditions:
  \begin{itemize}[noitemsep]
    \item $(x,y)$ is the $j$th element of $S'$;
    \item $\cc{x}{y}{\reduce{\T}{S}} = \cc{x}{y}{\reduce{\T}{(S \circ
          S'_{1, j-1})}}$; and
    \item for all $(x', y') \in S_{1,j-1}'$, $y' \neq x$ and $\{x', y'\} \neq
      \{x,y\}$.
  \end{itemize}
  If a sequence $S \circ S' \circ \langle (x, y) \rangle$ is dominated by $S
  \circ \langle(x,y)\rangle$, we say that $(x,y)$ is a \emph{redundant pair}
  for $S \circ S'$.
\end{defn}

\begin{lem}
  \label{lem:DominatedSameCherryCount}
  If a sequence $S \circ S'$ is dominated by $S \circ \langle(x,y)\rangle$,
  $(x, y)$ is the $j$th pair in $S'$, and $(x, y) \notin S'_{1,j-1}$,
  then $\cc{x}{y}{\reduce{\T}{(S \circ S'_{1,i})}} =
  \cc{x}{y}{\reduce{\T}{(S \circ S'_{1,i-1})}}$ for all $i \in [j-1]$.
\end{lem}

\begin{proof}
  Let~$\T' = \reduce{\T}{S}$ and let $(x_i', y_i')$ be the $i$th pair in $S'$,
  for some $i \in [j-1]$.
  Since $\{x'_i, y'_i\} \neq \{x, y\}$, the pair $(x_i', y_i')$ does not
  eliminate the cherry $\{x, y\}$ from any tree in $\T'/S_{1,i-1}'$ that
  contains this cherry, so $\cc{x}{y}{\reduce{\T'}{S'_{1,i}}} \geq
  \cc{x}{y}{\reduce{\T'}{S'_{1,i-1}}}$.
  Since $(x, y) \in S'_{1,j}$, Observation~\ref{obs:dominate-extension} shows
  that $\cc{x}{y}{\T'} = \cc{x}{y}{\T'/S_{1,j-1}'}$.
  Thus, if $\cc{x}{y}{\reduce{\T'}{S_{1,i}'}} >
  \cc{x}{y}{\reduce{\T'}{S_{1,i-1}'}}$, then there also
  exists an index $i' \in [j-1]$ such that $\cc{x}{y}{\reduce{\T'}{S_{1,i'}'}}
  < \cc{x}{y}{\reduce{\T'}{S_{1,i'-1}'}}$, a contradiction.
  This proves that $\cc{x}{y}{\reduce{\T'}{S_{1,i}'}} =
  \cc{x}{y}{\reduce{\T'}{S_{1,i-1}'}}$ for all $i \in [j-1]$.
\end{proof}

\markj{The next observation follows immediately from Definition~\ref{def:dominatingSequence} and Lemma~\ref{lem:DominatedSameCherryCount}.}

\begin{obs}
  \label{obs:dominate-extension}
  If a sequence $S \circ S'$ is dominated by $S \circ \langle (x,y) \rangle$,
  then so is any extension of $S \circ S'$ and any prefix $S \circ S''
  \subseteq S \circ S'$ such that $(x, y) \in S''$.
\end{obs}

\begin{lem}
  \label{lem:transitive}
  Let $(x, y) \in X \times X$, and $S \circ S_1 \circ S_2 \circ S_3$ a cherry
  picking sequence.
  If $S \circ \langle (x, y) \rangle$ dominates $S \circ S_1 \circ S_2 \circ
  \langle (x, y) \rangle$ and $S \circ S_1 \circ \langle (x, y) \rangle$
  dominates $S \circ S_1 \circ S_2 \circ S_3 \circ \langle (x, y) \rangle
  \circ S'$, for some sequence $S'$, then
  $S \circ \langle (x, y) \rangle$ also dominates
  $S \circ S_1 \circ S_2 \circ S_3 \circ \langle (x, y) \rangle \circ S''$,
  for any sequence $S''$.
\end{lem}

\begin{proof}
  First assume that $|S_1| > 0$, $|S_3| > 0$, and $(x, y) \notin S_1
  \circ S_2 \circ S_3$.
  Then $(x, y)$ is the $j$th element of $S_1 \circ S_2 \circ S_3 \circ (x,y)
  \circ S''$, for $j = |S_1 \circ S_2 \circ S_3| + 1 > 1$. 
  Since $S \circ \langle (x, y) \rangle$ dominates
  $S \circ S_1 \circ S_2 \circ \langle (x, y) \rangle$ and $(x, y) \notin S_1
  \circ S_2$, we have $y' \ne x$ and $\{x, y\} \ne \{x', y'\}$ for every pair
  $(x', y') \in S_1 \circ S_2$ and
  Lemma~\ref{lem:DominatedSameCherryCount} shows that
  $\cc{x}{y}{\reduce{\T}{S}} = \cc{x}{y}{\reduce{\T}{(S \circ S_1)}} =
  \cc{x}{y}{\reduce{\T}{(S \circ S_1 \circ S_2)}}$.
  Similarly, since $S \circ S_1 \circ \langle (x, y) \rangle$ dominates
  $S \circ S_1 \circ S_2 \circ S_3 \circ \langle (x, y) \rangle \circ S'$ and
  $(x, y) \notin S_2 \circ S_3$, we have
  $y' \ne x$ and $\{x', y'\}$ for every pair $(x', y') \in S_2 \circ S_3$ and
  $\cc{x}{y}{\reduce{\T}{(S \circ S_1)}} = \cc{x}{y}{\reduce{\T}{(S \circ S_1
      \circ S_2 \circ S_3}}$.
  Together, these two observations imply that
  $\cc{x}{y}{\reduce{\T}{S}} = \cc{x}{y}{\reduce{\T}{(S \circ S_1 \circ S_2
      \circ S_3)}}$ and $y' \ne x$ and $\{x', y'\} \ne \{x, y\}$ for every
  pair $(x', y') \in S_1 \circ S_2 \circ S_3$.
  Thus, $S \circ \langle (x, y) \rangle$ dominates
  $S \circ S_1 \circ S_2 \circ S_3 \circ \langle (x, y) \rangle \circ S''$.

  If $|S_1| = 0$, then $S \circ \langle (x, y) \rangle = S \circ S_1
  \circ \langle (x, y) \rangle$ and it follows immediately that
  $S \circ \langle (x, y) \rangle$ dominates
  $S \circ S_1 \circ S_2 \circ S_3 \circ \langle (x, y) \rangle \circ S'$.
  By Observations~\ref{obs:dominate-extension}, this implies that
  $S \circ \langle (x, y) \rangle$ dominates
  $S \circ S_1 \circ S_2 \circ S_3 \circ \langle (x, y) \rangle$ and thus also
  $S \circ S_1 \circ S_2 \circ S_3 \circ \langle (x, y) \rangle \circ S''$.

  If $|S_3| = 0$, then $S \circ S_1 \circ S_2 \circ \langle (x, y) \rangle = S
  \circ S_1 \circ S_2 \circ S_3 \circ \langle (x, y) \rangle$, so it follows
  immediately that $S \circ \langle (x, y) \rangle$ dominates
  $S \circ S_1 \circ S_2 \circ S_3 \circ \langle (x, y) \rangle$.
  By Observation~\ref{obs:dominate-extension}, this implies that
  $S \circ \langle (x, y) \rangle$ also dominates
  $S \circ S_1 \circ S_2 \circ S_3 \circ \langle (x, y) \rangle \circ S''$.
  
  If $(x, y) \in S_1 \circ S_2$, then the fact that $S \circ \langle (x, y)
  \rangle$ dominates $S \circ S_1 \circ S_2 \circ \langle (x, y) \rangle$ and
  Observation~\ref{obs:dominate-extension} imply that it also
  dominates $S \circ S_1 \circ S_2$ and thus also
  $S \circ S_1 \circ S_2 \circ S_3 \circ \langle (x, y) \rangle \circ S''$.

  Finally, if $(x, y) \in S_3$, then consider the longest prefix $S_3'
  \subseteq S_3$ such that $(x, y) \notin S_3'$.
  Then, by Observation~\ref{obs:dominate-extension},
  $S \circ S_1 \circ \langle (x, y) \rangle$ dominates
  $S \circ S_1 \circ S_2 \circ S_3' \circ \langle (x, y) \rangle$.
  As shown so far, this implies that
  $S \circ \langle (x, y) \rangle$ dominates $S \circ S_1 \circ S_2 \circ S_3'
  \circ \langle (x, y) \rangle$.
  Since $S_3' \circ \langle (x, y) \rangle$ is a prefix of $S_3$ and, thus,
  of $S_3 \circ \langle (x, y) \rangle \circ S''$,
  Observation~\ref{obs:dominate-extension} now shows that
  $S \circ \langle (x, y) \rangle$ dominates $S \circ S_1 \circ S_2 \circ S_3
  \circ \langle (x, y) \rangle \circ S''$.
\end{proof}

The significance of redundant pairs stems from the following proposition.

\begin{prop}
  \label{prop:dominationCorrectness}
  Let $\T$ be a set of $X$-trees, and $S \circ S'$ a tree-child cherry picking
  sequence for $\T$.
  Suppose that $S \circ S'$ is dominated by $S \circ \langle (x,y) \rangle$,
  for some pair $(x,y) \in X \times X$.
  Then there exists a tree-child cherry picking sequence $S \circ
  \langle(x,y)\rangle \circ S''$ for $\T$ with $w(S \circ \langle(x,y)\rangle
  \circ S'') \leq w(S \circ S')$.
\end{prop}

In other words: If some branch of the algorithm already looks for an optimal
solution of $(\T, S \circ \langle (x, y) \rangle)$, then there is no need
to also look for an optimal solution of $(\T, S \circ S''')$,
for any sequence $S \circ S'''$ that is dominated by $S \circ \langle (x,
y) \rangle$.

\begin{proof}
  We can write $S' = S'' \circ \langle (x, y) \rangle \circ S'''$ such that
  $(x, y) \notin S''$.
  Let $|S''| = k$.
  For $0 \le i \le k$, let $S'_i = S''_{1,i} \circ \langle (x, y) \rangle
  \circ S''_{i+1,k} \circ S'''$.
  We prove by induction on $k - i$ that
  $S \circ S'_i$ is a tree-child cherry picking sequence for $\T$, for all $0
  \le i \le k$.
  Since $S'_0 = \langle (x, y) \rangle \circ S'' \circ S'''$ and
  $w(S \circ S'_0) = w(S \circ S')$, this proves the proposition.

  $S \circ S'_k$ is clearly a tree-child cherry picking sequence for $\T$
  because $S'_k = S'$.
  So assume that $i < k$ and that $S \circ S'_{i+1}$ is a tree-child cherry
  picking sequence for $\T$.

  Let $(x', y')$ be the $(i+1)$st pair in $S''$, that is, $(x', y')$ is
  the predecessor pair of $(x, y)$ in $S'_{i+1}$.
  Since $S \circ \langle (x, y) \rangle$ dominates $S \circ S'$, the choice
  of $S''$ implies that $y' \ne x$ and, by
  Lemma~\ref{lem:DominatedSameCherryCount},
  $\cc{x}{y}{\reduce{\T}{(S \circ S''_{1,i})}} =
  \cc{x}{y}{\reduce{\T}{(S \circ S''_{1,i+1})}}$.
  Since $S \circ S'_{i+1}$ is tree-child, the former implies that
  $S \circ S'_i$ is tree-child.
  We use the latter in the following proof that $S \circ S'_i$ is a cherry
  picking sequence for $\T$.

  Let $T \in \T$ be an arbitrary tree, let $T' = \reduce{T}{(S \circ
    S''_{1,i})}$, let $T_a = \reduce{T'}{\langle(x',y'), (x,y)\rangle}$, and
  let $T_b = \reduce{T'}{\langle(x,y), (x',y')\rangle}$.
  We show that $T_b \subseteq T_a$ and that $T_a \setminus T_b \subseteq
  \{x'\}$.
  Thus, since $S \circ S'_{i+1}$ is a tree-child cherry picking sequence and,
  therefore, $x' \ne y''$ for all $(x'', y'') \in S''_{i+2,k} \circ S'''$,
  Lemma~\ref{lem:subtrees-are-preserved} shows that $\reduce{T}{(S \circ S'_i)}
  = \reduce{T_b}{(S''_{i+2,k} \circ S''')} \subseteq \reduce{T_a}{(S''_{i+2,k}
    \circ S''')} = \reduce{T}{(S \circ S'_{i+1})}$.
  Since $\reduce{T}{(S \circ S'_{i+1})}$ has a single leaf and $\reduce{T}{(S
    \circ S'_i)}$ has at least one leaf, this shows that $\reduce{T}{(S \circ
    S'_i)} = \reduce{T}{(S \circ S'_{i+1})}$, that is,
  $S \circ S'_i$ is a cherry picking sequence for~$T$.
  Since this is true for every tree $T \in \T$,
  $S \circ S'_i$ is a cherry picking sequence for $\T$.
 
  It remains to show that $T_b \subseteq T_a$ and $T_a \setminus T_b \subseteq
  \{x'\}$.
  Since $\cc{x}{y}{\reduce{\T}{(S \circ S''_{1,i})}} =
  \cc{x}{y}{\reduce{\T}{(S \circ S''_{1,i+1})}}$,
  either both $T' = \reduce{T}{(S \circ S''_{1,i})}$ and
  $\reduce{T'}{\langle (x', y') \rangle} = \reduce{T}{(S \circ S''_{1,i+1})}$
  contain $\{x, y\}$ as a cherry or neither of them does.
  
  If neither $T'$ nor $\reduce{T'}{\langle(x',y')\rangle}$ contains $\{x, y\}$
  as a cherry, then $T_a = \reduce{T'}{\langle(x',y'), (x,y)\rangle} =
  \reduce{T'}{\langle(x',y')\rangle} \break =
  \reduce{T'}{\langle(x,y), (x',y')\rangle} = T_b$, so $T_b \subseteq T_a$
  and $T_a \setminus T_b = \emptyset \subseteq \{x'\}$.

  If both $T'$ and $\reduce{T'}{\langle(x',y')\rangle}$ contain $\{x, y\}$ as
  a cherry, then observe that $\reduce{T'}{\langle(x',y')\rangle}$ does not
  contain $\{x', y'\}$ as a cherry.
  If $T'$ also does not contain $\{x',y'\}$ as a cherry, then we have that
  $T_a = \reduce{T'}{\langle(x',y'),(x,y)\rangle} =
  \reduce{T'}{\langle(x,y)\rangle}$ and
  $T_b = \reduce{T'}{\langle(x,y),(x',y')\rangle} = \reduce{T_a}{\langle (x',
    y') \rangle}$.
  Since applying the pair $(x', y')$ to $T_a$ can only remove the leaf $x'$,
  this shows that $T_a \subseteq T_b$ and $T_a \setminus T_b \subseteq \{x'\}$.
 
  The final case is when $T'$ contains both $\{x, y\}$ and $\{x', y'\}$ as
  cherries.
  Since $\{x',y'\} \neq \{x,y\}$, $T'$~must contain distinct vertices
  $p$ and $q$ such that $p$ is the common parent of $x$ and $y$, and $q$ is the
  common parent of $x'$ and $y'$.
  It follows that  $T_b$ and $T_a$ can both be derived from $T'$ by deleting
  $x$ and $x'$ and suppressing $p$ and $q$.
  Thus, $T_a = T_b$, that is, once again, $T_b \subseteq T_a$ and $T_a
  \setminus T_b = \emptyset \subseteq \{x'\}$.
\end{proof}

While our algorithm uses redundant pairs to ignore some dominated sequences
in its search for a shortest tree-child cherry picking sequence, it cannot
ignore \emph{all} dominated sequences.
Indeed, in many cases, every possible tree-child cherry picking sequence for
$\T$ is dominated by another sequence.
Consider, for example, a binary tree on $X = \{a,b,c,d\}$ with cherries
$\{a,b\}$ and $\{c,d\}$.
Any sequence for this tree must begin with $(a,b)$, $(b,a)$, $(c,d)$ or
$(d,c)$.
If the first pair is $(a,b)$, then the second pair must be either $(c,d)$ or
$(d,c)$.
But the sequence $\langle(a,b),(c,d)\rangle$ is dominated by
$\langle(c,d)\rangle$, and similarly $\langle(a,b),(d,c)\rangle$ is dominated
by $\langle(d,c)\rangle$.
A similar argument applies to any other sequence we might try.
Thus, if we did ignore \emph{all} redundant pairs for \emph{every} sequence,
the algorithm would not find any cherry picking sequence for $\T$.
This is the reason why procedure~\procref{alg:tree-child-sequence-refined}
explicitly keeps a set $R$ of redundant pairs that are safe to ignore;
it ignores a sequence $S \circ \langle (x, y) \rangle$ \emph{only} if $(x, y)
\in R$.

Following the terminology of Linz and Semple \cite{LinzSemple2017}, we call a
pair $(x_j, y_j)$ in a partial cherry picking sequence $S = \langle (x_1, y_1),
\ldots, (x_r, y_r) \rangle$ \emph{essential} if $\T/S_{1,j} \ne \T/S_{1,j-1}$,
that is, $\{x_j, y_j\}$ is a cherry of at least one tree in $\T/S_{1,j-1}$ and,
therefore, applying the pair $(x_j, y_j)$ to $\T/S_{1,j-1}$ removes $x_j$
from at least one tree in $\T/S_{1,j-1}$.

Our correctness proof of procedure~\procref{alg:tree-child-sequence-refined}
is divided into two parts:
First we prove that if, for a given invocation $\TreeChildSequencePruned(\T, S,
k, R)$, every pair in $S$ is essential and every pair in $R$ is redundant for
$S$, then
\begin{enumerate}[label=(\roman{*}),leftmargin=*,widest=ii,noitemsep]
  \item This is true at any time during the execution of of this invocation
    (even though the invocation may modify $S$ and $R$) and
  \item For every recursive call $\TreeChildSequencePruned(\T, S'', k, R'')$
    this invocation makes, every pair in $S''$ is essential and every pair in
    $R''$ is redundant for~$S''$.
\end{enumerate}
Since the top-level invocation $\TreeChildSequencePruned(\T, \langle\rangle, k,
\emptyset)$ satisfies $S = \langle\rangle$ and $R = \emptyset$, that is, all
pairs in $S$ are trivially essential and all pairs in $R$ are trivially
redundant for $S$, an inductive argument then implies that every pair in $S$
is essential and every pair in $R$ is redundant for $S$ at any time during the
execution of any invocation $\TreeChildSequencePruned(\T, S, k, R)$.
The second part of the proof shows that, under this condition, the invocation
$\TreeChildSequencePruned(\T, \langle\rangle, k, \emptyset)$ returns a
shortest tree-child cherry picking sequence for $\T$ if this sequence has
weight at most $k$; otherwise, it returns {\None}.

The following lemma shows that replacing $R$ with the set returned by
$\UpdateR(\T, S, (x, y), R)$ whenever we append a pair $(x, y)$ to a
sequence $S$ maintains the property that every pair in $R$ is redundant for
$S$.

\begin{lem}
  \label{lem:oneStepDominated}
  Let $S \circ \langle (x, y) \rangle$ be a partial tree-child cherry picking
  sequence whose pairs are all essential, and let $R \subseteq X \times X$.
  For every pair $(x', y')$ in the subset $R' \subseteq R$ returned by
  $\UpdateR(\T, S, (x,y), R)$, the sequence $S \circ \langle (x, y), (x', y')
  \rangle$ is dominated by $S \circ \langle (x', y') \rangle$.
\end{lem}

\begin{proof}
  By the definition of $R'$ in line~\ref{lin:define-R-prime} of
  procedure~\procref{alg:update-R}, we have $x' \ne y$ and
  $\cc{x'}{y'}{\reduce{\T}{S}} = \cc{x'}{y'}{\reduce{\T}{(S \circ \langle (x,y)
      \rangle)}}$ for all $(x', y') \in R'$.
  Observe also that $\{x, y\} \neq \{x', y'\}$.
  Indeed, since every pair in $S \circ \langle (x, y) \rangle$ is essential,
  there exists a tree in $\reduce{\T}{S}$ that has $\{x, y\}$ as a cherry,
  while there is no tree in $\reduce{\T}{(S \circ \langle (x, y) \rangle)}$ that
  has $\{x, y\}$ as a cherry.
  Thus, if $\{x, y\} = \{x', y'\}$, we would have
  $\cc{x'}{y'}{\reduce{\T}{S}} \ne \cc{x'}{y'}{\reduce{\T}{(S \circ \langle
      (x,y) \rangle)}}$, so $(x', y') \notin R'$.
  Since $(x', y')$ is not the first pair in $\langle (x, y), (x', y') \rangle$,
  the sequence $S \circ \langle (x, y), (x', y') \rangle$ is therefore
  dominated by $S \circ \langle (x', y') \rangle$.
\end{proof}

We are now ready to prove Claims~(i) and~(ii) above.
Since each invocation $\TreeChildSequencePruned(\T, S, k, R)$ may modify
$S$ and $R$, we use $S_0$ and $R_0$ to refer to the values of $S$ and $R$
passed as arguments to this invocation, and $S$ and $R$ to refer to the current
values of $S$ and $R$ at any point during the execution of
$\TreeChildSequencePruned(\T, S, k, R)$.

\begin{lem}
  \label{lem:maintain-redundant-pairs}
  Consider any invocation $\TreeChildSequencePruned(\T, S_0, k, R_0)$ such that
  every pair in $S_0$ is essential and every pair in $R_0$ is redundant for
  $S_0$.
  Then
  \begin{enumerate}[label=(\roman{*}),leftmargin=*,widest=ii,noitemsep]
    \item At any time during the execution of this invocation,
      every pair in $S$ is essential and there exists a proper prefix $S_p
      \subset S_0$ for each pair $(x', y') \in R$ such that $S_p \circ \langle
      (x', y') \rangle$ dominates $S \circ \langle (x', y') \rangle$; and
    \item For every recursive call $\TreeChildSequencePruned(\T, S'', k, R'')$
      this invocation makes, every pair in $S''$ is essential and every pair in
      $R''$ is redundant for $S''$.
  \end{enumerate}
\end{lem}

\begin{proof}
  (i) Initially, we have $S = S_0$ and $R = R_0$.
  Thus, since every pair in $S_0$ is essential and every pair in $R_0$ is
  redundant for $S_0$, (i) holds for this choice of $S$ and $R$.
  Next we prove that any modification the invocation makes to $S$ and $R$
  maintains (i).
  Observe that $\TreeChildSequencePruned(\T, S_0, k, R_0)$ modifies $S$ and
  $R$ only in lines~\ref{lin:ModTrivialCherryUpdateR}
  and~\ref{lin:ModTrivialCherryUpdateS}.
  Consider one iteration of the loop in
  lines~\ref{lin:ModLoopTrivialCherryReductionStart}--\ref{lin:ModLoopTrivialCherryReductionEnd} and let $(x, y)$ be the pair added to $S$ in this iteration.
  Since $\{x, y\}$ is a trivial cherry of $\T/S$ in this case and every pair
  in $S$ essential, every pair in $S \circ \langle (x, y) \rangle$ is
  essential.
  By Lemma~\ref{lem:oneStepDominated}, every pair $(x', y')$ in the set
  $R'$ returned by $\UpdateR(\T, S, (x, y), R)$ in
  line~\ref{lin:ModTrivialCherryUpdateR} has the property that
  $S \circ \langle (x', y') \rangle$ dominates $S \circ \langle (x, y), (x',
  y') \rangle$.
  Since $R' \subseteq R$, there exists a proper prefix $S_p \subset S_0$
  such that $S_p \circ \langle (x', y') \rangle$ dominates
  $S \circ \langle (x', y') \rangle$.
  Thus, by Lemma~\ref{lem:transitive}, $S_p \circ \langle (x', y')
  \rangle$ also dominates $S \circ \langle (x, y), (x', y') \rangle$
  \markj{(where $S$ and $S \circ S_1$ in Lemma~\ref{lem:transitive} correspond to $S_p$ and $S$ respectively, $S_2 = \langle \rangle$, and $S_3 = \langle(x,y)\rangle$)}
  .
  Therefore, replacing $S$ with $S \circ \langle (x, y) \rangle$, and $R$ with
  the set returned by $\UpdateR(\T, S, (x, y), R)$ maintains that every
  pair in $S$ is essential and, for every every pair $(x', y') \in R$, there
  exists a proper prefix $S_p \subset S_0$ such that $S_p \circ \langle (x',
  y') \rangle$ dominates $S \circ \langle (x', y') \rangle$.

  (ii) Consider any recursive call
  $\TreeChildSequencePruned(\T, S \circ \langle (x, y) \rangle, k, R'')$ the
  invocation $\TreeChildSequencePruned(\T, S, k, R)$ makes in
  line~\ref{lin:ModRecursiveCall}.
  By (i), all pairs in $S$ are essential.
  Since $(x, y) \in C$, $\{x, y\}$ is a cherry of $\T/S$.
  Thus, every pair in $S \circ \langle (x, y) \rangle$ is essential.
  By Lemma~\ref{lem:oneStepDominated}, the set $R''$ returned by
  $\UpdateR(\T, S, (x, y), R')$ in line~\ref{lin:ModRecursiveCallUpdateR}
  contains only pairs that are redundant for $S \circ \langle (x, y) \rangle$.
  Thus, (ii) holds.
\end{proof}

The following corollary follows by applying
Lemma~\ref{lem:maintain-redundant-pairs} inductively after observing that
$S_0 = \langle\rangle$ and $R_0 = \emptyset$ for the top-level invocation
$\TreeChildSequencePruned(\T, \langle\rangle, k, \emptyset)$.

\begin{cor}
  \label{cor:correct-redundant-pairs}
  At any point during the execution of an invocation
  $\TreeChildSequencePruned(\T, S_0, k, R_0)$, there exists a proper
  prefix $S_p \subset S_0$ for each pair $(x', y') \in R$ such that
  $S_p \circ \langle (x', y') \rangle$ dominates $S \circ \langle (x', y')
  \rangle$.
\end{cor}

The next lemma states the fairly weak correctness guarantee each
invocation $\TreeChildSequencePruned(\T, S_0, k, R_0)$ provides.
As we show below, in Corollary~\ref{cor:sequenceAlgCorrectnessRefined}, this
lemma implies that the invocation $\TreeChildSequencePruned(\T, \langle\rangle,
k, \emptyset)$ returns a shortest tree-child cherry picking sequence for $\T$
if there is such a sequence of weight at most $k$.

\begin{lem}
  \label{lem:sequenceAlgCorrectnessRefined}
  Consider any invocation $\TreeChildSequencePruned(\T, S_0, k, R_0)$ the
  algorithm makes.
  If $(\T, S_0)$ has a solution of weight at most $k$, then either
  $\TreeChildSequencePruned(\T, S_0, k, R_0)$ returns an optimal solution
  of $(\T, S_0)$ or there exist an extension $S_0 \circ S'$ of $S_0$, a pair
  $(x, y) \in R_0$, and a proper prefix $S_p \subset S_0$ such that
  $S_p \circ \langle (x, y) \rangle$ dominates $S_0 \circ S'$.
\end{lem}

\begin{proof}
  Since no invocation $\TreeChildSequencePruned(\T, S, k, R)$ makes more
  recursive calls than the corresponding invocation $\TreeChildSequence(\T, S,
  k)$, Proposition~\ref{prop:running-time} shows that each invocation
  $\TreeChildSequencePruned(\T, S, k, R)$ has a finite number of descendant
  invocations, which we denote by $|\TreeChildSequencePruned(\T, S, k, R)|$.
  Thus, if the lemma does not hold, we can choose an invocation
  $\TreeChildSequencePruned(\T, S_0, k, R_0)$ that violates the lemma and
  has the minimum number of descendant invocations
  $|\TreeChildSequencePruned(\T, S_0, k, R_0)|$ among all such invocations.

  Since $\TreeChildSequencePruned(\T, S_0, k, R_0)$ fails to find an optimal
  solution of $(\T, S_0)$, $\TreeChildSequencePruned(\T, S_0, k, R_0)$ returns
  {\None} in line~\ref{lin:ModFailRedundantTrivialCherry},
  \ref{lin:ModFailForbiddenCherry}, \ref{lin:ModFailTooManyCherries}
  or~\ref{lin:ModRecursiveSolutionEnd}, or it returns a suboptimal solution of
  $(\T, S_0)$ in line~\ref{lin:ModSolutionFound}
  or~\ref{lin:ModRecursiveSolutionEnd}.
  Next we consider these different cases:

  \begin{description}
    \item[\boldmath$\textrm{TCS2}(\T, S_0, k, R_0)$ returns {\None}
      in line~\ref{lin:ModFailForbiddenCherry}
      or~\ref{lin:ModFailTooManyCherries}:]
      In this case, $\TreeChildSequence(\T, S_0, k)$ would have returned
      {\None} in line~\ref{lin:FailForbiddenCherry}
      or~\ref{lin:FailTooManyCherries}.
      Thus, by Proposition~\ref{prop:algCorrectness}, $(\T, S_0)$ has no
      solution of weight at most $k$, a contradiction.
    \item[\boldmath$\textrm{TCS2}(\T, S_0, k, R_0)$ returns a sequence
      $S_0 \circ S'$ in line~\ref{lin:ModSolutionFound}:]
      In this case, $\TreeChildSequence(\T, S_0, k)$ would have returned the
      same sequence in line~\ref{lin:SolutionFound}.
      Thus, by Proposition~\ref{prop:algCorrectness}, $S_0 \circ S'$ is an
      optimal solution of $(\T, S_0)$, a contradiction.
    \item[\boldmath$\textrm{TCS2}(\T, S_0, k, R_0)$ returns {\None} in line~6:]
      In this case, consider the contents of $S$ and $R$ immediately before
      $\TreeChildSequencePruned(\T, S_0, k, R_0)$ returns.
      There exists a trivial cherry $\{x, y\}$ of $\reduce{\T}{S}$ such that
      $y$ is not forbidden with respect to $S$ and $(x, y) \in R$.
      Since $(\T, S_0)$ has a solution of weight at most $k$,
      Proposition~\ref{prop:move-up-common-cherries} shows that
      $(\T, S \circ \langle (x, y) \rangle)$ also has a solution of
      weight at most $k$ and any optimal solution of $(\T, S \circ \langle
      (x, y) \rangle)$ is also an optimal solution of $(\T, S)$.
      By Corollary~\ref{cor:correct-redundant-pairs}, there exists a proper
      prefix $S_p \subseteq S_0$ such that $S_p \circ \langle (x, y) \rangle$
      dominates $S \circ \langle (x, y) \rangle$ and, thus, by
      Observation~\ref{obs:dominate-extension},
      $S \circ \langle (x, y) \rangle \circ S'$, for any optimal solution
      $S \circ \langle (x, y) \rangle \circ S'$ of
      $(\T, S \circ \langle (x, y) \rangle)$, a contradiction.
    \item[\boldmath$\textrm{TCS2}(\T, S_0, k, R_0)$ returns {\None} or a
      suboptimal solution in line~\ref{lin:ModRecursiveSolutionEnd}:]
      In this case, the corresponding invocation $\TreeChildSequence(\T, S_0,
      k)$ would have reached line~\ref{lin:RecursiveSolutionEnd}.
      Since $(\T, S_0)$ has a solution of weight at most~$k$,
      Proposition~\ref{prop:algCorrectness} shows that $\TreeChildSequence(\T,
      S_0, k)$ would have returned an optimal solution $S_0 \circ S'$ of $(\T,
      S_0)$.
      This solution satisfies $S_0 \circ S' = S \circ \langle (x, y) \rangle
      \circ S''$, for some pair $(x, y) \in C$, referring to the state of $S$ in
      line~\ref{lin:PruneCherries} of $\TreeChildSequence(\T, S, k)$.
      This shows that there exists a pair $(x, y) \in C$ such that $(\T, S
      \circ \langle (x, y) \rangle)$ has a solution of weight at most $k$
      and any optimal solution of $(\T, S \circ \langle (x, y) \rangle)$
      is also an optimal solution of $(\T, S_0)$.

      Now consider the subset $C_{\textrm{opt}} \subseteq C$ of all pairs
      $(x, y)$ such that $(\T, S \circ \langle (x, y) \rangle)$ has a solution
      of weight at most $k$ and any optimal solution of $(\T, S \circ \langle
      (x, y) \rangle)$ is an optimal solution of $(\T, S_0)$.
      Order the pairs in $C_{\textrm{opt}}$ so that the pairs in
      $C_{\textrm{opt}} \setminus R$ precede the pairs in $C_{\textrm{opt}}
      \cap R$, and the pairs in $C_{\textrm{opt}} \setminus R$ are arranged in
      the order in which $\TreeChildSequencePruned(\T, S_0, k, R_0)$ makes the
      corresponding recursive calls $\TreeChildSequencePruned(\T, S \circ
      \langle (x, y) \rangle, R'')$.
      If for a pair $(x, y) \in C_{\textrm{opt}}$,
      $\TreeChildSequencePruned(\T, S_0, k, R_0)$ makes the recursive call
      $\TreeChildSequencePruned(\T, S \circ \langle (x, y) \rangle, R'')$ and
      this recursive call returns an optimal solution $S \circ \langle (x, y)
      \rangle \circ S''$ of $(\T, S \circ \langle (x, y))$, then
      $\TreeChildSequencePruned(\T, S_0, k, R_0)$ returns a solution
      $S_0 \circ S'$ of $(\T, S_0)$ that is no longer than $S \circ \langle (x,
      y) \rangle \circ S''$.
      By the choice of $C_{\textrm{opt}}$, $S_0 \circ S'$ is thus an
      optimal solution of $(\T, S_0)$.
      Since we assume that $\TreeChildSequencePruned(\T, S_0, k, R_0)$ does not
      return an optimal solution of $(\T, S_0)$, it follows that for each pair
      $(x, y) \in C_{\textrm{opt}}$, either $\TreeChildSequencePruned(\T, S_0,
      k, R_0)$ does not make the recursive call
      $\TreeChildSequencePruned(\T, S \circ \langle (x, y) \rangle, k, R'')$
      (that is, $(x, y) \in C_{\textrm{opt}} \cap R$) or it makes this
      recursive call (that is, $(x, y) \in C_{\textrm{opt}} \setminus R$) but
      the recursive call returns {\None} or a suboptimal solution of
      $(\T, S \circ \langle (x, y) \rangle)$.

      Now let $(x, y)$ be the first pair in $C_{\textrm{opt}}$ according to the
      ordering defined above.

      \begin{itemize}
        \item If $\TreeChildSequencePruned(\T, S_0, k, R_0)$ does not make the
          recursive call $\TreeChildSequencePruned(\T, S \circ \langle (x, y)
          \rangle, k, R'')$, then $(x, y) \in R$.
          Thus, by Corollary~\ref{cor:correct-redundant-pairs}, there exists a
          proper prefix $S_p \subset S_0$ such that $S_p \circ \langle (x, y)
          \rangle$ dominates $S \circ \langle (x, y) \rangle$.
          Since $S \circ \langle (x, y) \rangle$ is an extension of $S_0$, this
          is a contradiction.
        \item If $\TreeChildSequencePruned(\T, S_0, k, R_0)$ does make the
          recursive call $\TreeChildSequencePruned(\T, S \circ \langle (x, y)
          \rangle, k, R'')$, then $\TreeChildSequencePruned(\T, S \circ \langle
          (x, y) \rangle, k, R'')$ does not return an optimal solution of
          $(\T, S \circ \langle (x, y) \rangle)$.
          Thus, since $|\TreeChildSequencePruned(\T, S \circ \langle (x, y)
          \rangle, k, R'')| < |\TreeChildSequencePruned(\T, S_0, k, R_0)|$, the
          choice of $\TreeChildSequencePruned(\T, S_0, k, R_0)$ implies that
          there exist an extension $S \circ \langle (x, y) \rangle \circ S'$ of
          $S \circ \langle (x, y) \rangle$, a prefix $S_p \subseteq S$, and a
          pair $(x', y') \in R''$ such that $S_p \circ \langle (x', y')
          \rangle$ dominates $S \circ \langle (x, y) \rangle \circ S'$.
          Now we distinguish two cases.

          \begin{itemize}
            \item If $(x', y') \in R$, we prove that there exists a proper
              prefix $S_p' \subset S_0$ such that $S_p' \circ \langle (x', y')
              \rangle$ dominates $S \circ \langle (x, y) \rangle \circ S'$.
              Since $S_0 \subseteq S \circ \langle (x, y) \rangle \circ S'$ and
              $R \subseteq R_0$, this implies that
              $\TreeChildSequencePruned(\T, S_0, k, R_0)$ does not violate the
              lemma, a contradiction.
              If $S_p \subset S_0$, we can set $S_p' = S_p$.
              So assume that $S_0 \subseteq S_p \subseteq S$.
              Since $(x', y') \in R$,
              Cororally~\ref{cor:correct-redundant-pairs} shows that there
              exists a proper prefix $S_p' \subset S_0$ such that $S_p' \circ
              \langle (x', y') \rangle$ dominates $S \circ \langle (x', y')
              \rangle$.
              By Lemma~\ref{lem:transitive}, $S_p' \circ \langle (x', y')
              \rangle$ also dominates $S \circ \langle (x, y) \rangle \circ
              S'$
              \markj{(where $(x,y)$ in Lemma~\ref{lem:transitive} corresponds to $(x',y')$, and $S, S\circ S_1, S\circ S_1 \circ S_2, S\circ S_1 \circ S_2 \circ S_3 \circ \langle(x,y) \rangle \circ S''$ correspond to $S_p', S_0, S_p, S \circ \langle(x,y)\rangle \circ S'$ respectively).}

            \item If $(x', y') \notin R$, then $(x', y') \in R' \setminus R$,
              which implies that $(x', y') \in C \setminus R$ and, therefore,
              $\TreeChildSequencePruned(\T, S_0, k, R_0)$ makes a recursive
              call $\TreeChildSequencePruned(\T, S \circ \langle (x', y')
              \rangle, k, R'')$ before the recursive call
              $\TreeChildSequencePruned(\T, S \circ \langle (x, y) \rangle, k,
              R'')$.
              Since $S \circ \langle (x', y') \rangle$ dominates $S \circ
              \langle (x, y) \rangle \circ S'$,
              Proposition~\ref{prop:dominationCorrectness} shows that there
              exists a solution $S \circ \langle (x', y') \rangle \circ S''$ of
              $(\T, S \circ \langle (x', y') \rangle)$ that satisfies $w(S
              \circ \langle (x', y') \rangle \circ S'') \le w(S \circ \langle
              (x, y) \rangle \circ S')$.
              Since $S \circ \langle (x, y) \rangle \circ S'$ is an optimal
              solution of $(\T, S_0)$, this implies that $S \circ \langle (x',
              y') \rangle \circ S''$ is also an optimal solution of
              $(\T, S_0)$.
              Thus, $(x', y') \in C_{\textrm{opt}}$, a contradiction because
              $(x, y)$ is the first pair in $C_{\textrm{opt}}$.\qedhere
          \end{itemize}
      \end{itemize}
  \end{description}
\end{proof}

\begin{cor}
  \label{cor:sequenceAlgCorrectnessRefined}
  The invocation $\TreeChildSequencePruned(\T, \langle\rangle, k, \emptyset)$
  returns a shortest tree-child cherry picking sequence for $\T$ if there
  exists such a sequence of weight at most $k$.
  Otherwise, $\TreeChildSequencePruned(\T, \langle\rangle, k, \emptyset)$
  returns {\None}.
\end{cor}

\begin{proof}
  If there is no tree-child cherry picking sequence for $\T$ of weight at most
  $k$, then Proposition~\ref{prop:algCorrectness} shows that the invocation
  $\TreeChildSequence(\T, \langle\rangle, k)$ returns {\None}.
  Since each invocation $\TreeChildSequencePruned(\T, S, k, R)$ 
  is easily seen to return a sequence only if $\TreeChildSequence(\T, S, k)$
  returns a sequence, this implies that
  $\TreeChildSequence(\T, \langle\rangle, k, \emptyset)$ returns {\None}
  if there is no tree-child cherry picking sequence of weight at most $k$.

  So assume that there exists a tree-child cherry picking sequence for $\T$
  of weight at most $k$.
  If $\TreeChildSequencePruned(\T, \langle\rangle, k, \emptyset)$ does not
  return a shortest tree-child cherry picking sequence for $\T$, then
  Lemma~\ref{lem:sequenceAlgCorrectnessRefined} states that there exists an
  extension $S$ of $\langle\rangle$, proper prefix $S_p \subseteq
  \langle\rangle$, and a pair $(x, y) \in \emptyset$ such that $S_p \circ
  \langle (x, y) \rangle$ dominates $S$.
  However, neither $S_p$ nor the pair $(x, y)$ can exist.
  Thus,
  $\TreeChildSequencePruned(\T, \langle\rangle, k, \emptyset)$
  returns a shortest tree-child cherry picking sequence for $\T$.
\end{proof}

As already observed in the proof of
Lemma~\ref{lem:sequenceAlgCorrectnessRefined}, each invocation
$\TreeChildSequencePruned(\T, S, k, R)$ makes at most as many recursive calls
as its corresponding invocation $\TreeChildSequence(\T, S, k)$, so the total
number of recursive calls made by the algorithm is still bounded by
$O((8k)^k)$.
Using standard techniques, including binary search trees and integer sorting,
and a careful implementation of
lines~\ref{lin:ModLoopTrivialCherryReductionStart}--\ref{lin:ModLoopTrivialCherryReductionEnd}
that avoids calling $\UpdateR$ in each iteration,
it is possible to show that the cost per recursive call remains
$O(nt \lg t)$, including the cost to query and maintain $R$.
Thus, the worst-case running time of the algorithm remains $O((8k)^k nt \lg t +
nt \lg nt)$.
Since we are interested in using redundant branch elimination mainly as a
heuristic improvement of the running time of the algorithm in practice, we do
not prove this here.
Note that redundant branch elimination is a heuristic only as far as improving
the running time is concerned;
Corollary~\ref{cor:sequenceAlgCorrectnessRefined} above shows that it
preserves the algorithm's correctness.

\section{Implementation and Experiments}

\label{sec:experiments}

In order to evaluate the usefulness of the algorithm presented in this paper,
we implemented it and ran experiments on synthetic and realistic inputs to
answer the following questions:
\begin{itemize}[noitemsep]
  \item How difficult inputs can our algorithm handle, both in terms of
    the number of reticulations in the computed network and the number of
    trees in the input?
  \item How does the running time of our algorithm compare to that of its
    closest competitor, \textsc{HybroScale}?
\end{itemize}
The answer to this second question is that, for inputs with at least 3 trees,
our algorithm runs significantly faster than \textsc{HybroScale}.
Since \textsc{HybroScale} computes optimal hybridization networks, without any
restrictions on their structure, while our algorithm computes optimal
\emph{tree-child} networks, we effectively buy this faster running time at the
price of restricting the types of outputs we can compute and, consequently,
possibly missing some optimal networks that are not tree-child.
This raises the following natural question:
\begin{itemize}
  \item For inputs for which both our algorithm and \textsc{HybroScale} were
    able to compute a network, by how much did the reticulation numbers of the
    computed networks differ?
\end{itemize}

The discussion of our experimental results is divided into the following
subsections:
Section~\ref{sec:environment} discusses the hardware and software environment
on which we ran our experiments, as well as some high-level characteristics
of our implementation.
The complete source code, test data, and the programs we used to prepare
the test data are available from
\verb|https://github.com/nzeh/tree_child_code|, including detailed
documentation.
Section~\ref{sec:data} describes the data sets used in our experiments.
Section~\ref{sec:parameters} briefly discusses the tuning parameters of our
implementation used throughout our experiments.
Section~\ref{sec:results} discusses our experimental results.

\subsection{Evaluation Environment and Some Implementation Details}

\label{sec:environment}

Our evaluation platform was a Linux system with a quad-core Intel Xeon W3570
running at 1.7GHz and 24GB of DDR3 RAM clocked at 1333MHz.
The operating system was Debian GNU/Linux 9 with a 4.19.46-64 Linux kernel.
Our code for computing a tree-child network was implemented in Rust version
1.27.0.
Hybroscale was implemented in Java, and we used Java version 1.8.0\_161 to run it.

Our code implements procedure~\procref{alg:tree-child-sequence-refined}, that
is, it uses redundant branch optimization.
It also uses a number of additional optimizations:
\begin{description}
  \item[Check for redundant pairs using occurrence counts:] The check for
    redundant pairs (pairs in $R$) was implemented by recording for each cherry
    $\{x, y\}$ of $\reduce{\T}{S}$ how many trees contained the cherry $\{x,
      y\}$ the last time an ancestor invocation made a recursive call
    $\TreeChildSequencePruned(\T, S \circ \langle (x, y) \rangle, k, R)$ or
    $\TreeChildSequencePruned(\T, S \circ \langle (x, y) \rangle, k, R)$.
    It is easy to verify that $(x, y)$ is redundant for the current sequence
    $S$ if and only if the number of trees in $\reduce{\T}{S}$ that contain
    the cherry $\{x, y\}$ is the same as the number of trees that contained
    the cherry $\{x, y\}$ the last time a recursive call
    $\TreeChildSequencePruned(\T, S \circ \langle (x, y) \rangle, k, R)$
    was made.
  \item[No copying of an invocation's state for each recursive call:]
    The state of each invocation (current set of trees, set of trivial
    cherries, set of non-trivial cherries, partial tree-child cherry picking
    sequence, and information about the cherries and trees containing each
    leaf) is fairly large.
    To avoid the overhead of copying this state for each recursive call,
    each recursive call instead modifies its parent invocation's state without
    making a copy.
    These modifications are recorded in a log and are undone when the recursive
    call returns, thereby restoring the parent invocation's state.
  \item[\boldmath Search for the optimal $k$:] The search for an optimal
    tree-child cherry picking sequence calls the procedure
    $\TreeChildSequencePruned(\T, \langle\rangle, k, \emptyset)$ with
    increasing values of $k$ until it reports success.
    This guarantees that the parameter $k$ is no larger than the tree-child
    hybridization number of each input.
  \item[Parallelization:] The different branches of the recursive search for an
    optimal tree-child cherry picking sequence are clearly independent and can
    thus be assigned to different threads of a parallel implementation of
    procedure~\procref{alg:tree-child-sequence-refined}.
    One challenge is that, especially in the presence of redundant branch
    elimination, the computational costs of different branches can differ
    substantially.

    To balance the load between threads, we implemented a work sharing
    scheduler that allows idle threads to send messages to busy threads to
    request part of their workload.
    In response to such a request, the busy thread sends a branch on its
    recursion stack that is yet to be explored to the requesting thread.
    In the interest of minimizing the number of messages exchanged between
    threads, the busy thread always shares the next branch from the
    \emph{bottom} of its recursion stack, hopefully corresponding to a large
    subtree in the algorithm's recursion.

    The communication protocol was implemented using light-weight spinlocks
    to minimize the amount of time busy threads spend on communicating with
    other threads.
  \item[Cluster reduction:] Cluster reduction
    \cite{linz-binary-clustering,baroni-binary-clustering} has been observed
    to be the most important optimization in phylogenetic network construction
    methods for pairs of trees \cite{Li2017Clustering}.
    While we expect cluster reduction to be less effective for more than two
    trees, our implementation still applies cluster reduction because it is
    relatively cheap and should still have a significant impact on the
    algorithm's running time for real-world inputs.
\end{description}

In order to complete all our experiments in a reasonable amount of time, we
limited every run of our algorithm or of \textsc{HybroScale} to 60 minutes.
If the algorithm did not produce a result within this time limit, we consider
this input to be unsolvable by the algorithm in the context of this evaluation.

\subsection{Test Data}

\label{sec:data}

We used synthetic and real-world data for the performance evaluation of our
algorithm.

\subsubsection{Synthetic Data}

To generate a test instance with $t$ trees over a set of $n$ leaves and with
tree-child hybridization number close to $k$, we generated a random
tree-child network $N$ on $n$ leaves and with $k$ reticulations.
Then we extracted a random set of $t$ trees displayed by $N$.

\paragraph{Network generation.}

To generate the network $N$, we initialized $N$ to be a tree with two leaves.
A network with $n$ leaves and $k$ reticulations can then be obtained by adding
$s_r = n + k - 2$ tree nodes and $k_r = k$ reticulations to $N$.
The total number of non-leaf nodes to be added is $s_r + k_r$.
Thus, as long as $s_r > 0$ and $k_r > 0$, we added either a tree node or a
reticulation.

To add a tree node, we chose an existing leaf $u$ and added two new leaves $v$
and $w$ with parent $u$.
This turns $u$ into a tree node while not affecting any existing reticulations
or tree nodes.
Thus, $s_r$ decreases by one while $k_r$ remains unchanged.

To add a reticulation, we choose two leaves $u$ and $v$; merge $v$ into $u$,
making $u$ and $v$ the same node; and then add a new leaf $w$ with parent $u$.
This turns $u$ into a reticulation while not affecting any existing
reticulations or tree nodes.
Thus, $k_r$ decreases by one while $s_r$ remains unchanged.

In order to ensure that the network is tree-child, the two nodes $u$ and $v$
to be merged are chosen from the set $M$ of all nodes whose parents and
siblings are tree nodes.
We also ensure that the network has no parallel edges by picking $u$ and $v$
so that they have different parents.
Thus, if $|M| = 1$ or $|M| = 2$ and the two nodes in $M$ have the same parent,
then there exist no two nodes $u$ and $v$ that can be added while keeping the
network tree-child and not introducing any parallel edges.
In this case, we add a new tree node.
If it is possible to add a reticulation node, then we add a tree node
with probability $\frac{s_r}{s_r + k_r}$ and a reticulation with probability
$\frac{k_r}{s_r + k_r}$.

If we add a tree node, we choose the leaf $u$ to be turned into a tree node
uniformly at random from the current set of leaves.

If we add a reticulation, we choose $u$ and $v$ uniformly at random from the
set $M$.  If the two chosen nodes $u$ and $v$ have the same parent, we repeat
this selection process until they do not.

This random addition of tree nodes and reticulations continues until $s_r = 0$
or $k_r = 0$.
If $k_r = 0$ and $s_r > 0$, we keep adding tree nodes using the procedure above
until $s_r = 0$.
If $s_r = 0$ and $k_r > 0$, we keep adding reticulations using the procedure
above until either $k_r = 0$ or it is impossible to add more reticulation
because either $|M| = 1$ or $|M| = 2$ and the two leaves in $M$ have the same
parent.

\paragraph{Tree generation.}

We select $t$ \yuki{(or fewer)} trees displayed by $N$ by repeating the following process:
Delete one of the parent edges of each reticulation in $N$ uniformly at random
and suppress every node with only one child in the resulting tree.
\yuki{If the newly generated tree already exists within the list of trees (with the same Newick representation) then we do not add it to the list. We maintain a count on the number of times this occurs. Once this count reaches~$100$ or if we have~$t$ trees in our list then we terminate the process and return the trees.}

Note that the set of trees we obtain using this process may have tree-child
hybridization number less than $k$.
First, the network generation does not guarantee that we obtain a
network with $k$ reticulations if we stop the network generation with a value
of $k_r > 0$ and without any pairs of leaves that can still be unified.
Second, since we return only a subset of the trees displayed by $N$, there may
exist a tree-child network with fewer reticulations than $N$ that also displays
this set of trees.

\subsubsection{Real-World Data}

The real-world data we used in our experiments
was derived from a collection of gene trees for 159,905 distinct
homologous gene sets found in a set of 1,173 bacterial and archaeal genomes.
These gene trees were constructed by Beiko and are described in more detail in
\cite{Beiko_2011}.
They were also used as a test data set, for example, in the
evaluation of a method for constructing SPR supertrees
\cite{whidden2014supertrees}.
Beiko's data set (as almost every real-word data set) poses two challenges for
our algorithm.
First, bipartitions with low support in this data set were collapsed, so
the input trees are multifurcating.
Second, since not all genes are present in all taxa, the label sets of
the input trees differ.

To obtain a collection of binary trees over the same label set, we used
a two-step process:
First, given the desired number of leaves $n$ as a parameter, we selected a
subset of $n$ taxa $X$ and all trees that contain all of these taxa.
Then we restricted the selected trees to the chosen label set $X$, thereby
obtaining a collection of multifurcating trees over this set of $n$ taxa.
Second, we resolved multifurcations in these trees to obtain a collection
of binary trees.
If we had resolved multifurcations randomly, it would have been very likely
that any network displaying the constructed trees contains many reticulations
that result only from inconsistent resolutions of the input trees.
To avoid this, we introduced inconsistent resolutions into different input
trees only if the input trees forced us to do so.
This procedure is described in more detail below and at
\verb|https://github.com/nzeh/tree_child_code|.

We did not evaluate whether the resulting trees are biologically plausible
(beyond the degree to which every binary resolution of a well supported
multifurcating tree is plausible).
Our only goal was to construct a test data set whose characteristics,
in terms of number of reticulations and existence of clusters that allow
the input to be decomposed into easier inputs, resemble those of typical
real-world inputs, in order to evaluate the usefulness of our algorithm
to construct phylogenetic networks for non-trivial real-world inputs.

\paragraph{Selection of leaf set and trees.}

To extract as many trees with a given number of common leaves $n$, we used the
following strategy:
We started with an empty set of leaves $X = \emptyset$ and the entire set of
159,905 input trees $\T$.
Then we repeated the following process $n$ times:
Let $Y$ be the set of all unique taxa of the trees in $\T$ and let $x \in Y
\setminus X$ be a taxon that occurs in the maximum number of trees in $\T$.
Then we added $x$ to $X$ and discarded all trees from $\T$ that did not contain
$x$.
At the end of this iterative process, we obtained a set of trees $\T$ that
contained all taxa in $X$.
As already mentioned, the next step was to restrict every tree in $\T$ to the
label set $X$.

\paragraph{Binary resolution.}

Binary resolutions were obtained by repeating the following process until
all trees were binary:
Inspect the trees in $\T$ in an arbitrary order.
For each tree, inspect its multifurcations in an arbitrary order.
For each multifuraction $u$, consider all pairs $\{v, w\}$ such that $v$ and
$w$ are children of $u$.
For each such pair, count the number of resolved triples (triplets of the
form $ab|c$ as opposed to $a|b|c$) that would be introduced by resolving
$\{v, w\}$ (that is, by making $v$ and $w$ children of a new node $u'$ and
making $u'$ a child of $u$) and which are also present in at least one other
tree in $\T$.

If there exists such a pair $\{v, w\}$ with at least one introduced resolved
triplet that exists also in some other tree in $\T$, then resolve the pair
$\{v, w\}$ that maximizes the number of introduced resolved triplets that
exist on other trees.
If no such pair is found, then move on to the next multifurcation in the
current tree or to the next tree if there are no more multifurcations left
to inspect in the current tree.

If the above steps resolve at least one multifurcation, then start another
iteration.
Otherwise, pick an arbitrary multifurcation in one of the trees and a random
pair of children of this multifurcation and resolve it.
Then start another iteration.
(This random resolution will be matched by all other trees in the next
iteration, thus forcing consistency between the trees.)

\paragraph{Test instances.}

By running the above procedure with parameter $n \in \{10, 20, 30, 40, 50, 60,
  80, 100, 150\}$, we generated tree sets with this number of leaves and with
between 21 and 1,684 trees for $n = 150$ and $n = 20$, respectively.
To obtain an input with a given number of leaves $n$ and a given number
of trees $t$, we selected $t$ of the trees with $n$ leaves uniformly at random.

\subsection{Parameter Tuning}

\label{sec:parameters}

Our implementation of procedure~\procref{alg:tree-child-sequence-refined}
accepts a number of command-line arguments, mainly to facilitate the type of
performance evaluation we conducted.
The most important options are turning cluster reduction on and off,
turning redundant branch elimination on or off, configuring the number of
threads across which to distribute the algorithm's work, and controlling how
frequently busy threads check for work requests from idle threads.
More threads allow the operating system to help with load balancing but too
many threads result in scheduling overhead.
Similarly, frequent checks for work requests from idle threads help with load
balancing by ensuring that idle threads never remain idle for too long but
increase the overhead that slows down busy threads.

In preliminary experiments, we determined that we obtained the best performance
using 8 threads (\verb|-p 8|) on our system.
The frequency of checks for work requests had negigible impact on the
algorithm's performance as long as idle threads do not wait for work for
too long.
Throughout the experiments discussed here, we made a busy thread check for
work requests from idle threads every 100 iterations through its main
loop (\verb|-w 100|).
Cluster reduction never hurt performance but helped substantially on most
real-world inputs, so we never turned it off.
Since redundant branch elimination is a potentially important optimization
of our algorithm discussed in Section~\ref{sec:redundant-branch-elimination},
we dedicate a separate section to discussing its impact on the algorihm's
performance.

\subsection{Results}

\label{sec:results}

\subsubsection{Does Redundant Branch Elimination Help?}

Our first experiments concerned whether redundant branch elimination helps
to reduce the running time of the algorithm in practice.
To evaluate this, we ran the algorithm with redundant branch elimination
on a synthetic data set.
For the runs with redundant branch elimination, we used
three test inputs for every possible combination of the following parameters:
\begin{itemize}[noitemsep]
\item Number of trees: $t \in \{2,5,10,15,20,50,100\}$
\item Number of reticulations: $k \in \{2,3,4,5,6,7,8,9,10,11,12\}$
\item Number of leaves: $n \in \{20,50,100,150,200\}$
\end{itemize}
resulting in a set of 1,155 inputs.
The algorithm was able to solve 1,016 of these inputs within the 1-hour time
limit.
Without redundant branch elimination, the algorithm was not able to solve
any synthetic inputs with $k > 8$ within the time limit.
Of the 735 inputs with $k \le 8$, it was able to solve 658 inputs within
the time limit.

Figure~\ref{fig:BranchReduction} shows the speed-up achieved by using redundant
branch elimination on the 658 inputs the algorithm was able to solve without
it.  
As can be seen, the effect of redundant branch elimination increases with
increasing reticulation number and, correspondingly, with increasing running
time of the algorithm, reaching a speed-up of up to 1,000 on some instances
with 6 and 7 reticulations.

Figure~\ref{fig:Raw+-BR} shows that redundant branch elimination increases the
difficulty of inputs our algorithm can solve within the 1-hour time limit.
Without branch reduction, the algorithm was able to solve all instances with
reticulation numbers up to 6 and some instances with up to 8 reticulations.
With redundant branch reduction, the algorithm was able to solve all instances
with reticulation numbers up to 8 and some instances with up to 11
reticulations.

\begin{figure}[p]
  \includegraphics[width=.4\textwidth]{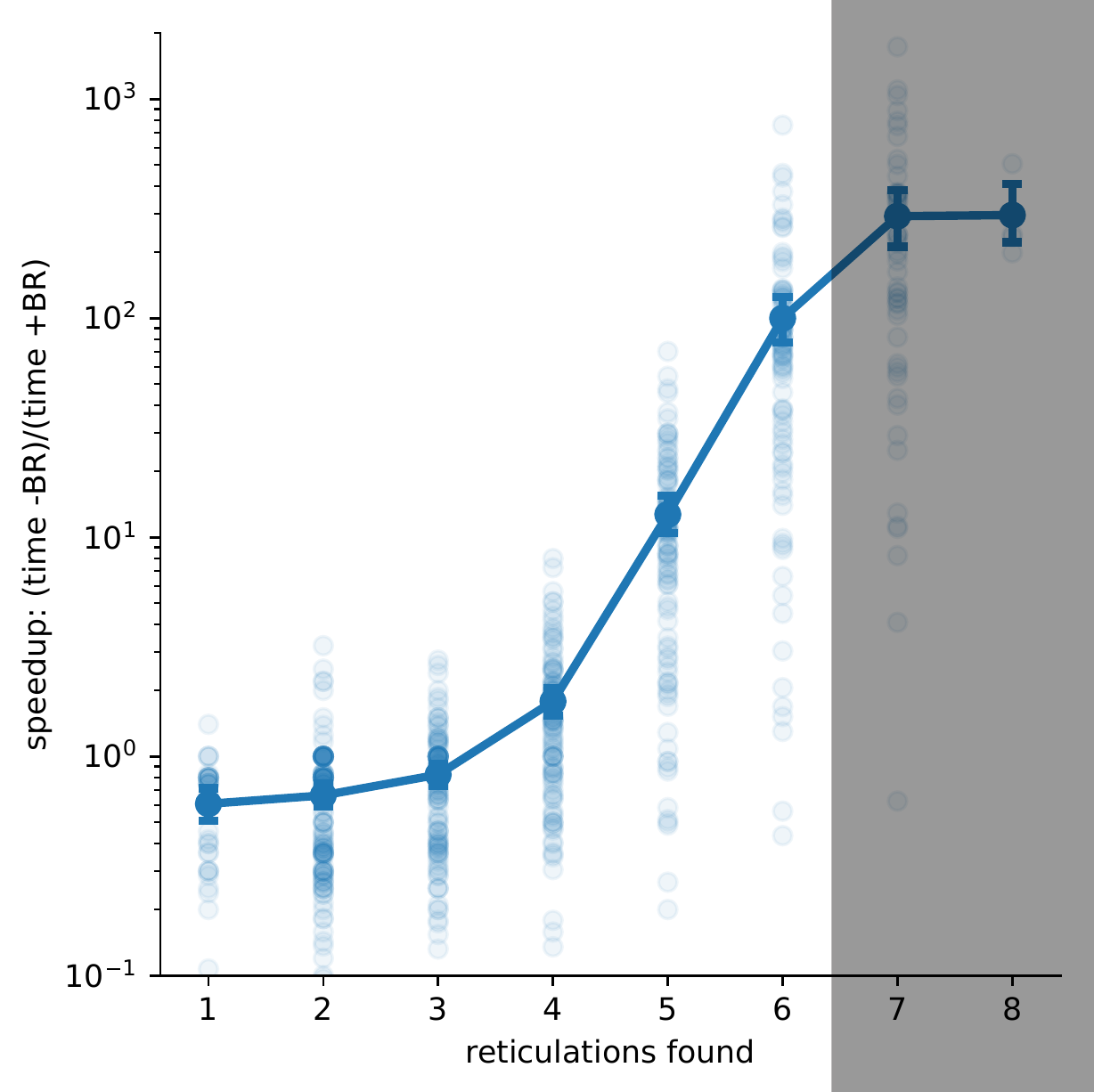}%
  \hspace{\stretch{1}}%
  \includegraphics[width=.56\textwidth]{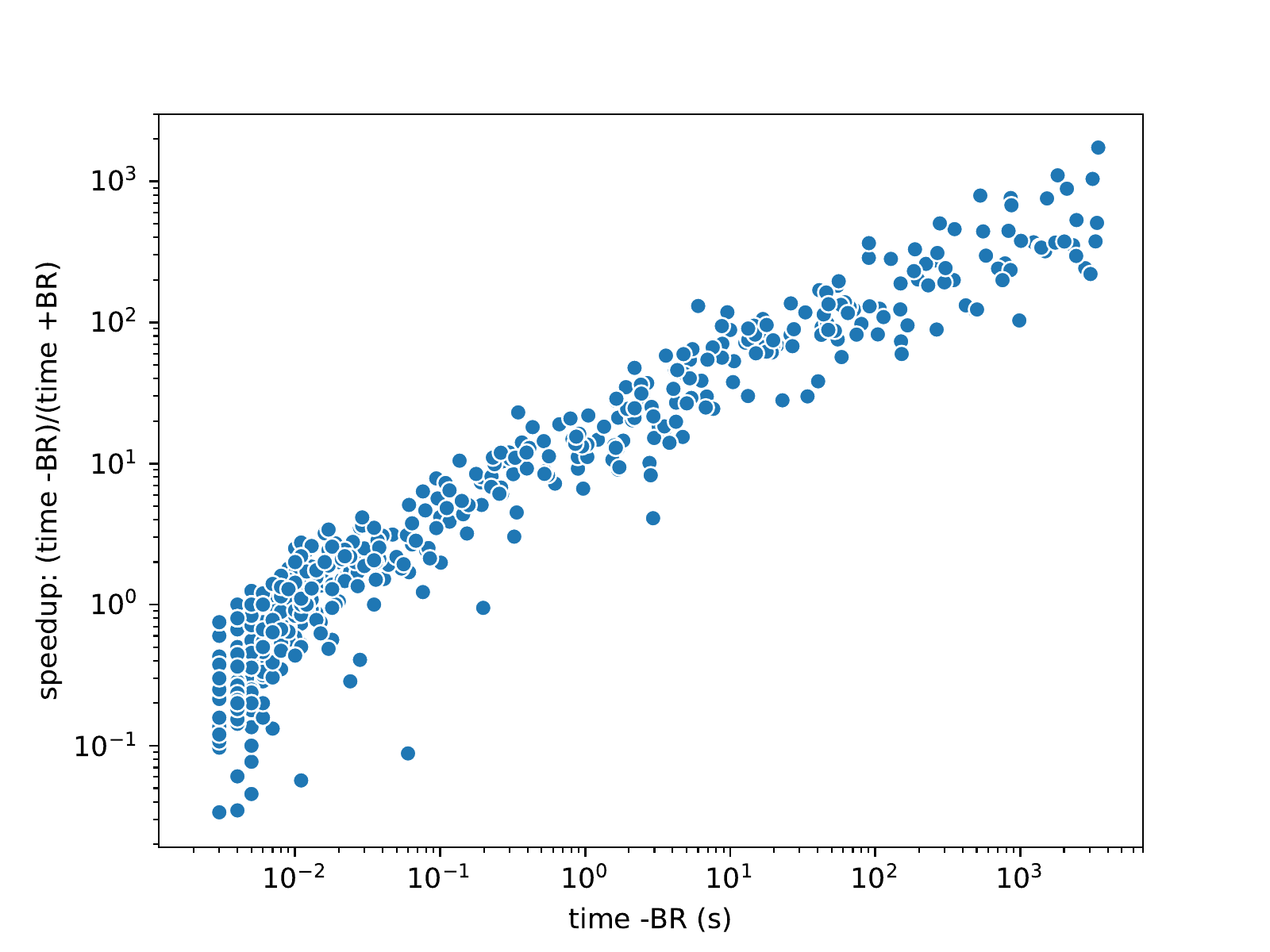}\hspace{-20pt}
  \caption{The speed-up (running time without redundant branch elimination
    divided by the running time with redundant branch elimination) achieved by
    redundant branch elimination on 658 instances solvable with and without
    redundant branch elimination. (a) as a function of the number of
    reticulations and (b) as a function of the running time without redundant
    branch elimination.
    The shading of reticulation numbers 7 and 8 indicate that not all inputs
    with 7 or 8 reticulations were solved by the algorithm, so particularly
    the flattening of the curve may be the result of limiting the running
    time of the algorithm and testing only a restricted set of inputs.
    We would expect that the effect of redundant branch elimination keeps
    increasing as the number of reticulations increases, given that the seems
    to be no plateauing of the speed-up as a function of running time in
    Figure~(b).}\label{fig:BranchReduction}
\end{figure}
\begin{figure}[p]
  \centering
  \includegraphics[width=.48\textwidth]{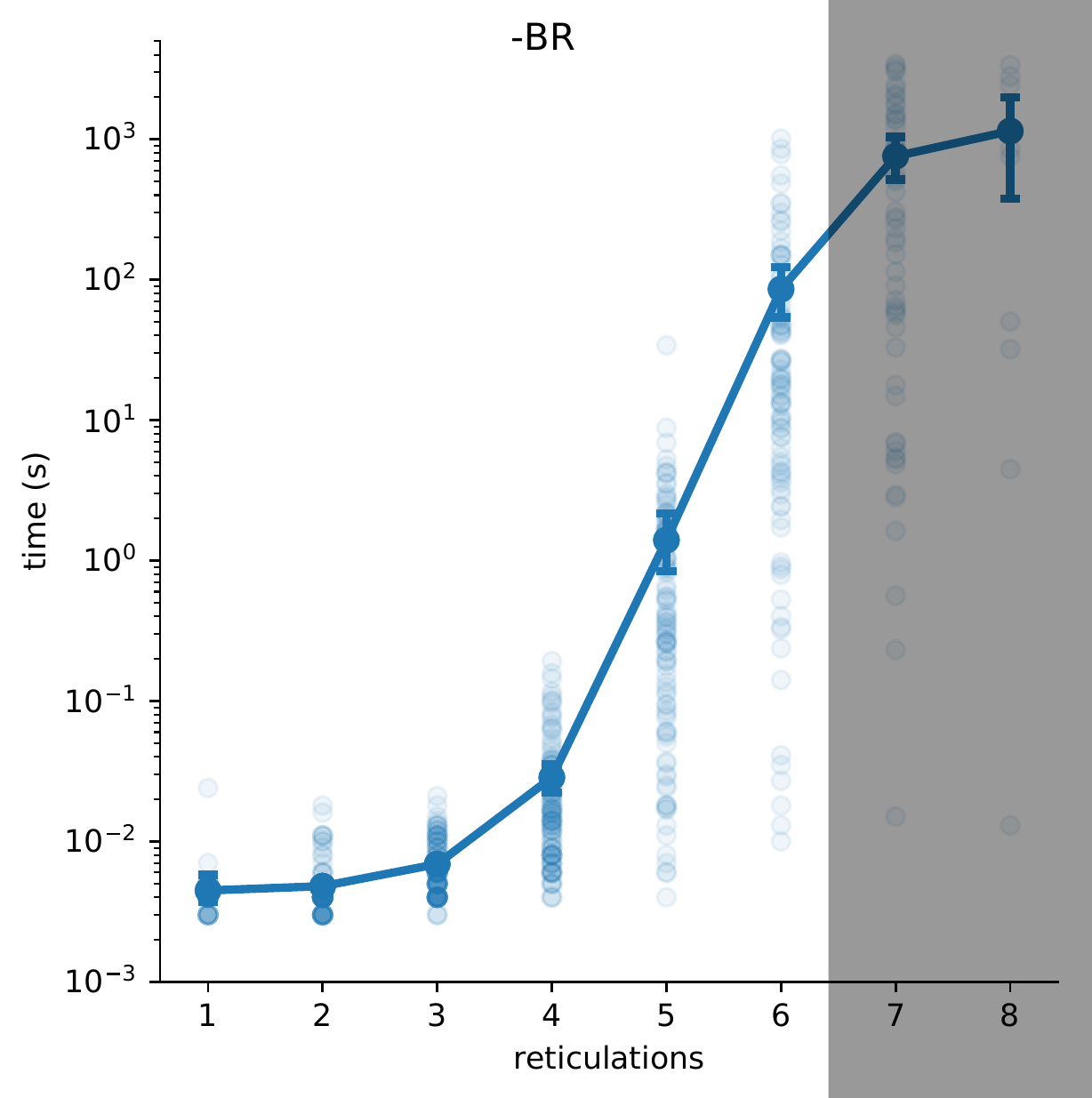}
  \includegraphics[width=.48\textwidth]{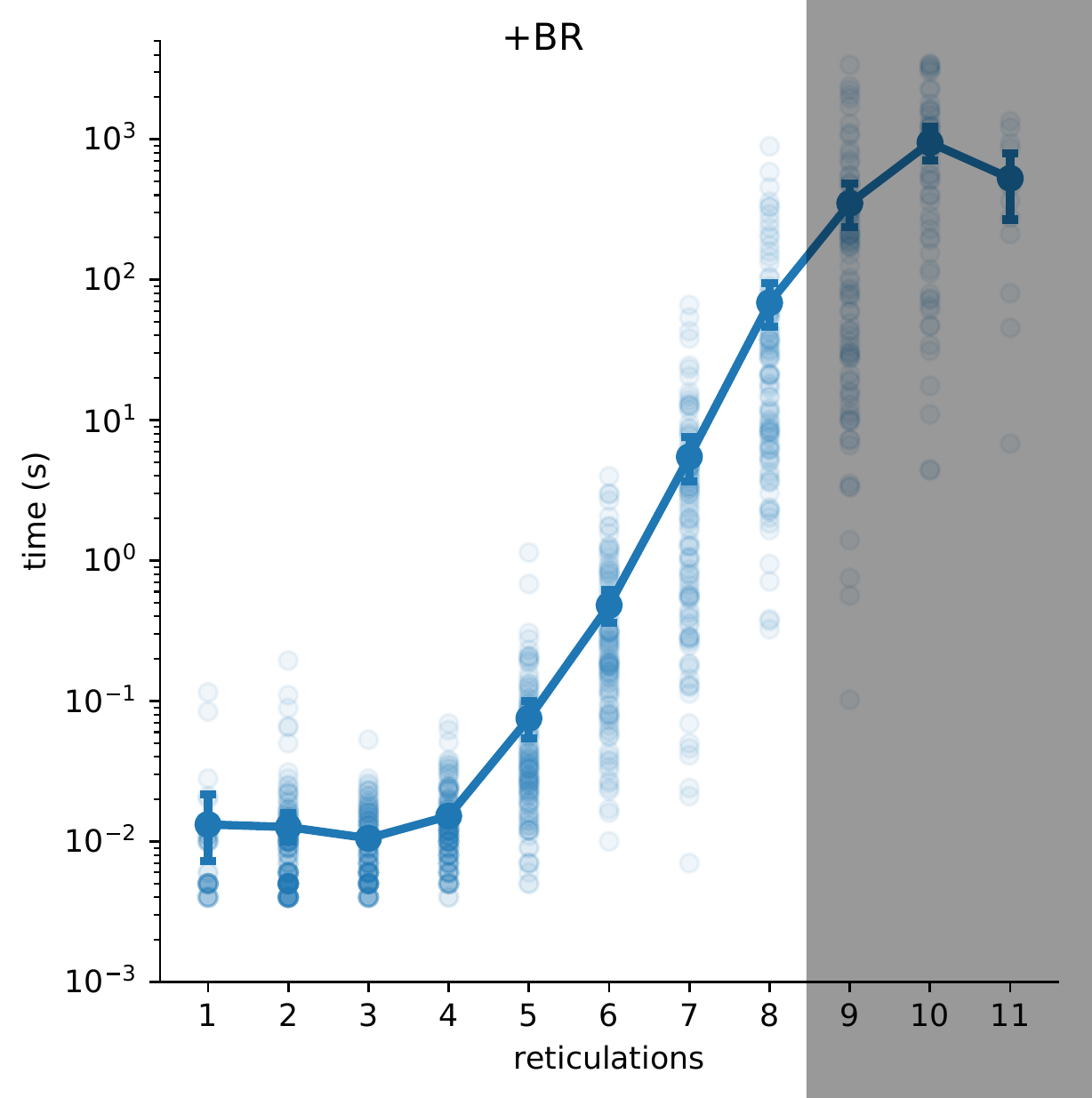}
  \caption{Running times of our algorithm with and without redundant branch
    elimination, as functions of the number of reticulations.
    As in Figure~\ref{fig:BranchReduction}, the shaded regions indicate
    reticulation numbers for which not all input instances were solved
    within the 1-hour time limit.
    Transparent dots are data points, opaque dots indicate the average together
    with the 95\% confidence intervals.}\label{fig:Raw+-BR}
\end{figure}

\subsubsection{Real-World Inputs That Can Be Solved}

\label{sec:real-data}

Our next experiment tested whether we can solve real-world instances with
non-trivial numbers of reticulations efficiently using our algorithm.
For this experiment, we extracted 10 test instances from the real-world
data set for every possible combination of the following parameters:

\begin{itemize}[noitemsep]
  \item Number of trees: $t \in \{2, 3, 4, 5, 6, 7, 8\}$
  \item Number of leaves: $n \in \{10, 20, 30, 40, 50, 60, 80, 100, 150\}$
\end{itemize}

\begin{figure}[t]
  \includegraphics[width=.58\textwidth]{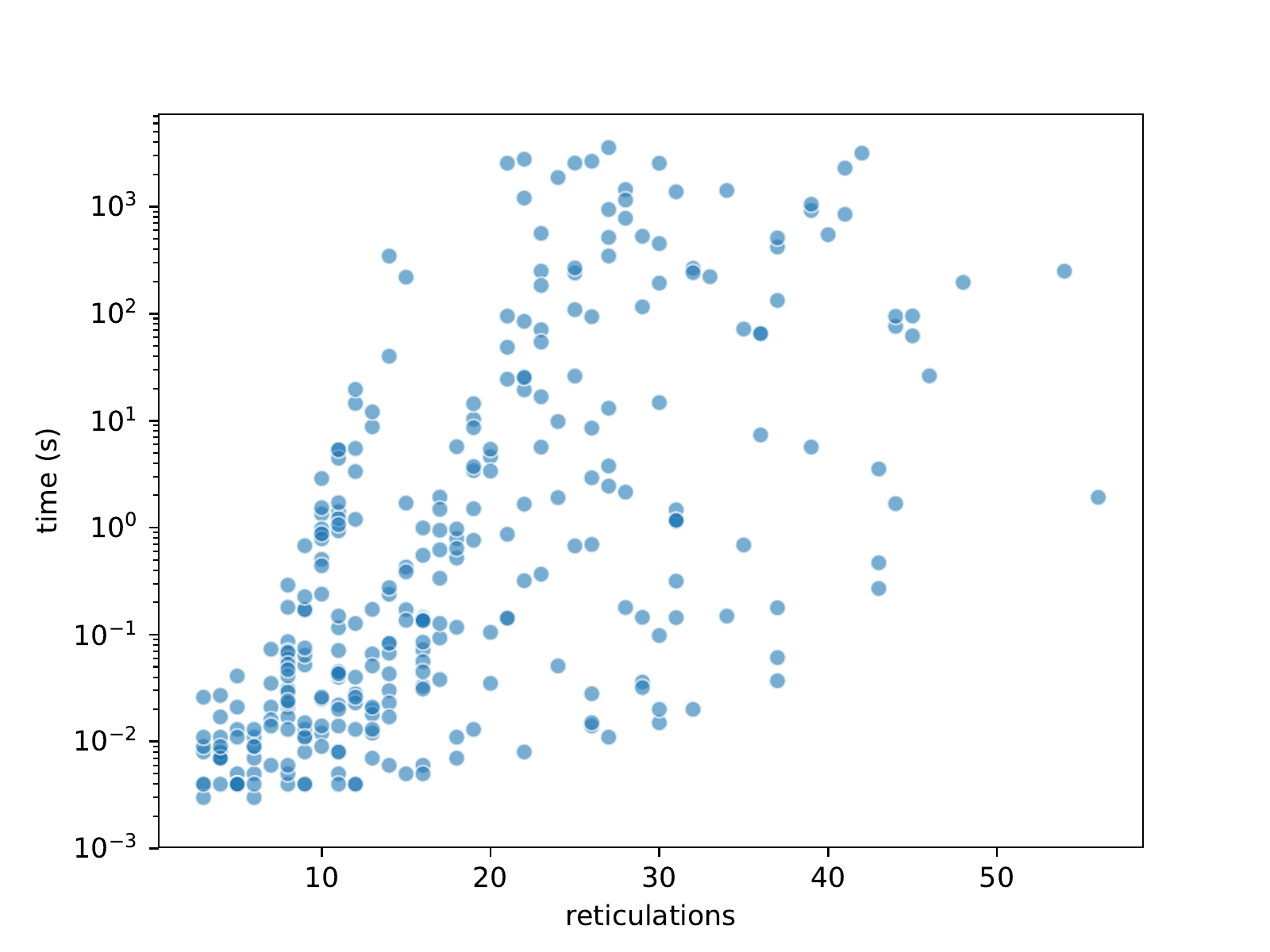}%
  \hspace*{\stretch{1}}%
  \includegraphics[width=.4\textwidth]{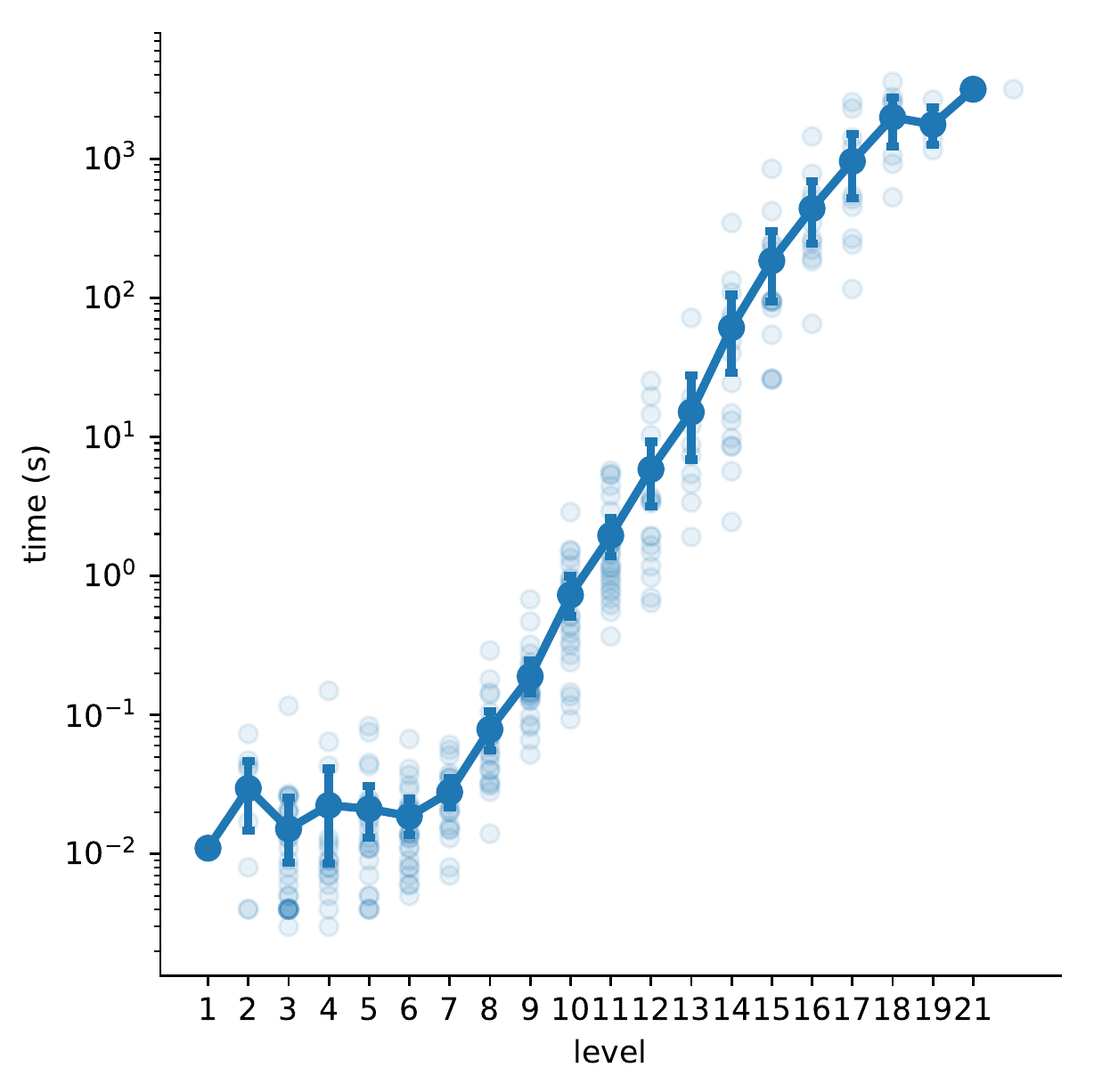}
  \caption{Running times of our algorithm on real-world data as a function
    of the reticulation number (left) or the level
    (right).}\label{fig:RealData}
\end{figure}
\begin{figure}[t]
  \includegraphics[width=.45\textwidth]{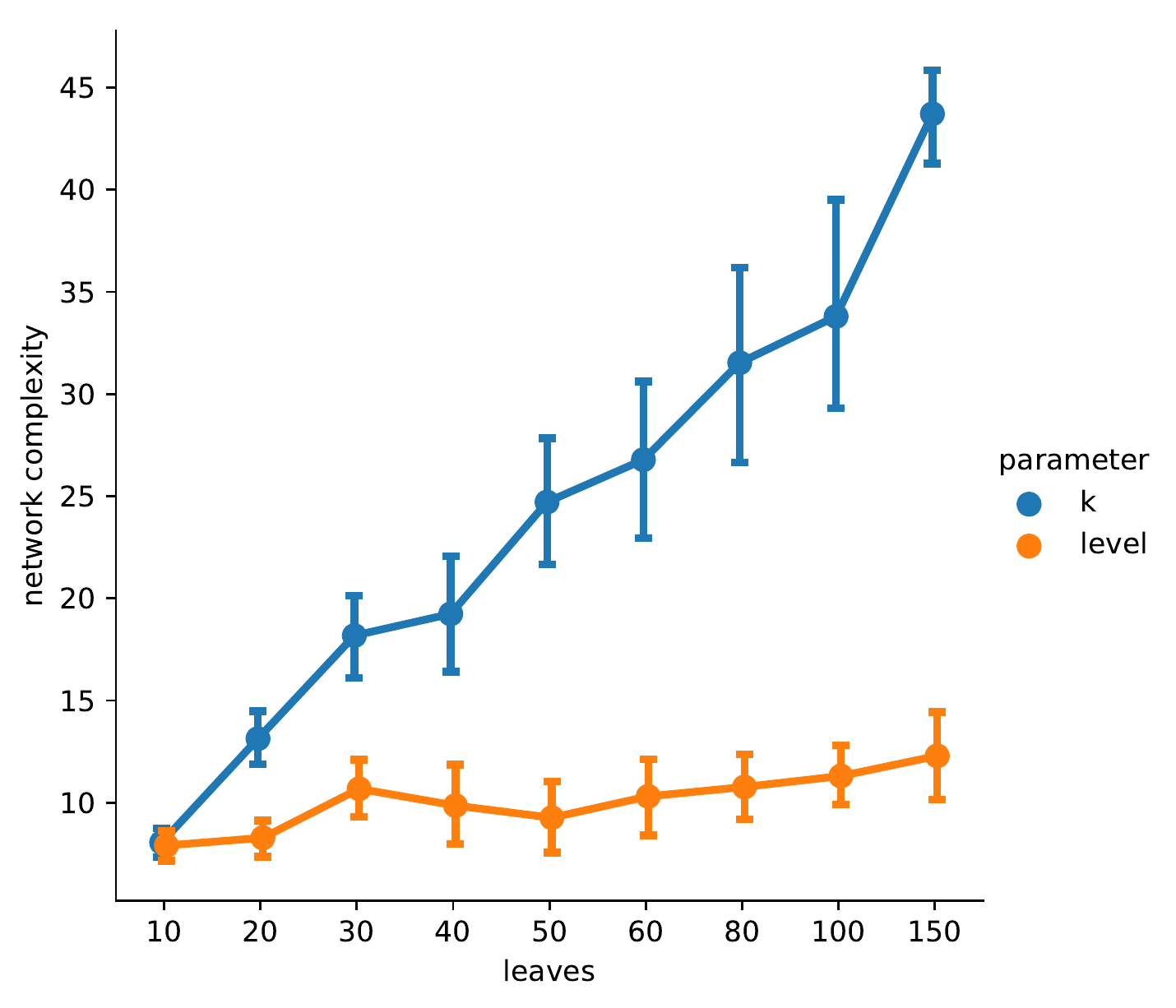}%
  \hspace{\stretch{1}}%
  \includegraphics[width=.45\textwidth]{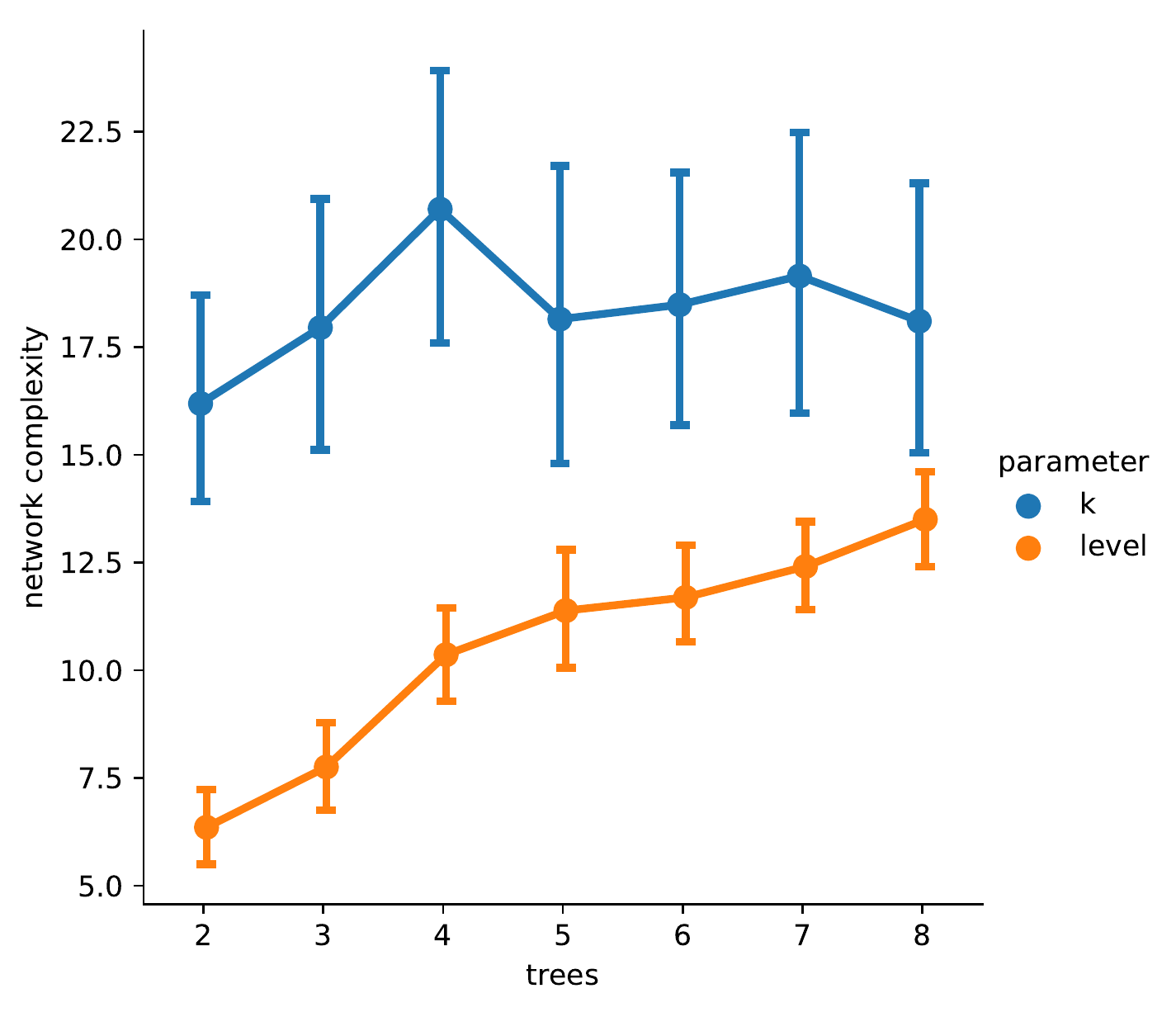}
  \caption{The reticulation number and the level as a function of the number of
    leaves and trees in the real-world inputs.}\label{fig:RealDataTreesLeaves}
\end{figure}

The algorithm was run with redundant branch elimination but with cluster
reduction.
Of the 630 test inputs, our algorithm was able to solve 306 within the 1-hour time limit.
The left graph in Figure~\ref{fig:RealData} shows the running time of our
algorithm on the instances it was able to solve as a function of the number
of reticulations.
We make two important observations:
First, even though our algorithm was not able to solve any synthetic inputs
with more than 11 reticulation even with redundant branch elimination turned
on, it was able to solve real-world inputs with up to 50 reticulations.
Second, the running time varies greatly across instances with the same number
of reticulations.
Both observations can be explained by the fact that the real-world data has
much more structure and can be decomposed into non-trivial clusters.
Figure~\ref{fig:RealDataTreesLeaves} shows the number of reticulations and
the level of the real-world inputs as a function of the number of trees.
These figures demonstrate that the network levels are significantly lower than
the number of reticulations, something that was observed for inputs consisting
of two trees and which is the key to the fast running times of MAF-based
algorithms for pairs of trees.
It comes a bit of a surprise that the same is true also for more than two
trees.
However, the right graph in Figure~\ref{fig:RealDataTreesLeaves} demonstrates
that the gap between level and reticulation number narrows as the number of
trees increases.

Using cluster reduction, the running time of the algorithm is determined by the
level of the computed network rather than the reticulation number.
Thus, the right graph in Figure~\ref{fig:RealData} shows the running time as a
function of the level of the computed network.
This figure highlights another important fact:
We were able to solve real-world instances with level up to 21 whereas level
11 was the limit for synthetic inputs.
This suggests that even the clusters seem have significantly more structure
than random instances, which allows the algorithm to branch on fewer
non-trivial cherries in each recursive call than on synthetic instances.

\subsubsection{Dependence of the Running Time on the Number of Trees and
  Number of Leaves}

The theoretical analysis of our algorithm predicts an exponential dependence
of its running time only on the number of reticulations $k$, whereas the
running time should only depend nearly linearly on both $n$ and $t$.
To verify this, we divided the observed running times, for each value of
$k$ between $1$ and $8$, by $n$ and then by $t$.
Figures~\ref{fig:NorbertTreesLeavesCorrected} shows the results.
The negative slopes of these curves confirms that the running time in
practice depends at most linearly on each of $n$ and $t$.

\begin{figure}[t]
  \includegraphics[width=.48\textwidth]{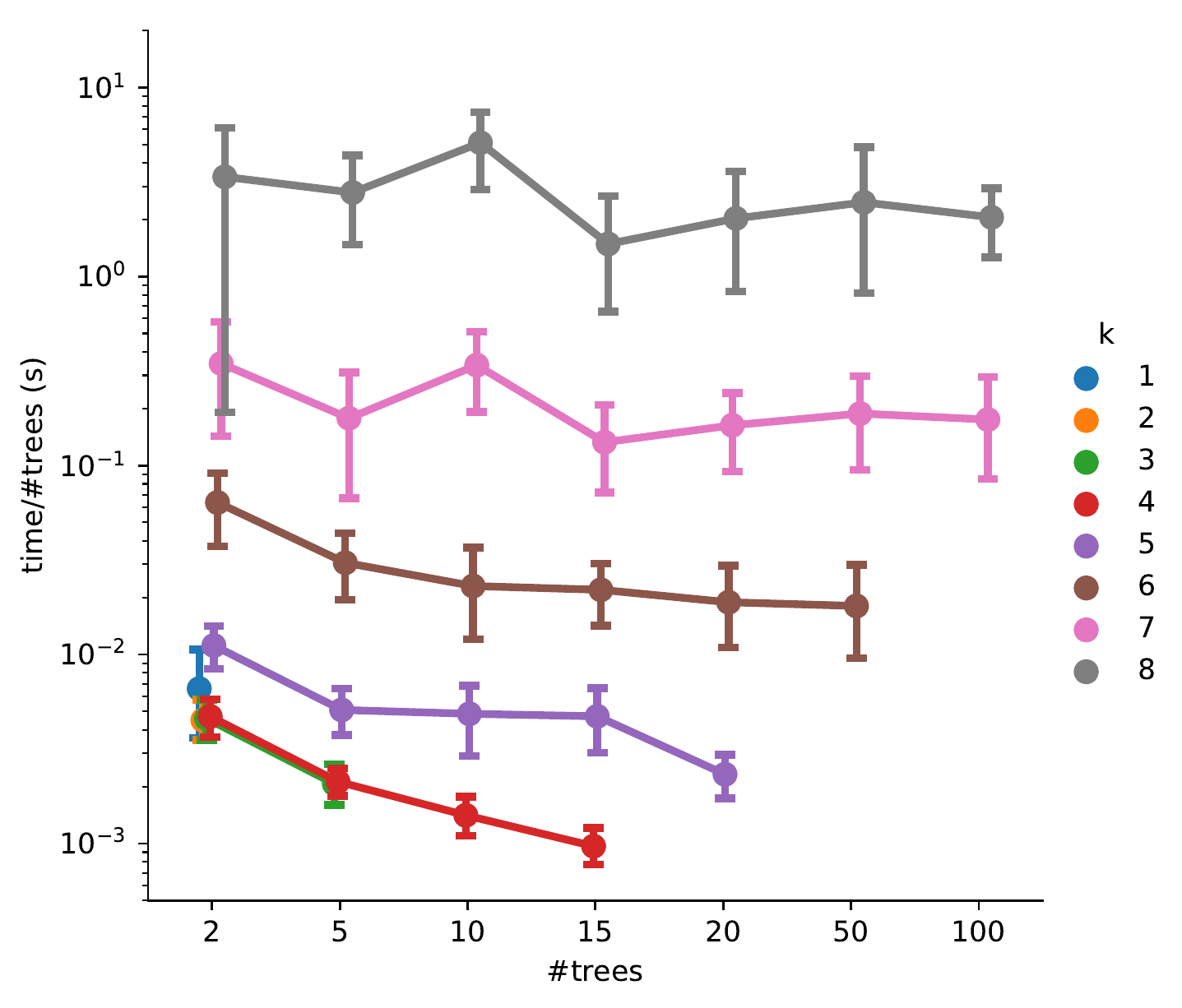}%
  \hspace{\stretch{1}}%
  \includegraphics[width=.48\textwidth]{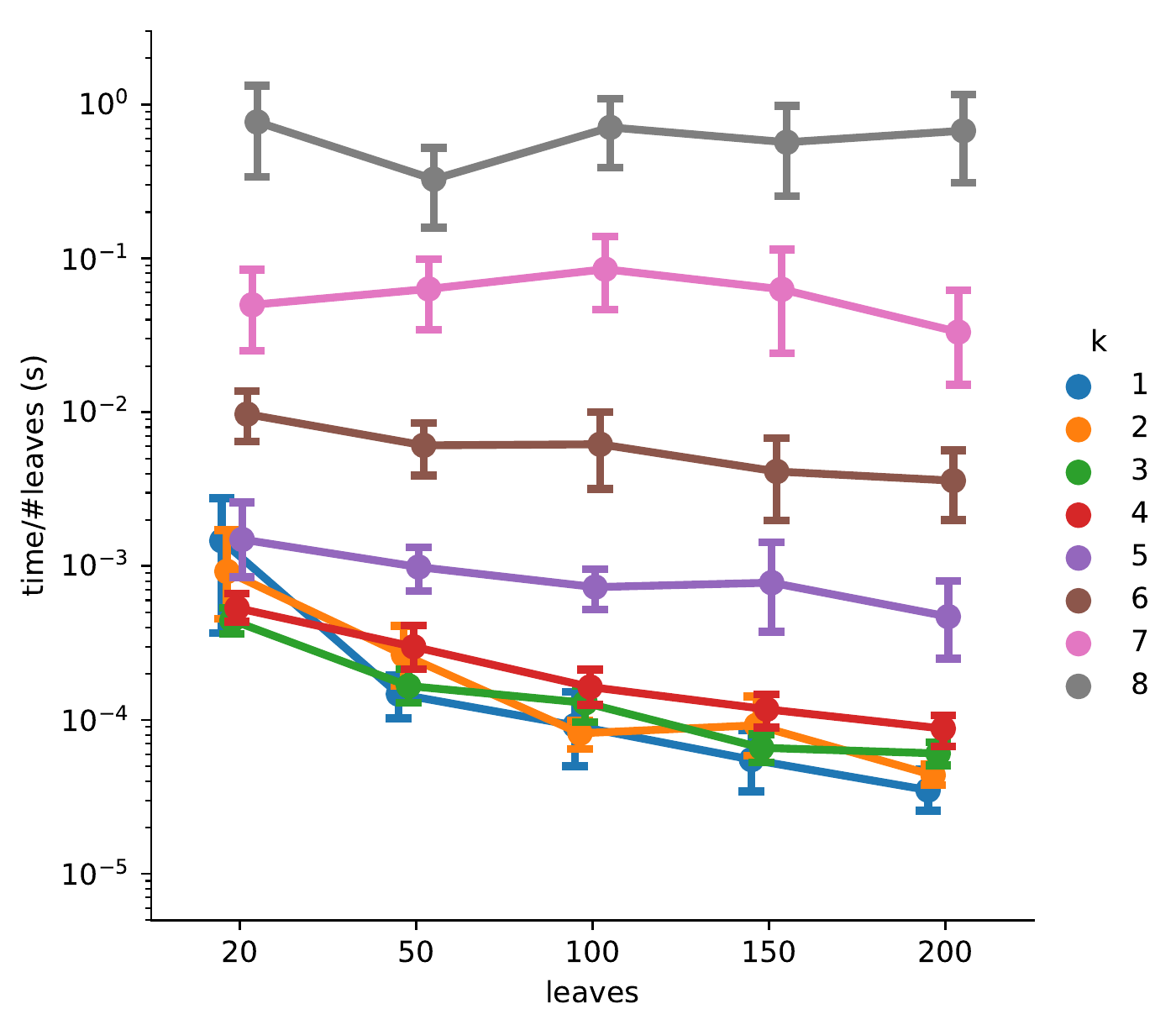}
  \caption{Running times of the algorithm with redundant branch elimination
    on al synthetic test inputs divided by the number of trees (left) and
    the number of leaves (right).
    Error bars denote a 95\% confidence
    interval.}\label{fig:NorbertTreesLeavesCorrected}
\end{figure}

\subsubsection{Comparison with \textsc{HybroScale}}

The most interesting question is whether optimal tree-child networks can
be computed significantly faster than unrestricted hybridization networks.
To answer this question, we compared the running time of our algorithm
against that of its closest competitor \textsc{HybroScale}, which computes
unrestricted hybridization networks.
For this comparison, we used synthetic data and real-world data.
In order to test a wide range of test inputs, we
limited the time per run to 20 minutes for synthetic inputs and to 60 minutes
for real-world inputs.
Since we ran our algorithm with 8 threads, we did the same for
\textsc{HybroScale}.

\paragraph{Synthetic data.}

We tested both our algorithm and \textsc{HybroScale} on 6 test inputs
for every possible combination of the following parameters:
\begin{itemize}[noitemsep]
\item Number of trees: $t \in \{3,5,10,20\}$
\item number of reticulations: $k \in \{1,2,3,4,5,6,7,8,9,10,11,12\}$
\end{itemize}
and on six inputs with 2 trees and $k \in \{2,4,\ldots,28, 30\}$. 
All instances had 20 leaves.
We use a wider range of reticulation numbers (and compensate for this by
using only three instances for each value of $k$) for inputs with only two
trees because we expected \textsc{HybroScale} to run very fast on such inputs
(because MAF-based algorithms are very fast for pairs of trees).

As can be seen in Figure~\ref{fig:HybroVsNorbertTrees} 
and as expected,
\textsc{HybroScale} outperforms our algorithm on inputs consisting of two trees
and for more than 7 reticulations.
For more than two trees, our algorithm runs faster than \textsc{HybroScale} due
to the near-linear dependence of our algorithm on the number of trees and the
exponential dependence of \textsc{HybroScale} on the number of trees.
The difference becomes very pronounced for 10 and 20 trees, where
\textsc{HybroScale} was unable to solve most instances whereas our algorithm
solved all test instances within the 20-minute time limit.
Additionally, Hybroscale ran out of memory on certain occasions. 

\begin{figure}[t]
  \centering
  \includegraphics[width=\textwidth]{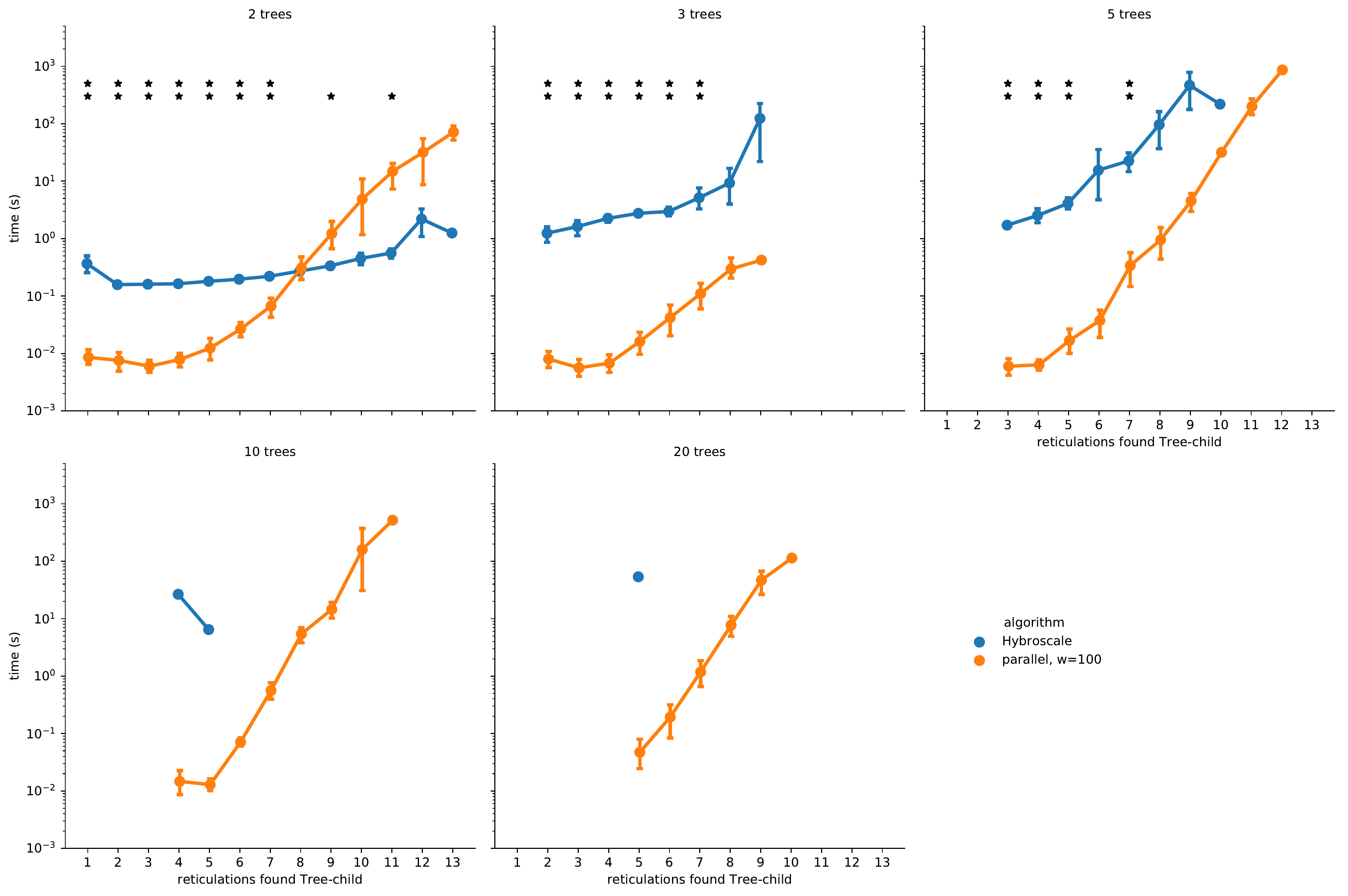}
  \caption{Running times of our algorithm and \textsc{HybroScale} on
    synthetic inputs.
    Since our algorithm solves all test instances and \textsc{HybroScale} does
    not, we choose the tree-chlid hybridization number as the $x$-axis.
    Bars indicate a 95\% confidence interval.
    Stars indicate significant differences between the running times of
    the two algorithms using an independent $t$-test with unequal variances (*:
    $p < 0.05$, **: $p < 0.01$).}\label{fig:HybroVsNorbertTrees}
\end{figure}

\paragraph{Real-world data.}

For this experiment, we used the same data set as in
Section~\ref{sec:real-data}.
As mentioned before, our algorithm solved 306 of the 630 inputs in the 1-hour time limit;
\textsc{HybroScale} solved 152 inputs, which were a subset of the 306
inputs solved by our algorithm.
On 5 of the 2-tree inputs, \textsc{HybroScale} outperformed our algorithm.
On all other inputs, including all other 2-tree inputs, our algorithm
was faster.
Figure~\ref{fig:HybroVsNorbertTreesRealLevel} shows the detailed results.

\begin{figure}[p]
  \centering
  \includegraphics[width=\textwidth]{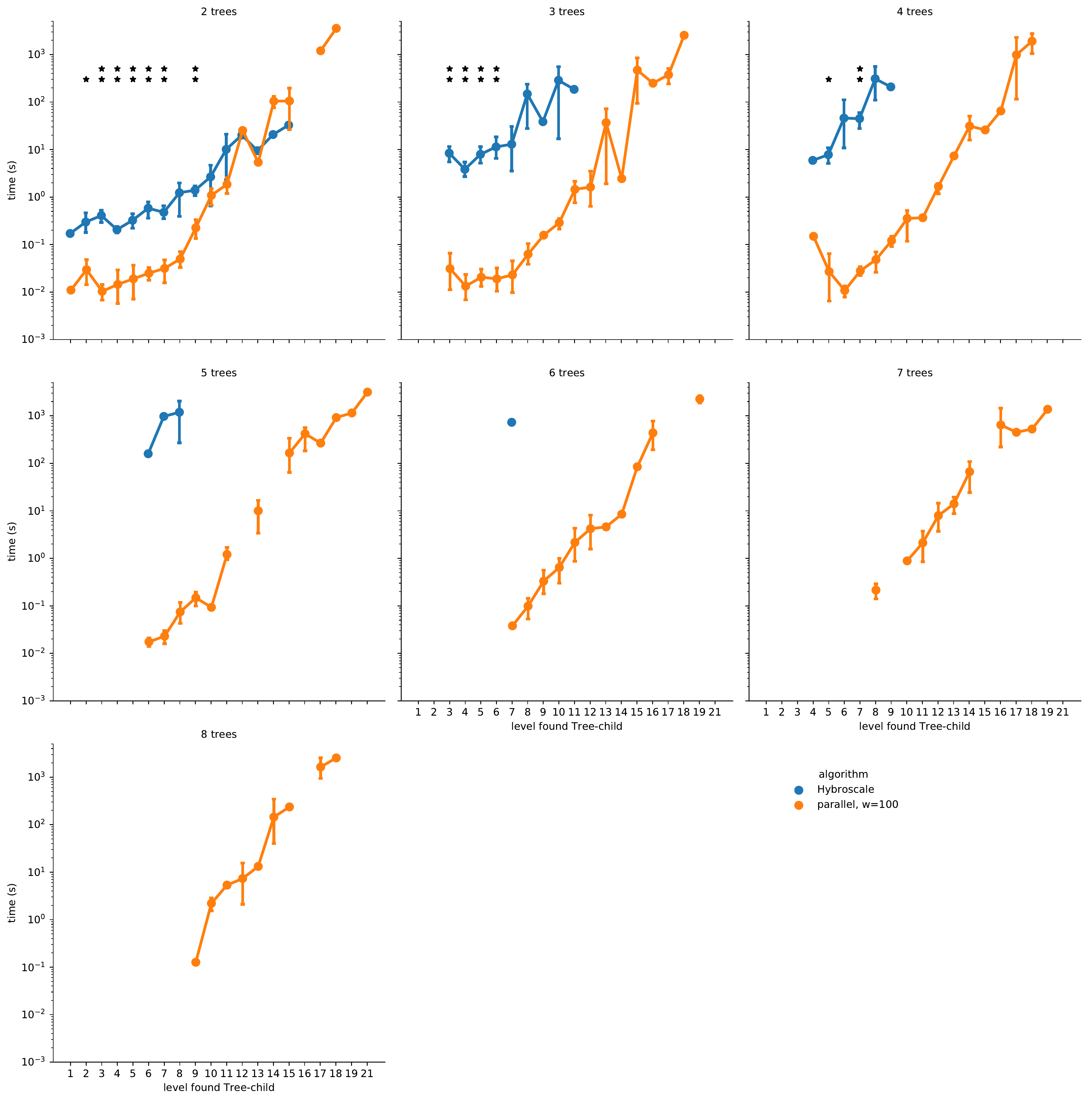}
  \caption{Running times of our algorithm and \textsc{HybroScale} on
    real-world inputs.
    Since our algorithm solves all test instances that \textsc{HybroScale} was
    able to solve, we chose the tree-child level as the $x$-axis.
    Bars indicate a 95\% confidence interval.
    Stars indicate significant differences between the running times of
    the two algorithms using an independent $t$-test with unequal variances (*:
    $p < 0.05$, **: $p < 0.01$).}\label{fig:HybroVsNorbertTreesRealLevel}
\end{figure}

\subsubsection{Hybridization versus Tree-Child Hybridization}

The final question we were interested in was whether optimal tree-child
hybridization networks have significantly more reticulations than the
optimal unrestricted hybridization networks for the same sets of trees or
whether tree-child hybridization networks are often also optimal hybridization
networks.

Of the 268 synthetic inputs that both our algorithm and
\textsc{HybroScale} were able to solve, only 3 had a greater
tree-child hybridization number than their hybridization number.
For all three inputs, the difference was 1.

Of the 142 real-world inputs solved by both our algorithm and
\textsc{HybroScale}, 21 had a greater tree-child hybridization number than
their hybridization number.
For 20 of these inputs, the difference was 1; for 1 input, the difference
was 2.

This indicates that very often, tree-child hybridization networks achieve
the optimal hybridization number and, even when they do not, they offer
a reasonable approximation of optimal hybridization networks.
Given that they are substantially easier to compute, as our results in
the previous subsection demonstrate, tree-child networks therefore offer
a useful analysis tool that can be used in place of hybridization networks
in many instances.

\section{Conclusion}

\label{sec:conclusions}

We have presented the first fixed-parameter algorithm for computing optimal
tree-child networks for many trees, based on the recently introduced
concept of tree-child cherry picking sequences.
While the theoretical running time of our algorithm is substantially greater
than MAF-based network construction methods for two trees,
our experimental results confirm that our algorithm can be used to solve
non-trivial real-world inputs efficiently.
Similarly to MAF-based algorithms for two trees, a key factor determining
whether an instance can be solved efficiently is whether it can be decomposed
into non-trivial clusters.
While it comes as no surprise that randomly generated inputs consisting of
more than two trees (almost) cannot be decomposed into clusters and thus
cannot be solved efficiently, except for fairly small numbers of reticulations,
the real-world inputs in our experiments contained sufficiently many
non-trivial clusters, which allowed us to solve some inputs with up to 50
reticulations within one hour or less.

The closest competitor of our algorithm, \textsc{Hybroscale}, which computes
unrestricted hybridization networks, outperforms our algorithm on inputs
consisting of two trees, which is to be expected because MAF-based methods
are very efficient for computing optimal hybridization networks for pairs
of trees.
Already for 3 trees, our algorithm outperforms \textsc{Hybroscale} and,
for more than 6 trees, \textsc{Hybroscale} cannot solve any of the inputs
our algorithm can solve, due to its exponential dependence on the number
of trees.

While our results are promising, they should only be considered to be a first
important step towards efficient algorithms for computing (tree-child)
hybridization networks from many input trees.
Here are two natural and important open questions to be addressed
by future work:

Can tree-child hybridization networks be computed faster than in
$O((ck)^k \cdot \textrm{poly}(n, t))$ time, ideally in $O(c^k \cdot
\textrm{poly}(n, t))$ time?
For temporal networks, a recent result \cite{sander} shows that this is indeed
the case.
An interesting open question is whether the techniques used in that algorithm
can also be used to obtain faster algorithms for computing general
tree-child networks.

Most real-world inputs are multifurcating, as a result of suppressing
\leo{branches} in gene trees with low support.
Thus, it would be of great importance to obtain efficient methods for
constructing (tree-child) hybridization networks from multifurcating trees.
Our algorithm is able to do this but only if we sacrifice the FPT bound
on its running time:
the bound on the number of non-trivial cherries in
Proposition~\ref{prop:few-non-trivial-cherries}, which is the key to bounding
the branching number of our algorithm, holds only if the input trees are
binary.
It remains an open question whether there exists a fixed-parameter algorithm
for computing optimal tree-child hybridization networks for multifurcating
trees.


\bibliographystyle{plain}
\bibliography{bibliography}

\clearpage

\appendix

\section{Construction of a Tree-Child Network from a Tree-Child Cherry Picking
  Sequence}

\begin{procedure}[h]
  \caption{TreeChildNetworkFromSequence($\T,S$)\label{alg:tree-child-network-from-sequence}}
  \KwIn{A set of $X$-trees $\T$ and a tree-child cherry picking sequence $S =
    \langle(x_1, y_1), \ldots, (x_r, y_r), (x_{r+1}, -)\rangle$ for $\T$}
  \KwOut{ A tree-child phylogenetic network $N$ on $X$ that displays $\T$ and with
    reticulation number at most $w(S)$}
  \uIf{$|X| = 1$}{
    \Return{the unique network consisting of a single node labelled with the element of $X$\;}
  }
  \Else {
    $N \gets {}$the directed graph with nodes $\rho$ and $x_{r+1}$ and a single
    edge ${\rho}x_{r+1}$\;
    \For{$j \gets r$ \textbf{downto} $1$}{
      Split the parent edge of $y_j$ in $N$ by adding a node $p$\;
      \uIf{$x_j$ is a leaf of $N$}{
        \uIf{$x_j$'s parent in $N$ is a reticulation $r$}{
          $q \gets r$\;
        }
        \Else{
          Split the parent edge of $x_j$ in $N$ by adding a node $q$\;
        }
      }
      \Else{
        Add $x_j$ to $N$\;
        $q \gets x_j$\;
      }
      Add the edge $pq$ to $N$\;
    }
    \Return{$N$\;}
  }
\end{procedure}
 
 
\end{document}